\documentclass[11pt,oneside]{book}

\usepackage[a4paper,margin=3cm,headheight=14pt]{geometry}
\usepackage{microtype}
\usepackage{setspace}
\usepackage[T1]{fontenc}
\usepackage[utf8]{inputenc}
\usepackage{csquotes}

\usepackage{xcolor}
\newcommand\RVBoxTitle{}

\usepackage{hyperref}
\hypersetup{
  colorlinks=true,
  linkcolor=blue!50!black,
  citecolor=blue!50!black,
  urlcolor=blue!50!black
}

\usepackage{amsmath,amssymb,mathtools}

\usepackage{cite}
\usepackage{booktabs}
\usepackage{xspace}
\usepackage{tikz}
\usetikzlibrary{calc}
\usetikzlibrary{arrows.meta, positioning, automata, calc}

\usepackage{titlesec}

\definecolor{ChapterAccent}{HTML}{1F4E79}
\definecolor{ChapterLight}{HTML}{D9E7F5}

\titleformat{\chapter}[display]
  {\normalfont\bfseries\Huge}
  {%
    \begin{tikzpicture}[baseline=(n.base)]
      \node[draw=ChapterAccent, line width=1.2pt,
            fill=ChapterLight, inner xsep=14pt, inner ysep=10pt] (n)
            {\color{ChapterAccent}\fontsize{34}{34}\selectfont \thechapter};
    \end{tikzpicture}%
  }
  {12pt}
  {\color{ChapterAccent}\Huge}
\titlespacing*{\chapter}{0pt}{-10pt}{18pt}

\usepackage{fancyhdr}
\usepackage[most]{tcolorbox}
\tcbuselibrary{breakable}
\usepackage{xparse}

\ExplSyntaxOn
\NewDocumentCommand{\RVTitleSuffix}{m}
  {
    \tl_if_blank:nF { #1 } { \space (#1) }
  }
\ExplSyntaxOff

\definecolor{DefBack}{HTML}{E8F2FF}
\definecolor{ExBack}{HTML}{EAF7EC}
\definecolor{ThmBack}{HTML}{FFF2E6}
\definecolor{NoteBack}{HTML}{F3F0FF}
\definecolor{PropBack}{HTML}{FFF2E6}
\definecolor{LemmaBack}{HTML}{F6F6F6}
\definecolor{ExerBack}{HTML}{F2F2F2}

\tcbset{%
  rvtitle/.store in=\RVBoxTitle,%
  rvtitle={},%
  rvbox/.style={%
    enhanced,%
    breakable,%
    colframe=ChapterAccent,%
    boxrule=0.6pt,%
    arc=0pt,%
    left=10pt,right=10pt,top=4pt,bottom=8pt,%
    before skip=15pt, after skip=15pt,%
    before upper={%
      \setlength{\parskip}{6pt}
      \ifx\RVBoxTitle\empty\else%
        {\sffamily\bfseries\normalsize\color{ChapterAccent}\RVBoxTitle}\par\smallskip%
        \hrule height 0.4pt\relax\vspace{6pt}%
      \fi%
    },%
  },%
  rvboxnote/.style={%
  rvbox,%
  before upper={%
    \setlength{\parskip}{6pt}%
    \ifx\RVBoxTitle\empty\else%
      {\sffamily\bfseries\normalsize\color{ChapterAccent}\RVBoxTitle}\par\smallskip%
      \hrule height 0.4pt\relax\vspace{6pt}
    \fi%
  },%
},%
}%

\newcounter{rvbox}[chapter]
\renewcommand{\thervbox}{\thechapter.\arabic{rvbox}}

\newenvironment{proof}[1][Proof]{\par\noindent\textbf{#1.}\hspace{0.5em}}{\hfill$\square$\par}

\newenvironment{solution}[1][Solution]{\par\noindent\textbf{#1.}\hspace{0.5em}}{\hfill$\square$\par}

\NewDocumentEnvironment{definition}{O{}}
{
  \refstepcounter{rvbox}%
  \begin{tcolorbox}[rvbox, colback=DefBack,
    rvtitle={Definition~\thervbox\RVTitleSuffix{#1}}]%
}
{\end{tcolorbox}}

\NewDocumentEnvironment{example}{O{}}
{
  \refstepcounter{rvbox}%
  \begin{tcolorbox}[rvbox, colback=ExBack,
    rvtitle={Example~\thervbox\RVTitleSuffix{#1}}]%
}
{\end{tcolorbox}}

\NewDocumentEnvironment{theorem}{O{}}
{
  \refstepcounter{rvbox}%
  \begin{tcolorbox}[rvbox, colback=ThmBack,
    rvtitle={Theorem~\thervbox\RVTitleSuffix{#1}}]%
}
{\end{tcolorbox}}

\NewDocumentEnvironment{proposition}{O{}}
{
  \refstepcounter{rvbox}%
  \begin{tcolorbox}[rvbox, colback=PropBack,
    rvtitle={Proposition~\thervbox\RVTitleSuffix{#1}}]%
}
{\end{tcolorbox}}

\NewDocumentEnvironment{lemma}{O{}}
{
  \refstepcounter{rvbox}%
  \begin{tcolorbox}[rvbox, colback=LemmaBack,
    rvtitle={Lemma~\thervbox\RVTitleSuffix{#1}}]%
}
{\end{tcolorbox}}

\NewDocumentEnvironment{exercise}{O{}}
{
  \refstepcounter{rvbox}%
  \begin{tcolorbox}[rvbox, colback=ExerBack,
    rvtitle={Exercise~\thervbox\RVTitleSuffix{#1}}]%
}
{\end{tcolorbox}}

\NewDocumentEnvironment{notebox}{O{}}
{
  \begin{tcolorbox}[rvboxnote, colback=NoteBack,
    rvtitle={Note\RVTitleSuffix{#1}}]%
}
{\end{tcolorbox}}

\usepackage{enumitem}
\setlist{itemsep=2pt,topsep=8pt}

\usepackage{listings}
\usepackage{stmaryrd} 

\newcommand{\sem}[1]{\llbracket #1 \rrbracket}

\newcommand{\yes}{\checkmark}
\newcommand{\no}{\ensuremath{\times}}

\newcommand{\obs}{\mathit{obs}}
\newcommand{\trace}{\mathit{trace}}

\newcommand{\Knows}{\mathsf{K}}
\newcommand{\KnowsWhether}{\mathsf{W}}
\newcommand{\Ag}{\mathit{Ag}}
\newcommand{\AP}{\mathit{AP}}
\renewcommand{\phi}{\varphi}

\newcommand{\Pref}{\mathit{Pref}^+}

\newcommand{\last}{\mathit{last}}

\newcommand{\LTL}{LTL$^\sim$\xspace}
\newcommand{\TLTL}{TLTL$^\sim$\xspace}
\newcommand{\TLTLFP}{{TLTL}\xspace}

\newcommand{\pureLTL}{LTL\xspace}

\newcommand{\formset}[1]{\bigl\{\;#1 \;\bigr\}}
\newcommand{\tmp}[1]{\mathsf{#1}}
\newcommand{\Trans}{\Delta}
\newcommand{\Npos}{\mathbb{N}_{\ge 1}}

\newcommand{\sUntil}{{\tmp{U}}^+}
\newcommand{\sSince}{{\tmp{S}}^+}

\newcommand{\sFuture}{{\tmp{F}}^+}
\newcommand{\sPast}{{\tmp{P}}^+}

\newcommand{\tverdict}{\mathsf{true}}
\newcommand{\fverdict}{\mathsf{false}}
\newcommand{\unknown}{{?}}

\newcommand{\dbound}{D}

\newcommand{\outmap}{\lambda}

\title{\textbf{Runtime Verification}\\
\Large Lecture Notes}
\author{Benedikt Bollig}
\date{\today}

\begin{document}

\frontmatter

\begin{titlepage}
\begin{tikzpicture}[remember picture,overlay]
  \fill[white] (current page.south west) rectangle (current page.north east);
  
  \node[anchor=center] at (current page.center) {
    \begin{minipage}{\textwidth}
      \centering
      
{\fontsize{36}{50}\selectfont\bfseries\color{ChapterAccent}
Runtime Verification}

\vspace{0.8cm}

{\fontsize{22}{26}\selectfont\color{ChapterAccent}
Monitoring, Knowledge, and Uncertainty}

\vspace{2cm}

{\fontsize{16}{20}\selectfont\color{black!60}
Lecture Notes}\\

\vspace{2cm}

{\fontsize{18}{22}\selectfont\color{black!75}
Benedikt Bollig}

\vspace{0.5cm}

Université Paris-Saclay, CNRS, ENS Paris-Saclay, LMF\\Gif-sur-Yvette, France

      
    \end{minipage}
  };
    
\end{tikzpicture}
\end{titlepage}

\clearpage

\begin{center}
{\Large\bfseries Abstract}
\end{center}

Runtime verification is a lightweight verification technique that complements
model checking by analyzing system executions at runtime rather than exploring a
complete system model in advance. It is particularly useful for partially
observable or black-box systems, where uncertainty can only be resolved through
observation. These lecture notes present automata-theoretic, temporal-logical,
and epistemic foundations of runtime verification. They cover specification
formalisms, diagnosis, opacity, and monitorability, and explain how offline
analysis can be used to construct monitors that operate online on observed
executions. The notes also discuss timed extensions and the additional
algorithmic and semantic challenges that arise in the real-time setting.

\bigskip
\bigskip
\bigskip

\begingroup
\let\clearpage\relax
\chapter*{Preface}
\endgroup

These notes accompany an MPRI M2 course given in the lecture series
\emph{Advanced Techniques of Verification}.
They present runtime verification from the perspective of partial information,
with a unifying treatment of diagnosis, opacity, and monitoring.
Comments and corrections are welcome at
\texttt{bollig@lmf.cnrs.fr}.

\bigskip

I am grateful to Martin Leucker for sharing his runtime verification
lecture material and for helpful feedback that significantly
improved the runtime monitoring part.
I would also like to thank the MPRI M2 students of 2025/2026 in the
lecture series \emph{Advanced Techniques of Verification} for their
participation and for their thoughtful questions and remarks, which
helped refine the presentation.

%

\tableofcontents

\mainmatter

\chapter{What Is Runtime Verification?}
\label{ch:introduction}

\section{Context}

\emph{Runtime verification} studies how to observe a running system and raise an alarm
when something goes wrong, or confirm that a behavior conforms with a specification.
It can be applied even when a system model is not available.
In this sense, it complements \emph{model checking}:
instead of exploring the full state space of a given model in advance, runtime
verification analyzes executions (also called traces) as they are generated.

A monitor observes the execution step by step and produces a verdict based on the
information available so far.
If the observed behavior violates the specification, the monitor raises an alarm;
if the behavior is guaranteed to satisfy the specification, it may confirm
correctness.

A central theme of runtime verification is \emph{uncertainty}.
The system may be only partially known, or even unknown, and the monitor typically
observes only a subset of the system’s actions or propositions.
Moreover, in reactive systems, executions are typically infinite, but at runtime
only a finite prefix is available.
As a result, verdicts may be inconclusive, reflecting the inherent limits of
runtime observation.

Runtime verification shares this focus on uncertainty with closely related
techniques such as diagnosis and opacity checking.
The goal of these notes is to present a unifying treatment of these topics.
To that end, we use epistemic logic as a common language to express what an
observer knows (and does not know) from its observations.

In these notes, we actually use the term \emph{runtime verification} in a broad sense, to
refer to observation-based verification problems such as monitoring, diagnosis,
and opacity, which all reason about what can be inferred from executions at
runtime.

\section{Model Checking, Diagnosis, Opacity, and Monitoring}

We consider reactive systems that produce infinite executions.
Throughout these notes, we assume that a (possibly abstract) model
\(\mathcal S\) of the system is available.
Depending on the amount of information encoded in this model, we
distinguish three idealized cases.

\begin{itemize}
\item \emph{White-box systems}, where the model accurately reflects the
implementation.

\item \emph{Gray-box systems}, where the model intentionally abstracts or
hides information, for instance by introducing nondeterminism to represent
unknown environment behavior, partial observability, or underspecification.

\item \emph{Black-box systems}, where no reliable model is available and the
system can only be observed through its executions.
\end{itemize}

\paragraph{Model checking.}
When a sufficiently precise model of the system is available (white box, or
controlled gray box), correctness questions can be addressed using
\emph{model checking}.
Given a temporal specification \(\varphi\), the model-checking problem asks
whether
\[
\mathcal S \models \varphi,
\]
that is, whether \(\varphi\) holds on \emph{all} infinite executions of
\(\mathcal S\).

In gray-box models, however, nondeterminism often represents uncertainty rather
than real implementation choices.
In this case, universal quantification over all executions may be too strong or
unfounded. A model may include executions that are possible in the abstraction but
impossible in reality, or behaviors that will only be resolved at runtime.
As a result, a purely offline verdict may be misleading.

\paragraph{From offline verification to runtime observation.}
An alternative is to observe executions as they occur.
At runtime, the system produces a concrete execution, but an observer
typically sees only a projection.
Based on the observed prefix, the observer reasons about the set of executions
of \(\mathcal S\) that are compatible with what has been seen so far.

This naturally leads to epistemic reasoning.
At any time, the observer may or may not \emph{know} that a given property
holds, depending on whether it holds on all executions consistent with the
current observation.

\newcommand{\formeq}{=}

\paragraph{Diagnosis.}
In diagnosis, the goal is not to decide an arbitrary specification, but to
detect the occurrence of a particular event, such as a fault.
Let \(\phi\) be a property representing such an event.

A diagnoser should eventually know that the event occurred, whenever it did.
The existence of a diagnoser is naturally expressed by the property
\[
\Phi_{\mathsf{diagnosable}}(\phi) \formeq 
\tmp{G}\bigl(\phi \to \tmp{F}\Knows_a \tmp{P}\phi\bigr),
\]
which can be read as
\begin{center}
\parbox{10cm}{
``It is globally true (\(\tmp{G}\)) that if $\phi$ holds, then eventually agent~$a$ knows that $\phi$ held in the past.''}
\end{center}
Thus, the existence of a diagnoser is a model-checking question, while the
diagnoser itself operates online on the observation stream.

In diagnosis, fault propositions are typically \emph{unobservable} to the
observer.
This does not mean that faults are absent from the system model.
Rather, faults represent internal events or mode changes that may influence the
future behavior of the system, even if they are not directly visible.
Modeling faults explicitly allows us to define correctness of diagnosis and to
reason about whether their occurrence can be inferred indirectly from
observations.

\begin{example}[Diagnosable vs. non-diagnosable systems]
\label{ex:intro-diagnosis}
Let $\mathit{AP}=\{p,r,e\}$ be a set of propositions and
$\Sigma = 2^{\mathit{AP}}$.
A letter $\sigma\in\Sigma$ specifies which propositions hold
at a discrete time step. It is also referred to as an \emph{event}.
Think of $e$ as ``an error occurred''.
An execution (or trace) is an infinite sequence in $\Sigma^\omega$
(i.e., it is an infinite word over $\Sigma$),
and the behavior of a system $\mathcal S$ is given by its language
$L(\mathcal S)\subseteq\Sigma^\omega$.

Given an event with letter $\sigma \in \Sigma$,
agent $a$ observes only $\sigma\cap\AP_a$, where
$\AP_a=\{p,r\}$, i.e., the error proposition $e$ is unobservable.

Consider the four systems $\mathcal S_1$, $\mathcal S_2$, $\mathcal S_3$, and $\mathcal S_4$
with the following languages:
\begin{align*}
L(\mathcal S_1) &= \{p\}^\omega + \{p\}^\ast \{p,e\}\{p\}\{r\}^\omega\\
L(\mathcal S_2) &= \{p\}^\omega + \{p\}^\ast \{p,e\}\{p\}^\omega\\
L(\mathcal S_3) &= \{p\}^\omega + \{p\}^\ast \{e\}\{p\}^\omega\\
L(\mathcal S_4) &= \{p\}^\omega + \{p\}^\ast \{p,e\}\{p\}^\ast\{r\}^\omega
\end{align*}
where we use $\alpha\beta^\omega$ as shorthand for the corresponding
$\omega$-language.

For example, the infinite word
$\{p\}\{p\}\{p,e\}\{p\}\{r\}\{r\}\{r\}\ldots$
is a possible trace of $\mathcal S_1$.
Agent $a$ observes the projection
$\{p\}\{p\}\{p\}\{p\}\{r\}\{r\}\{r\}\ldots$.
Since observing $r$ eventually allows $a$ to infer that an error has
occurred, $\mathcal S_1$ is diagnosable for the property $\phi=e$.

In contrast, $\mathcal S_2$ is not diagnosable: faulty and non-faulty
executions are observationally indistinguishable, as $a$ observes
only sequences of $\{p\}$ in both cases.

The system $\mathcal S_3$ is again diagnosable.
Although the error proposition $e$ is unobservable and all observations after the
error are $\{p\}$, the error occurs at a position where the observer sees
$\emptyset$, whereas non-faulty executions produce $\{p\}$ from the first step.
Under the synchronous perfect-recall semantics, the observer remembers the number
of steps, so these prefixes are distinguishable.
Hence agent $a$ can infer that an error has occurred, and $\mathcal S_3$ is
diagnosable.

The system $\mathcal S_4$ is also diagnosable.
Indeed, on every faulty execution, the symbol $\{r\}$ eventually occurs (and then
persists forever), whereas on non-faulty executions only $\{p\}$ is observed.
Thus, once $\{r\}$ is observed, agent $a$ can conclude that an error has occurred.
However, $\mathcal S_4$ illustrates that \emph{unbounded} diagnosis does not imply
a \emph{uniform} delay bound: after an error, the system may produce $\{p\}$ for an
arbitrarily long finite time before switching to $\{r\}^\omega$.
Therefore, for every $D\in\mathbb N$ there exists a faulty execution in which
agent $a$ cannot yet infer the error within $D$ further steps.

We conclude that
\[
\mathcal S_1 \models \Phi_{\mathsf{diagnosable}}(e)
\quad\text{and}\quad
\mathcal S_3 \models \Phi_{\mathsf{diagnosable}}(e)
\quad\text{and}\quad
\mathcal S_4 \models \Phi_{\mathsf{diagnosable}}(e),
\]
whereas
\[
\mathcal S_2 \not\models \Phi_{\mathsf{diagnosable}}(e).
\]
Note that only $\mathcal S_1$ and $\mathcal S_3$ satisfy a bounded form
of diagnosability: $\tmp{G}(e \to \tmp{X}^{D}\Knows_a\tmp{P}e)$ for suitable fixed $D$,
where $\tmp{X}^{D}$ means \emph{after} $D$ \emph{steps}.
\end{example}

\paragraph{Opacity.}
Opacity addresses the dual problem of diagnosis.
Here, the secret is not a current-state fact, but the \emph{occurrence of an event}
at some point in the execution.
The goal of opacity is to ensure that the observer can never infer that this
event has occurred.

Accordingly, we define opacity with respect to a past property.
Given a secret event $\phi$, an epistemic opacity requirement is
\[
\Phi_{\mathsf{opaque}}(\phi)
=
\tmp{G}\,\neg \Knows_a \tmp{P}\phi .
\]
Thus, opacity requires that at all times the observer does not know that the
secret event has occurred in the past.
Note that this does not prevent the observer from sometimes knowing that the
secret event has not occurred.

As before, this is a model-checking problem on the system~$\mathcal S$.

\begin{example}[Past opacity]
Let $\mathit{AP}=\{p,s\}$ and $\Sigma=2^{\mathit{AP}}$.
Agent $a$ observes only $p$, i.e., $\AP_a=\{p\}$. The proposition $s$ represents
a secret internal event and is unobservable.

Consider the system $\mathcal S$ with language
\[
L(\mathcal S)
= \{p\}^\omega + \{p\}^\ast\{p,s\}\{p\}^\omega .
\]

Let the secret event be $\phi = s$.
Then $\mathcal S$ is opaque with respect to agent $a$ and $\phi$, since for
every observation prefix consisting only of $\{p\}$, there exists an
indistinguishable execution in which $s$ has not occurred.
Hence,
\[
\mathcal{S} \models \Phi_{\mathsf{opaque}}(\phi).
\]
\end{example}

%
%
%
%

\paragraph{Monitoring.}
A monitor observes a finite prefix of an execution and produces a verdict based on
the information available so far.
In the presence of partial observability, such verdicts naturally correspond to
what the observer can infer from the observations.
Accordingly, we model monitoring using three-valued verdicts derived from
epistemic reasoning.

Let \(\varphi\) be a temporal property interpreted at the initial position of the
execution, and let
\(\hat\varphi = \tmp{P}(\mathit{first} \wedge \varphi)\).
Given an observation prefix, the monitor produces:
\begin{itemize}
\item the verdict \(\tverdict\) if \(\Knows_a \hat\varphi\) holds, i.e., the
observer knows that \(\varphi\) holds at the initial position;
\item the verdict \(\fverdict\) if \(\Knows_a \neg\hat\varphi\) holds, i.e.,
the observer knows that \(\varphi\) does not hold;
\item the verdict \(\unknown\) otherwise, reflecting remaining uncertainty.
\end{itemize}
Monitoring thus tracks whether the observer has obtained sufficient
knowledge to decide the truth of \(\varphi\).

\begin{example}[A simple epistemic monitor]
Let $\AP=\{p,r,e\}$ and $\Sigma=2^{\mathit{AP}}$.
As before, think of $e$ as ``an error occurred''.
We consider a single agent $a$ that can observe all propositions, i.e.,
$\AP_a=\mathit{AP}$.

Consider the system $\mathcal S$ with
\[
L(\mathcal S)
= \{p\}^\ast\{r\}\{p\}^\omega + \{p\}^\ast\{e\}\{p\}^\omega.
\]
That is, after some (possibly empty) prefix of $\{p\}$, exactly one of $\{r\}$ or
$\{e\}$ occurs, and then only $\{p\}$ forever.

Let the property be $\varphi = \neg \tmp{F}e$, expressing that no error ever occurs.
As long as the observer sees only $\{p\}$, both executions in which an error will
occur and executions in which no error occurs remain compatible with the
observations.
Hence, neither $\Knows_a \hat\varphi$ nor $\Knows_a \neg\hat\varphi$ holds, and the
monitor outputs $\unknown$.

Upon observing the first non-$\{p\}$ letter, epistemic uncertainty is resolved.
If the letter is $\{e\}$, then the property is violated, $\Knows_a \neg\hat\varphi$
holds, and the monitor outputs $\fverdict$.
If the letter is $\{r\}$, then it is guaranteed that no error will ever occur,
$\Knows_a \hat\varphi$ holds, and the monitor outputs $\tverdict$.
\end{example}

\paragraph{A unifying perspective.}
Model checking, runtime verification, diagnosis, and opacity can all be
analyzed within a common epistemic temporal framework over the same
system model.
The key difference lies in \emph{when} and \emph{how} the result is used.

Model checking provides an offline verdict about all executions.
Diagnosis and opacity are typically formulated as epistemic temporal properties
whose satisfaction can be decided offline.
Monitoring, in contrast, produces verdicts online from finite observations, but
relies on similar epistemic reasoning about what an observer can infer.

In this sense, runtime verification does not replace offline analysis.
Instead, offline analysis is used to construct monitors, diagnosers, or observers,
which are then applied at runtime to the observed execution.

\section{Take-Away}

\begin{itemize}
\item Runtime verification reasons about system correctness from \emph{observed
executions}, rather than from complete system models.

\item Partial observation and infinite executions introduce uncertainty, which
can be captured naturally using epistemic notions of knowledge.

\item Model checking, diagnosis, opacity, and monitoring can all be analyzed using
epistemic temporal reasoning over the same abstract system model.

\item Diagnosis and opacity express availability or impossibility of knowledge,
while monitoring general properties raises questions about when
definitive verdicts can be reached.

\item Although monitors operate online, their existence and correctness are
typically established by an \emph{offline analysis} of the system model.
\end{itemize}

\section{Bibliographic Notes}

Runtime verification has been extensively studied as a lightweight alternative
to exhaustive model checking, in particular for linear-time temporal logics
\cite{BartocciFFR18,LeuckerS09,BauerLS11}.
For background on model checking and temporal logics, we refer to the standard
textbook~\cite{BK-model-checking-2008}.
Diagnosis and opacity originate from the discrete-event systems community, where
they are traditionally formulated in terms of observation-based
distinguishability of behaviors
\cite{SampathSLST95,LafortuneLH18}.
Opacity has also been studied on transition systems and from an
information-flow perspective \cite{BryansKMR08}.
Epistemic logic provides a unifying formalism for reasoning about such
distinguishability, by making explicit what an observer can infer from partial
observations.
The interaction between temporal and epistemic operators has been studied in the context of knowledge and time under perfect-recall semantics (e.g.,~\cite{FHMV1995}). In these notes, we make use of this connection and adopt epistemic temporal logic as a common language to relate monitoring, diagnosis, and opacity in a runtime setting.

\chapter{Observation, Knowledge, and Uncertainty}

\section{Words, Languages, and Moore Machines}

We let $\mathbb{N} = \{0, 1, 2, 3, \ldots\}$ and $\Npos = \{1, 2, 3, \ldots\}$.

A possible behavior of a reactive system is naturally viewed as a sequence of observations
over time.
Such behaviors are represented as \emph{words}, also called \emph{traces}, which
may be finite or infinite sequences over a given alphabet.

\begin{definition}[Words]
Let $\Sigma$ be a finite alphabet.
A \emph{finite word} over
$\Sigma$ is a sequence $u=\sigma_1 \ldots \sigma_n$ with $n \in \mathbb{N}$ and $\sigma_i \in \Sigma$ for $i \in \{1,\ldots,n\}$.
We let $|u|$ refer to the \emph{length} $n$ of $u$.
If $n=0$, then $u$ is the empty word, denoted by $\varepsilon$.
By $\Sigma^\ast$, we denote the set of all finite words
over $\Sigma$. The set of nonempty finite words is denoted by
$\Sigma^+ = \Sigma^\ast \setminus \{\varepsilon\}$.

\medskip

An \emph{infinite word} over $\Sigma$ is a sequence $w=\sigma_1 \sigma_2 \sigma_3 \ldots$ with
$\sigma_i \in \Sigma$ for $i \in \Npos$.
The set of infinite words is denoted by $\Sigma^\omega$.
Moreover, we let $\Sigma^\infty = \Sigma^\ast \cup \Sigma^\omega$.
\end{definition}

Given $w=\sigma_1 \sigma_2 \sigma_3 \ldots \in \Sigma^\omega$ and
$k \in \mathbb{N}$, we let $w_{\le k}$ refer to the finite prefix
$\sigma_1 \ldots \sigma_k \in \Sigma^\ast$
of length $k$. Note that $w_{\le 0} = \varepsilon$.
For a finite word
$u = \sigma_1 \ldots \sigma_n \in \Sigma^\ast$ and $k \in \{0,\ldots,n\}$,
the prefix of length $k$ of $u$ is
defined accordingly, as $u_{\le k} = \sigma_1 \ldots \sigma_k$.

A \emph{language} is a set of words. We mainly consider languages $L \subseteq \Sigma^\ast$
over finite words, and languages $L \subseteq \Sigma^\omega$ over infinite words.
For $L \subseteq \Sigma^\infty$, we let
$\Pref(L)$ refer to $\{w_{\le k} \mid w \in L,\ k \in \Npos,\ \text{and if } w \in \Sigma^\ast,\ \text{then } k \le |w|\}$, i.e., the set of \emph{nonempty} finite prefixes of words in $L$.

\begin{example}[A simple word]\label{ex:simple-trace}
Let \(\AP = \{p,r\}\), i.e., $\Sigma = \{\emptyset, \{p\}, \{r\}, \{p, r\}\}$.
Then, $\{p\}\emptyset\{p\}\emptyset\{p\}\emptyset \ldots \in \Sigma^\omega$ is an infinite word.
Moreover, \(\{p\}\{p,r\}\emptyset \in \Sigma^\ast\) is a finite word of length \(3\).
\end{example}

While we will later define Büchi automata as a central automata model running over
infinite words, we already introduce deterministic automata with output, also
known as Moore machines.
These will play a crucial role as runtime monitors and as building blocks in
constructions over Büchi automata.

\begin{definition}[Moore machine]
A \emph{Moore machine} over an input alphabet $\Sigma$ and an output alphabet
$\Gamma$ is a tuple
\[
\mathcal{A}=(Q, \iota, \delta, \outmap)
\]
where
\begin{itemize}
  \item $Q$ is a finite set of \emph{states},
  \item $\iota\in Q$ is the \emph{initial state},
  \item $\delta\colon Q \times \Sigma \to Q$ is a \emph{transition function},
  \item $\outmap\colon Q \to \Gamma$ is an \emph{output mapping}.
\end{itemize}
\end{definition}

We extend $\delta$ inductively to $\delta^\ast\colon Q \times \Sigma^\ast \to Q$
by $\delta^\ast(q,\varepsilon)=q$ and
$\delta^\ast(q,w\sigma)=\delta(\delta^\ast(q,w),\sigma)$ for all
$q\in Q$, $w\in\Sigma^\ast$, and $\sigma\in\Sigma$.
As its semantics, the Moore machine $\mathcal{A}$ defines the following mapping:
\[
\sem{\mathcal{A}}\colon
\begin{cases}
\Sigma^\ast \to \Gamma\\
u \mapsto \outmap(\delta^\ast(\iota,u))
\end{cases}
\]
In other words, it scans the input word and outputs the symbol associated with the
resulting state.

\medskip
\noindent
\textbf{Special case: deterministic finite automata.}
For $\Gamma=\{0,1\}$, a Moore machine
$\mathcal A=(Q,\iota,\delta,\outmap)$ can be viewed as a deterministic finite
automaton over $\Sigma$ by identifying the set of accepting states as
\[
F=\{\,q\in Q \mid \outmap(q)=1\,\}.
\]
In this case, we will equivalently write $\mathcal A=(Q,\iota,\delta,F)$ and
understand it as a recognizer for the language
\[
L(\mathcal A)=\{\,u\in\Sigma^\ast \mid \delta^\ast(\iota,u)\in F\,\}.
\]

\section{Propositions, Agents, and Partial Observability}

In our setting, the behavior of a reactive system is modeled as an infinite word
over an alphabet whose letters indicate which atomic propositions hold at each
time step.
A single position of such a word, together with its associated set of
propositions, is called an \emph{event}.

In many systems, an observer cannot see all aspects of an event.
Instead, each agent has access only to a subset of the atomic propositions.
As a consequence, different system executions may be indistinguishable from the
point of view of a given agent.
This form of partial observability is formalized by projecting global executions
onto the propositions visible to the agent and will later form the basis of our
epistemic semantics.

Throughout these notes, unless stated otherwise, we fix a finite set of atomic
propositions $\AP$ and let $\Sigma = 2^{\AP}$ be the corresponding global alphabet.
We also fix a finite set $\Ag$ of \emph{agents}.
For each agent $a\in\Ag$, let $\AP_a\subseteq\AP$ denote the set of propositions
observable by $a$, and let $\Sigma_a=2^{\AP_a}$ be the corresponding observation
alphabet.

\begin{definition}[Observable propositions]\label{def:observable-props}
The observation function for agent $a$ is
\[
  \obs_a \colon
  \begin{cases}
  \Sigma \to \Sigma_a\\
  \sigma \mapsto \sigma \cap \AP_a
  \end{cases}
\]
It is extended pointwise to words in $\Sigma^\ast$ and $\Sigma^\omega$.
Note that $\obs_a$ is then length-preserving and that $\obs_a(\varepsilon) = \varepsilon$.
\end{definition}

The observation function formalizes the local view of an agent at a single step. Given a global event, the agent retains only the information carried by the propositions it can observe and ignores the rest. Extending observations from single events to words allows us to compare entire executions from the agent's perspective. 

\begin{example}[Observation of words]
Let $\mathit{AP}=\{p,r,e\}$ and let agent $a$ observe only
$\AP_a=\{p,r\}$.

Consider the finite word
\[
u = \{p\}\{p,e\}\{r\} \in \Sigma^\ast.
\]
Then
\[
\obs_a(u) = \{p\}\{p\}\{r\}.
\]

If an event contains only unobservable propositions, the observation is empty.
For example, for
\[
u' = \{p\}\{e\}\{r\} \in \Sigma^\ast,
\]
we have
\[
\obs_a(u') = \{p\}\emptyset\{r\}.
\]
\end{example}

Using observations, we can now formalize when an agent is unable to distinguish between two executions. Intuitively, two executions are indistinguishable for an agent if they give rise to exactly the same sequence of observations.

\begin{definition}[Indistinguishability for finite words]
\label{def:indist}
Two finite words $u, u' \in \Sigma^\ast$ are 
\emph{indistinguishable to agent $a$}, written $u \sim_a u'$, if they have 
the same observation for agent $a$:
\[
  u \sim_a u'
  \quad\text{if}\quad
  \obs_a(u) = \obs_a(u')
\]
In particular, indistinguishable words are of the same length.
\end{definition}

Note that indistinguishability is an equivalence relation on finite words.

\begin{example}[Indistinguishable finite words]
\label{ex:indist-finite}
Let $\AP=\{p,r,e\}$ and let agent $a$ observe
$\AP_a=\{p,r\}$, as in Example~\ref{ex:intro-diagnosis}.
Consider the finite words
\[
u = \{p\}\{p,e\}\{p\}
\qquad\text{and}\qquad
u' = \{p\}\{p\}\{p\}.
\]
Since the proposition $e$ is unobservable for agent~$a$, we have
\[
\obs_a(u) = \obs_a(u') = \{p\}\{p\}\{p\}.
\]
Hence, $u \sim_a u'$.
\end{example}

When reasoning about properties at a specific time point, it is not sufficient to compare complete executions. We therefore lift indistinguishability to pairs consisting of an execution and a position. Two such pairs are indistinguishable if the agent has seen the same observations up to that position.

\begin{definition}[Indistinguishability for pointed infinite words]
\label{def:pointed-word}
A \emph{pointed word} is a pair $(w, k)$ where $w \in \Sigma^\omega$
and $k$ is a position in $w$, i.e., $k \in \Npos$.

Two pointed words $(w,k)$ and $(w',k)$ are 
\emph{indistinguishable to agent $a$} (perfect recall), written 
$(w,k)\sim_a(w',k)$, if their prefixes up to position $k$ are indistinguishable:
\[
  (w,k)\sim_a(w',k)
  \quad\text{if}\quad
  w_{\le k} \sim_a w'_{\le k}
\]
\end{definition}


\begin{notebox}[Synchronous vs. asynchronous observation]
In these notes, we consider a \emph{synchronous} observation semantics: system progress
is observable by all agents. Formally, at each step, every agent receives an
observation, possibly empty. Equivalently, one may assume a distinguished atomic
proposition $\mathit{tick}$ that holds at every step and is observable by all agents.
Thus, while a transition may be uninformative, it is never (epistemically) silent.

Under this assumption, indistinguishability is naturally defined between prefixes
of equal length, as in Definition~\ref{def:indist}.

\medskip
An alternative \emph{asynchronous} observation semantics can be obtained by letting
agents observe only steps carrying nonempty observations. Formally, one replaces
\[\obs_a \colon \Sigma^\ast \to \Sigma_a^\ast\] by the projection that erases events
$\sigma$ with \[\obs_a(\sigma)=\emptyset\] and accordingly extend this to indistinguishability by
\[
(w,k)\sim_a(w',k')
\quad\Longleftrightarrow\quad
\obs_a(w_{\le k})=\obs_a(w'_{\le k'}) \, .
\]
In this case, agents may be uncertain about the number of system steps or the
elapsed time.

All epistemic definitions and constructions presented in these notes extend
to this asynchronous variant by replacing the observation
function accordingly. For simplicity, we focus on the synchronous case.
\end{notebox}

\section{A Base System Model}

Systems are characterized by their set of infinite traces, i.e.,
the infinite words they possibly generate.
These can, for example, be described in terms of $\omega$-regular expressions (cf.~Chapter~\ref{ch:introduction})
or in terms of automata.
For an algorithmic treatment and also from a modeling point of view,
it is convenient to use Büchi automata.
They provide a robust finite representation of infinite-state behavior at the
level of traces.
Moreover, they are particularly well suited to the constructions developed in
the following sections.

\begin{definition}[Büchi automaton]\label{def:base-model}
A \emph{B\"uchi automaton} over $\Sigma$ is a tuple
\[\mathcal A=(Q,\iota,\Trans,F)\] where
\begin{itemize}
  \item \(Q\) is a finite set of \emph{states} and \(\iota \in Q\) is an \emph{initial state},
  \item \(\Trans \subseteq Q \times \Sigma \times Q\) is a \emph{transition relation},
  \item \(F\subseteq Q\) is a set of \emph{accepting states}.
\end{itemize}
\end{definition}

Let $q, q' \in Q$ be states.
For $\sigma \in \Sigma$, we write $q \xrightarrow{\sigma} q'$ if $(q, \sigma, q') \in \Trans$.
A \emph{run} of $\mathcal{A}$ is an infinite sequence of the form 
$q_0 \xrightarrow{\sigma_1} q_1 \xrightarrow{\sigma_2} q_2 \xrightarrow{\sigma_3} \ldots$
such that $q_0 = \iota$.

A \emph{run from $q$} (or $q$ run) is defined analogously, except that $q_0=q$.
However, unless stated otherwise, the term \emph{run} always
refers to a run starting in the initial state.

We let $\pi(i)$ refer to $q_i$ for all $i \in \mathbb{N}$.
Run $\pi$  is called \emph{accepting} if $q_i \in F$ for infinitely many $i \in \mathbb{N}$.
The \emph{induced trace} is $\sigma_1 \sigma_2 \sigma_3 \ldots \in \Sigma^\omega$, denoted by $\trace(\pi)$.
We also say that $\pi$ is a run \emph{over} $w$.

The \emph{language} of $\mathcal{A}$ is defined as
\[
L(\mathcal A) = \formset{\trace(\pi) \mid \pi \text{ is an accepting run of } \mathcal A }.
\]

It is convenient to also introduce finite runs. A \emph{finite run} of $\mathcal{A}$
is a finite sequence of the form 
$\pi = q_0 \xrightarrow{\sigma_1} q_1 \xrightarrow{\sigma_2} \ldots \xrightarrow{\sigma_n} q_n$,
where $n \in \mathbb{N}$, such that $q_0 = \iota$.
We also set $\pi(i) = q_i$ for all $i \in \{0,\ldots, n\}$.
Moreover, we let $\trace(\pi) = \sigma_1 \ldots \sigma_n$, and we say that
$\pi$ is a run over $\trace(\pi)$.
We do not define acceptance for finite runs.

For an automaton $\mathcal A$, we write $|\mathcal A|$ for the number
$|Q|$ of its states, which we take as a measure of its size.

\begin{notebox}
For technical convenience in some constructions, we also use \emph{generalized Büchi 
automata} (GBAs), where the acceptance condition is given by a set of accepting sets 
$\mathcal{F} = \{F_1, \ldots, F_m\}$ with $F_i \subseteq Q$. A run is accepting if, for each $i \in \{1,\ldots,m\}$, 
it visits some state from $F_i$ infinitely often. Every GBA can be converted to a 
standard Büchi automaton with at most $m \cdot |Q|$ states (cf.~\cite{BK-model-checking-2008}).
\end{notebox}

\newcommand{\Obs}{\mathit{Obs}}

We define
\[
\Obs_a(\mathcal A) = \bigl\{\,\obs_a(u) \mid u \in \Pref(L(\mathcal A))\,\bigr\}
\]
for the set of all (nonempty) observation sequences of agent~$a$ that are compatible
with some finite execution of $\mathcal A$.

\begin{notebox}
System states are not observed directly.
Instead, the observer sees only the projected execution trace over the observation
alphabet~\(\Sigma_a\).
Classical runtime monitoring corresponds to a maximally permissive special case,
where the system model consists of a single (accepting) state with a self-loop for every
event~\(\sigma \in \Sigma\).
\end{notebox}

\newcommand{\wprop}{\mathit{wr}}
\newcommand{\rprop}{\mathit{rd}}

\begin{example}[Four trace languages under partial observation]\label{ex:four-languages}
Let us revisit Example~\ref{ex:intro-diagnosis} and provide
Büchi automata for the languages considered there.
Recall that $\AP=\{p,r,e\}$ and $\Sigma=2^{\AP}$.
A letter $\sigma\in\Sigma$ specifies which propositions hold at a discrete time
step.
Fix an agent $a$ with $\AP_a=\{p,r\}$, so the proposition $e$ is unobservable.

We consider four systems $\mathcal S_1,\mathcal S_2,\mathcal S_3,\mathcal S_4$
given by the following languages:
\begin{align*}
L(\mathcal S_1) &= \{p\}^\omega + \{p\}^\ast \{p,e\}\{p\}\{r\}^\omega
\\
L(\mathcal S_2) &= \{p\}^\omega + \{p\}^\ast \{p,e\}\{p\}^\omega
\\
L(\mathcal S_3) &= \{p\}^\omega + \{p\}^\ast \{e\}\{p\}^\omega
\\
L(\mathcal S_4) &= \{p\}^\omega + \{p\}^\ast \{p,e\}\{p\}^\ast\{r\}^\omega
\end{align*}

Each language is generated by the following (small) B\"uchi automaton
(the initial state has an incoming arrow and
final states are depicted as a double circle):

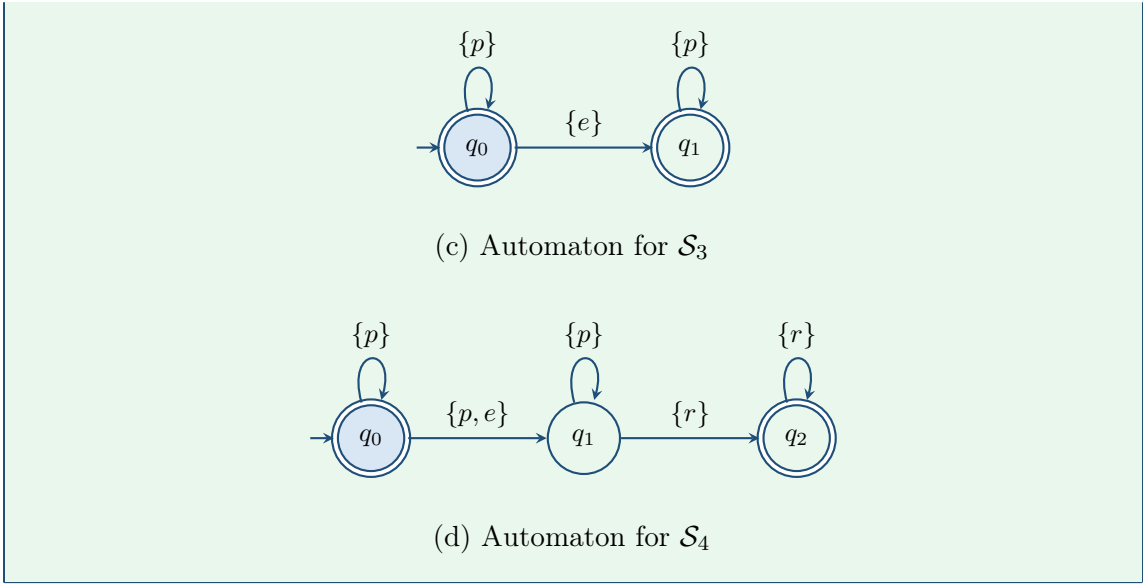
\begin{center}
\begin{tikzpicture}[
  ->, >=stealth, node distance=2.8cm,
  state/.style={circle, draw=ChapterAccent, thick, minimum size=0.95cm, inner sep=1pt},
  initial/.style={state, fill=ChapterLight},
  accepting/.style={state, double, double distance=1.4pt},
  every edge/.style={draw=ChapterAccent, thick},
  every node/.style={font=\small}
]
\node[initial,accepting] (i1) {$q_0$};
\node[accepting] (m1) [right of=i1] {$q_1$};
\node[accepting] (r1) [right of=m1] {$q_2$};

\draw[thick, ChapterAccent, <-] (i1) -- ++(-0.8,0);

\path
(i1) edge[loop above] node {$\{p\}$} (i1)
     edge node[above] {$\{p,e\}$} (m1)
(m1) edge node[above] {$\{p\}$} (r1)
(r1) edge[loop above] node {$\{r\}$} (r1);
\end{tikzpicture}

\smallskip
(a) Automaton for $\mathcal S_1$
\end{center}

\begin{center}
\begin{tikzpicture}[
  ->, >=stealth, node distance=2.8cm,
  state/.style={circle, draw=ChapterAccent, thick, minimum size=0.95cm, inner sep=1pt},
  initial/.style={state, fill=ChapterLight},
  accepting/.style={state, double, double distance=1.4pt},
  every edge/.style={draw=ChapterAccent, thick},
  every node/.style={font=\small}
]
\node[initial,accepting] (i2) {$q_0$};
\node[accepting] (m2) [right of=i2] {$q_1$};

\draw[thick, ChapterAccent, <-] (i2) -- ++(-0.8,0);

\path
(i2) edge[loop above] node {$\{p\}$} (i2)
     edge node[above] {$\{p,e\}$} (m2)
(m2) edge[loop above] node {$\{p\}$} (m2);
\end{tikzpicture}

\smallskip
(b) Automaton for $\mathcal S_2$
\end{center}

\begin{center}
\begin{tikzpicture}[
  ->, >=stealth, node distance=2.8cm,
  state/.style={circle, draw=ChapterAccent, thick, minimum size=0.95cm, inner sep=1pt},
  initial/.style={state, fill=ChapterLight},
  accepting/.style={state, double, double distance=1.4pt},
  every edge/.style={draw=ChapterAccent, thick},
  every node/.style={font=\small}
]
\node[initial,accepting] (i3) {$q_0$};
\node[accepting] (m3) [right of=i3] {$q_1$};

\draw[thick, ChapterAccent, <-] (i3) -- ++(-0.8,0);

\path
(i3) edge[loop above] node {$\{p\}$} (i3)
     edge node[above] {$\{e\}$} (m3)
(m3) edge[loop above] node {$\{p\}$} (m3);
\end{tikzpicture}

\smallskip
(c) Automaton for $\mathcal S_3$
\end{center}

\begin{center}
\begin{tikzpicture}[
  ->, >=stealth, node distance=2.8cm,
  state/.style={circle, draw=ChapterAccent, thick, minimum size=0.95cm, inner sep=1pt},
  initial/.style={state, fill=ChapterLight},
  accepting/.style={state, double, double distance=1.4pt},
  every edge/.style={draw=ChapterAccent, thick},
  every node/.style={font=\small}
]
\node[initial,accepting] (q0) {$q_0$};
\node[state]             (q1) [right of=q0] {$q_1$};
\node[accepting]         (q2) [right of=q1] {$q_2$};

\draw[thick, ChapterAccent, <-] (q0) -- ++(-0.8,0);

\path
(q0) edge[loop above] node {$\{p\}$} (q0)
     edge node[above] {$\{p,e\}$} (q1)
(q1) edge[loop above] node {$\{p\}$} (q1)
     edge node[above] {$\{r\}$} (q2)
(q2) edge[loop above] node {$\{r\}$} (q2);
\end{tikzpicture}

\smallskip
(d) Automaton for $\mathcal S_4$
\end{center}

\end{example}

\section{Epistemic Linear-Time Temporal Logic}

We now formalize the epistemic temporal logic outlined in the introduction.
The logic allows us to specify temporal properties of systems under
partial observability, and to reason explicitly about what agents know
based on their observations.
In particular, it provides a logical framework for formalizing
diagnosability or opacity properties.

\begin{definition}[Epistemic LTL]
Formulas from \emph{Epistemic Linear-time Temporal Logic}, denoted by \LTL, are given by
the following grammar:
\[
\begin{array}{rcl}
\varphi &::=& p \;\mid\; \neg \varphi \;\mid\; \varphi \land \varphi
			\;\mid\; \phi\sUntil\phi \;\mid\; \phi \sSince \phi \;\mid\; \Knows_a \varphi 
\end{array}
\]
where $p \in \AP$ is a proposition and $a \in \Ag$ is an agent.

The fragment without epistemic operators is denoted by \pureLTL.
It still contains both future and past modalities.
\end{definition}

Formulas are interpreted over pointed words.
Recall that a pointed word consists of an execution together with a distinguished
reference position (cf.\ Definition~\ref{def:pointed-word}).
This position separates the finite prefix that has already been
observed from the infinite future that is yet to unfold.
Both temporal and epistemic operators are evaluated relative to this
reference position.

\begin{definition}[Satisfaction]
Let $\mathcal S$ be a Büchi automaton over $\Sigma$ (representing a system).
Let $w = \sigma_1 \sigma_2 \sigma_3 \ldots \in L(\mathcal S)$ and $k \in \Npos$.
Satisfaction $\mathcal{S}, w, k \models \phi$ for a formula $\phi$ is
inductively defined as follows:

\begin{itemize}
  \item $\mathcal{S}, w, k \models p$ if $p \in \sigma_k$
  
  \item $\mathcal{S}, w, k \models \neg \phi$ if $\mathcal{S}, w, k \not\models \phi$

  \item $\mathcal{S}, w, k \models \phi \wedge \psi$ if $\mathcal{S}, w, k \models \phi$ and $\mathcal{S}, w, k \models \psi$
    
  \item $\mathcal{S}, w, k \models \phi\sUntil\psi$ if
        there exists $j$ with $k < j$ such that
        \begin{itemize}
        \item $\mathcal{S}, w, j \models \psi$ and
        \item for all $i$ with $k < i < j$, we have $\mathcal{S}, w, i \models \phi$
        \end{itemize}
  
  \item $\mathcal{S}, w, k \models \phi \sSince \psi$ if
        there exists $j$ with $1 \le j < k$ such that
        \begin{itemize}
        \item $\mathcal{S}, w, j \models \psi$ and
        \item for all $i$ with $j < i < k$, we have $\mathcal{S}, w, i \models \phi$
        \end{itemize}
        
  \item $\mathcal{S}, w, k \models \Knows_a\varphi$ if for all $w' \in L(\mathcal{S})$ with
        $(w,k) \sim_a (w',k)$, we have $\mathcal{S}, w', k \models \varphi$
\end{itemize}

Moreover, we write $\mathcal S \models \varphi$ if, for every $w \in L(\cal S)$,
we have $\mathcal S,w,1 \models \varphi$.
\end{definition}

\begin{notebox}[Strict temporal operators]
We take the \emph{strict} temporal operators $\sUntil$ and $\sSince$ as primitives.
Thus, $\phi \sUntil \psi$ requires $\psi$ to hold at a strictly future
position, and $\phi \sSince \psi$ at a strictly past one.
The left-hand side $\phi$ is required to hold only at positions strictly between
the current position and the witnessing position. Its truth at the current position
is irrelevant. This avoids separate next and previous operators as primitives, since they can be
defined using strict until and since.
It also anticipates the timed setting.
In timed logics, strictness is useful to constrain the elapsed time to the closest
future or past occurrence of a subformula.
Using strict operators already in the untimed logic yields a uniform semantics and
simplifies the timed extension.
\end{notebox}

\paragraph{Derived temporal operators.}
Using the strict operators $\sUntil$ and $\sSince$, we introduce the following
abbreviations, which include the standard non-strict operators:
\[
\begin{array}{rcl@{\qquad}rcl}
\top &:=& p \lor \neg p & 
\bot &:=& \neg\top
\\[0.8ex]
\varphi \lor \psi &:=& \neg(\neg\varphi \land \neg\psi) &
\varphi \to \psi &:=& \neg \varphi \lor \psi
\\[0.8ex]
\tmp{X}\varphi &:=& \bot \sUntil \varphi &
\tmp{Y}\varphi &:=& \bot \sSince \varphi
\\[0.8ex]
\varphi \tmp{U} \psi &:=& \psi \lor (\varphi \land (\varphi \sUntil \psi)) &
\varphi \tmp{S} \psi &:=& \psi \lor (\varphi \land (\varphi \sSince \psi))
\\[0.8ex]
\tmp{F}\varphi &:=& \top \tmp{U} \varphi &
\tmp{P}\varphi &:=& \top \tmp{S} \varphi
\\[0.8ex]
\tmp{G}\varphi &:=& \neg\tmp{F}\neg\varphi &
\tmp{H}\varphi &:=& \neg\tmp{P}\neg\varphi
\\[0.8ex]
\mathit{first} &:=&
\neg\tmp{Y}\top
&
\tmp{W}_a \varphi &:=&
\Knows_a \varphi \lor \Knows_a \neg \varphi
\end{array}
\]
Here, $a$ is an agent and $p \in \AP$ any arbitrary proposition.
In particular, the formula $\tmp{W}_a\varphi$ denotes that $a$ knows \emph{whether} $\phi$ holds.

In particular, $\tmp{X}\phi$ requires that $\phi$ holds at the immediately
next position, while $\tmp{Y}\phi$ requires that an immediately preceding
position exists and that $\phi$ holds at that position.

\paragraph{Bounded temporal operators.}
For runtime verification and diagnosability, we also introduce bounded eventually and 
past operators. For $\dbound \in \mathbb{N}$:
\begin{align*}
\tmp{F}^{\le \dbound}\phi &\;:=\; \phi \lor \tmp{X}\phi \lor \tmp{X}^2\phi \lor \cdots \lor \tmp{X}^\dbound\phi 
  && \text{(eventually within $\dbound$ steps)} \\
\tmp{P}^{\le \dbound}\phi &\;:=\; \phi \lor \tmp{Y}\phi \lor \tmp{Y}^2\phi \lor \cdots \lor \tmp{Y}^\dbound\phi 
  && \text{(sometime in last $\dbound$ steps)}
\end{align*}
where $\tmp{X}^0\phi = \phi$ and $\tmp{X}^{i+1}\phi = \tmp{X}(\tmp{X}^i\phi)$, 
and similarly for $\tmp{Y}^i$.

Semantically, $\mathcal{S}, w, k \models \tmp{F}^{\le \dbound}\phi$ holds iff there 
exists $j$ with $k \le j \le k+\dbound$ such that $\mathcal{S}, w, j \models \phi$.

Bounded always operators $\tmp{G}^{\le \dbound}\phi := \neg\tmp{F}^{\le \dbound}\neg\phi$ 
and $\tmp{H}^{\le \dbound}\phi := \neg\tmp{P}^{\le \dbound}\neg\phi$ can be defined analogously.

\begin{definition}[Logical equivalence]
\LTL formulas $\varphi$ and $\psi$ are \emph{logically equivalent}, written
\[
\varphi \equiv \psi ,
\]
if for every Büchi automaton $\mathcal S$, every word
$w = \sigma_1\sigma_2\sigma_3\ldots \in L(\mathcal S)$,
and every position $k \in \Npos$,
\[
\mathcal S, w, k \models \varphi
\quad\Longleftrightarrow\quad
\mathcal S, w, k \models \psi.
\]
\end{definition}

\begin{notebox}
Throughout these notes, we fix a system $\mathcal S$ and restrict attention to
executions of $\mathcal S$.
In particular, whenever satisfaction $\mathcal S,w,k \models \varphi$ is evaluated,
we assume $w \in L(\mathcal S)$.
Thus, $L(\mathcal S)$ plays the role of the universe of possible executions.
We will restate this convention explicitly whenever it is relevant.
\end{notebox}

\begin{exercise}[Basic axioms]\label{ex:knowledge-sanity}
Let $\mathcal S$ be a system, let $a\in Ag$, and let $\phi$ be a formula.

\begin{enumerate}[label=\textbf{(\arabic*)}]
\item
Show that
\[
\mathcal S \models \phi \quad\Longleftrightarrow\quad \mathcal S \models \Knows_a \phi.
\]

\item
Show that, for all $w \in L(\mathcal S)$ and all $k \in \Npos$, we have
\[
\begin{array}{lrcll}
\text{(a)} & \mathcal S,w,k \models \Knows_a \phi
&\;\Longrightarrow\;&
\mathcal S,w,k \models \phi
&\text{(Veridicality)}\\[0.4em]
\text{(b)} & \mathcal S,w,k \models \Knows_a \phi
&\;\Longrightarrow\;&
\mathcal S,w,k \models \Knows_a \Knows_a \phi
&\text{(Positive introspection)}\\[0.4em]
\text{(c)} & \mathcal S,w,k \models \neg \Knows_a \phi
&\;\Longrightarrow\;&
\mathcal S,w,k \models \Knows_a \neg \Knows_a \phi
&\text{(Negative introspection)}\\[0.4em]
\text{(d)} & \mathcal S,w,k \models \Knows_a \phi
&\;\Longrightarrow\;&
\mathcal S,w,k \models \tmp{G}\Knows_a \tmp{P} \phi
&\text{(Non-forgetting)}
\end{array}
\]

\item
Show that the converse of the implication in~(2a) does not hold in general.
\end{enumerate}
\end{exercise}

\begin{solution}
We provide here the solution to (2d):

Let $w\in L(\mathcal S)$ and $k\in\Npos$, and assume
\[
\mathcal S,w,k \models \Knows_a \phi.
\]
We show
\[
\mathcal S,w,k \models \tmp{G}\Knows_a\tmp{P}\phi.
\]
By the semantics of $\tmp{G}$, it is enough to prove that for every $\ell\ge k$,
\[
\mathcal S,w,\ell \models \Knows_a\tmp{P}\phi.
\]
So fix $\ell\ge k$, and let $w'\in L(\mathcal S)$ with
\[
(w,\ell)\sim_a(w',\ell).
\]
By definition of indistinguishability, this means
\[
\obs_a(w_{\le \ell})=\obs_a(w'_{\le \ell}).
\]
Taking prefixes of length $k$ (since $k\le \ell$), we obtain
\[
\obs_a(w_{\le k})=\obs_a(w'_{\le k}),
\]
hence $(w,k)\sim_a(w',k)$.
From $\mathcal S,w,k \models \Knows_a \phi$, it follows that
\[
\mathcal S,w',k \models \phi.
\]
Because $k\le \ell$, this implies
\[
\mathcal S,w',\ell \models \tmp{P}\phi.
\]

Since $w'$ was arbitrary among words indistinguishable from $w$ at time $\ell$, we get
\[
\mathcal S,w,\ell \models \Knows_a\tmp{P}\phi.
\]
As $\ell\ge k$ was arbitrary, we conclude
\[
\mathcal S,w,k \models \tmp{G}\Knows_a\tmp{P}\phi.
\]
This completes the proof.
\end{solution}

\begin{example}[Diagnosability properties]\label{ex:diagnosability}
We consider diagnosability properties for the systems
$\mathcal S_1,\mathcal S_2,\mathcal S_3,\mathcal S_4$
from Example~\ref{ex:four-languages}, where $e$ is unobservable to agent~$a$.

\paragraph{Immediate diagnosability.}
$\tmp{G}(e \to \Knows_a e)$ requires that errors are immediately detected.
This fails for all four systems, as the letter $\{p,e\}$ is observed as $\{p\}$.

\paragraph{Eventual diagnosability.}
$\tmp{G}(e \to \tmp{F}\Knows_a \tmp{P} e)$ requires that errors are eventually detected.
Systems $\mathcal S_1$, $\mathcal S_3$, and $\mathcal S_4$ satisfy this property:
in $\mathcal S_1$, the observation $\{r\}$ occurs only on faulty executions;
in $\mathcal S_3$, the error occurs when the observer sees $\emptyset$, distinguishing it from non-faulty executions;
in $\mathcal S_4$, faulty executions eventually produce $\{r\}^\omega$ while non-faulty ones remain in $\{p\}^\omega$.
System $\mathcal S_2$ fails this property, as faulty and non-faulty executions remain indistinguishable.

\paragraph{Bounded diagnosability.}
For fixed $D\in\mathbb N$, the formula $\tmp{G}(e \to \tmp{X}^{D}\Knows_a \tmp{P} e)$ requires detection within $D$ steps.
Systems $\mathcal S_1$ and $\mathcal S_3$ satisfy this for suitable $D$.
System $\mathcal S_4$ shows that eventual diagnosability does not imply bounded diagnosability: after an error, the system may produce $\{p\}$ for arbitrarily many steps before switching to $\{r\}^\omega$.
\end{example}

%
%
%

\begin{example}[Nested knowledge with two agents]\label{ex:nested-knowledge}
We present a system where agent~$a$ knows that agent~$b$ knows a fact, although agent~$a$ does not know the fact itself.

Let $AP = \{h,r,s\}$ and $Ag = \{a,b\}$ with $AP_a = \{r\}$ and $AP_b = \{s\}$.
Thus $h$ is unobservable to both agents.

The system $\mathcal{S}$ nondeterministically chooses $h\in\{0,1\}$, emits $r$, then emits $s$ iff $h=1$, and remains silent thereafter.
Formally, $\mathcal{S}$ is the Büchi automaton below:

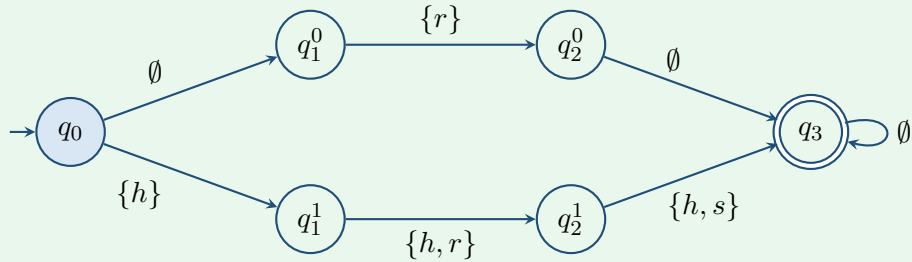
\begin{center}
\begin{tikzpicture}[
    ->,
    >=stealth,
    node distance=2.5cm,
    state/.style={
      circle,
      draw=ChapterAccent,
      thick,
      minimum size=0.9cm,
      inner sep=1pt
    },
    accepting/.style={
      state,
      double,
      double distance=1.5pt
    },
    initial/.style={
      state,
      fill=ChapterLight
    },
    every edge/.style={
      draw=ChapterAccent,
      thick
    },
    every edge quotes/.style={
      font=\small,
      auto
    }
  ]
  
  \node[initial] (q0) {$q_0$};
  
  \node[state] (q10) [above right=0.5cm and 2.5cm of q0] {$q_1^0$};
  \node[state] (q11) [below right=0.5cm and 2.5cm of q0] {$q_1^1$};
  
  \node[state] (q20) [right=2.5cm of q10] {$q_2^0$};
  \node[state] (q21) [right=2.5cm of q11] {$q_2^1$};
  
  \node[accepting] (qbot) [below right=0.5cm and 2.5cm of q20] {$q_3$};
  
  \draw[thick, ChapterAccent, <-] (q0) -- ++(-0.8,0);
  
  \path[->]
    (q0) edge node[above left, pos=0.4] {$\emptyset$} (q10)
    (q0) edge node[below left, pos=0.4] {$\{h\}$} (q11)
    (q10) edge node[above] {$\{r\}$} (q20)
    (q11) edge node[below] {$\{h,r\}$} (q21)
    (q20) edge node[above, pos=0.4] {$\emptyset$} (qbot)
    (q21) edge node[below, pos=0.4, xshift=0.4cm] {$\{h,s\}$} (qbot)
    (qbot) edge[loop right] node[right] {$\emptyset$} (qbot);
    
\end{tikzpicture}
\end{center}

The language $L(\mathcal{S})$ contains two words, namely $\emptyset \{r\} \emptyset \emptyset \ldots$ (when $h=0$) and $\{h\} \{h,r\} \{h,s\} \emptyset \emptyset \ldots$ (when $h=1$).
Agent~$a$ observes only $\{r\}$ at position 2 in both cases, while agent~$b$ observes $\{s\}$ at position 3 iff $h=1$.

Recall that \[\KnowsWhether_b h := \Knows_b h \lor \Knows_b \neg h.\] Consider
\[
\varphi = \tmp{X}\tmp{X}\bigl(\Knows_a\KnowsWhether_b h \land \neg \KnowsWhether_a h\bigr).
\]
At position 3, agent~$b$'s observation determines $h$, so $\KnowsWhether_b h$ holds.
Agent~$a$ has seen only $\{r\}$ (identical in both runs), so $\neg \KnowsWhether_a h$ holds.
However, agent~$a$ knows the system structure and knows that agent~$b$'s observation suffices to determine $h$, hence $\Knows_a \KnowsWhether_b h$ holds.
Thus $\mathcal{S} \models \varphi$.
\end{example}

%
%
%
%

So far, formulas have been interpreted at a concrete execution and a
reference position.
From the perspective of an agent, however, the actual execution is not
directly accessible.
Instead, the agent only has access to the finite observation it has
made so far.

We therefore provide a satisfaction relation that is defined relative to an
observation rather than to a single execution.
Intuitively, a formula holds under an observation if it holds for all
executions of the system that are compatible with that observation.

\begin{definition}[Satisfaction under observation]
Let $\mathcal S$ be a system, $a\in Ag$, and $\phi$ a formula.
Let $u = \sigma_1 \ldots \sigma_k \in \Obs_a(\mathcal S)$.
We write
\[
\mathcal S, u \models_a \phi
\]
if, for all $w\in L(\mathcal S)$ such that $\obs_a(w_{\le k})=u$, we have
$
\mathcal{S},w,k\models \phi .
$
\end{definition}

\begin{proposition}\label{prop:models-a-iff-K}
Let $\mathcal S$ be a system, $a\in Ag$, and $\phi$ be an \LTL formula.
Let $w\in L(\mathcal S)$ and $k\ge 1$, and set $u = \obs_a(w_{\le k})$.
Then,
\[
\mathcal S, u \models_a \phi
\quad\Longleftrightarrow\quad
\mathcal S, w, k \models \Knows_a \phi\, .
\]
\end{proposition}

\begin{proof}
Let $u=\obs_a(w_{\le k})$.

($\Longrightarrow$)
Assume $\mathcal S,u\models_a\phi$.
Let $w'\in L(\mathcal S)$ be such that $(w,k)\sim_a (w',k)$.
By definition of indistinguishability (perfect recall),
\[
(w,k)\sim_a (w',k)\quad\Longleftrightarrow\quad
\obs_a(w'_{\le k})=\obs_a(w_{\le k})=u.
\]
Hence $\obs_a(w'_{\le k})=u$, and by the definition of $\models_a$, we obtain
$\mathcal S,w',k\models \phi$.
Since this holds for all such $w'$, we conclude $\mathcal S,w,k\models \Knows_a\phi$.

($\Leftarrow$)
Assume $\mathcal S,w,k\models \Knows_a\phi$.
Let $w'\in L(\mathcal S)$ with $\obs_a(w'_{\le k})=u$.
Then $(w,k)\sim_a (w',k)$, again by the definition of indistinguishability.
By the semantics of knowledge, this implies $\mathcal S,w',k\models \phi$.
Since this holds for every $w'$ with $\obs_a(w'_{\le k})=u$, we conclude
$\mathcal S,u\models_a\phi$.
\end{proof}

\section{The Model-Checking Problem}

We now turn to the algorithmic problem of deciding whether a system
satisfies a given epistemic temporal specification.
It serves as a basis for the construction
of monitors and verification procedures.

\begin{definition}[Model-checking problem]\label{def:model-checking}
Given a Büchi automaton $\mathcal{S} = (Q, \iota, \Trans, F)$ 
and an \LTL formula $\varphi$, the \emph{model-checking problem} asks 
whether
\[
  \mathcal{S} \models \varphi,
\]
that is, whether for every $w \in L(\cal S)$,
we have $\mathcal S,w,1 \models \varphi$.
\end{definition}

Model-checking problems are usually solved by translating the negation of $\phi$ into
another Büchi automaton and checking the product for emptiness.
Due to the knowledge operator, unlike the classical construction, 
this Büchi automaton $\mathcal{A}_{\neg \phi}^{\mathcal{S}}$ also depends on the given system $\mathcal{S}$.

To make this explicit, we introduce a transducer that
combines the behavior of the system with the evaluation of the formula
along its executions.
That is, it maps each position along an execution to the truth value of the formula at that position.

\begin{definition}[System-formula transducer]\label{def:automaton-construction}
Let
$\mathcal{S}$ be a Büchi automaton over $\Sigma$ (representing a system) and $\varphi$ be an \LTL formula.
An $(\mathcal{S}, \varphi)$-\emph{transducer} consists of
\begin{itemize}
\item a Büchi automaton $\mathcal{A}^{\mathcal{S}}_{\varphi} = (Q_\phi, \iota_\phi, \Trans_\phi, F_\phi)$ over $\Sigma$, and
\item a mapping $\gamma_\phi\colon Q_\phi \to \{0, 1\}$
\end{itemize}
such that
$L(\mathcal{A}^{\mathcal{S}}_{\varphi}) = L(\cal{S})$ and, for all $w \in L(\cal{S})$,
all accepting runs
$\pi = q_0 \xrightarrow{\sigma_1} q_1 \xrightarrow{\sigma_2} q_2 \xrightarrow{\sigma_3} \ldots$ 
such that $\trace(\pi) = w$, and all $k \in \Npos$, we have
\[
\gamma_\phi(q_k) = 1
\quad\Longleftrightarrow\quad \mathcal{S}, w, k \models \varphi.
\]
Intuitively, $\gamma_\varphi$ marks the positions that satisfy \(\varphi\).
\end{definition}

\begin{theorem}[Transducer construction]\label{thm:automaton-construction}
For a Büchi automaton $\mathcal{S}$ and an \LTL formula $\varphi$, we can 
effectively construct an $(\mathcal{S}, \varphi)$-transducer.
\end{theorem}

In the construction, we will actually make use of a \emph{monitor}:

\begin{definition}[Knowledge monitor]\label{def:monitor-knowledge}
Let $\mathcal S$ be a Büchi automaton over $\Sigma$, an agent $a\in Ag$, and
$\phi$ be an \LTL formula. An $(\mathcal{S}, \phi)$-\emph{knowledge monitor for} $a$
is a DFA $\mathcal K_{\phi}^{\mathcal S,a}$ over $\Sigma_a$
such that, for all finite words $u \in \Obs_a(\mathcal S) \subseteq \Sigma_a^+$,
\[
u \in L(\mathcal K_{\phi}^{\mathcal S,a})
\quad\Longleftrightarrow\quad 
\mathcal{S}, u \models_a \phi.
\]
\end{definition}

\begin{theorem}[Monitor construction]\label{thm:monitor-knowledge}
Given a Büchi automaton over $\Sigma$, an agent $a\in Ag$, and
an \LTL formula $\phi$, we can effectively construct an
$(\mathcal{S}, \phi)$-knowledge monitor for $a$.
\end{theorem}

The rest of the section is devoted to the simultaneous proof of Theorems~\ref{thm:automaton-construction}
and \ref{thm:monitor-knowledge}.
We construct $\mathcal{A}^{\mathcal{S}}_{\varphi}$ and $\gamma_\varphi$ by structural induction on $\varphi$.
We assume that $L(\mathcal{S}) \neq \emptyset$ because, otherwise, the construction is trivial.

\subsection*{Base case: $\varphi = p$ \normalfont{where} $p \in \AP$}

We have $\mathcal{S}, w, k \models p$ iff $p \in \sigma_k$.

Assume $\mathcal{S} = (Q_{\mathcal{S}}, \iota_{\mathcal{S}}, \Trans_{\mathcal{S}}, F_{\mathcal{S}})$.

We construct the Büchi automaton:
\[
\mathcal{A}^{\mathcal{S}}_{p} = (Q, \iota, \Trans, F)
\]
where:
\begin{itemize}
\item $Q = Q_{\mathcal{S}} \times \{0, 1\}$
\item $\iota = (\iota_{\mathcal{S}}, 0)$ (second component arbitrary)
\item $\Trans = \bigl\{((q, \theta), \sigma, (q', \theta')) \mid (q, \sigma, q') \in \Trans_{\mathcal{S}} \text{ and } \theta' = 1 \Longleftrightarrow p \in \sigma\bigr\}$
\item $F = F_{\mathcal{S}} \times \{0, 1\}$
\end{itemize}
The output function is given by $\gamma_p((q,\theta)) = \theta$.

The automaton tracks $\mathcal{S}$ in the first component and stores in the second 
component whether $p$ holds.

\subsection*{Negation: $\varphi = \neg\psi$}

Assume we have constructed:
\[
\mathcal{A}^{\mathcal{S}}_{\psi} = (Q_{\psi}, \iota_{\psi}, \Trans_{\psi}, F_{\psi})
\]
with output function $\gamma_\psi\colon Q_\psi \to \{0, 1\}$.
We have $\mathcal{S}, w, k \models \neg\psi$ iff $\mathcal{S}, w, k \not\models \psi$.

We construct:
\[
\mathcal{A}^{\mathcal{S}}_{\neg\psi} = (Q_{\psi}, \iota_{\psi}, \Trans_{\psi}, F_{\psi})
\]
with output function $\gamma_{\neg\psi}(q) = 1 - \gamma_\psi(q)$.

The automaton is identical to $\mathcal{A}^{\mathcal{S}}_{\psi}$; only the output 
function is flipped.

\subsection*{Conjunction: $\varphi = \psi_1 \land \psi_2$}

Assume we have constructed:
\begin{align*}
\mathcal{A}^{\mathcal{S}}_{\psi_1} &= (Q_1, \iota_1, \Trans_1, F_1) \text{ with } \gamma_{\psi_1} : Q_1 \to \{0,1\} \\
\mathcal{A}^{\mathcal{S}}_{\psi_2} &= (Q_2, \iota_2, \Trans_2, F_2) \text{ with } \gamma_{\psi_2} : Q_2 \to \{0,1\}
\end{align*}

We have $\mathcal{S}, w, k \models \psi_1 \land \psi_2$ iff $\mathcal{S}, w, k \models \psi_1$ 
and $\mathcal{S}, w, k \models \psi_2$.

We first construct the generalized Büchi automaton
\[
\mathcal{A} = (Q, \iota, \Trans, \mathcal{F})
\]
where:
\begin{itemize}
\item $Q = Q_1 \times Q_2$
\item $\iota = (\iota_1, \iota_2)$
\item $\Trans = \bigl\{((q_1, q_2), \sigma, (q'_1, q'_2)) \mid (q_1, \sigma, q'_1) \in \Trans_1 \text{ and } (q_2, \sigma, q'_2) \in \Trans_2\bigr\}$
\item $\mathcal{F} = \{F_1 \times Q_2,\, Q_1 \times F_2\}$
\end{itemize}

The output function is given by $\gamma((q_1, q_2)) = \gamma_{\psi_1}(q_1) \land \gamma_{\psi_2}(q_2)$.

We then convert this GBA to a standard Büchi automaton to obtain 
$\mathcal{A}^{\mathcal{S}}_{\psi_1 \land \psi_2}$, extending the output function 
accordingly during the conversion to obtain $\gamma_{\psi_1 \land \psi_2}$.

\subsection*{Strict since: $\varphi=\psi_1 \sSince\psi_2$}

Assume we have constructed
\[
\mathcal A^{\mathcal S}_{\psi_i}=(Q_i,\iota_i,\Trans_i,F_i)
\]
with output maps $\gamma_{\psi_i}:Q_i\to\{0,1\}$ for $i\in\{1,2\}$.

We construct a generalized B\"uchi automaton
\[
\mathcal{A}=(Q,\iota,\Trans,\mathcal F)
\]
where:
\begin{itemize}
\item $Q = Q_1\times Q_2\times\{\mathit{init},0,1\}$
\item $\iota = (\iota_1,\iota_2,\mathit{init})$
\item $\Trans$ contains $((q_1,q_2,\theta),\sigma,(q_1',q_2',\theta'))$ if
      $(q_1,\sigma,q_1')\in\Trans_1$, $(q_2,\sigma,q_2')\in\Trans_2$, and
      \[
      \theta' =
      \begin{cases}
      0
      & \text{if } \theta=\mathit{init}\\[0.4em]
      1
      & \text{if } \theta\in\{0,1\}\text{ and }
        \bigl(\gamma_{\psi_2}(q_2)=1 \ \text{or}\ (\gamma_{\psi_1}(q_1)=1 \text{ and } \theta=1)\bigr)\\[0.4em]
      0
      & \text{otherwise}
      \end{cases}
      \]
\end{itemize}
\paragraph{Intuition.}
After reading $\sigma_k$, the component states $q_1$ and $q_2$ correspond to
position $k$ (as in the base-case construction), and $\theta\in\{0,1\}$ stores the
truth value of $\psi_1\sSince\psi_2$ at position~$k$.
At the very first step, strict since is always false (no previous position),
so we force $\theta'=0$ when $\theta=\mathit{init}$.

For subsequent steps, the update implements the following recurrence:
\[
\mathcal S,w,k{+}1\models \psi_1\sSince\psi_2
\quad\Longleftrightarrow\quad
\bigl(\mathcal S,w,k\models \psi_2\bigr)
\ \ \text{or}\ \
\bigl(\mathcal S,w,k\models \psi_1 \ \text{and}\ \mathcal S,w,k\models \psi_1\sSince\psi_2\bigr)
\]
\paragraph{Acceptance condition.}
The acceptance condition only has to ensure that the components for
$\psi_1$ and $\psi_2$ are accepting on their own:
\[
\mathcal F =
\bigl\{\;
F_1\times Q_2\times\{0,1\},\;
Q_1\times F_2\times\{0,1\}
\;\bigr\}
\]
\paragraph{Output map.}
Accordingly, we define the output map by:
\[
\gamma((q_1,q_2,\theta))=
\begin{cases}
0 & \text{if } \theta=\mathit{init}\\
\theta & \text{if } \theta\in\{0,1\}
\end{cases}
\]
Finally, we transform the GBA $\mathcal{A}$
into an ordinary Büchi automaton to obtain $\mathcal{A}^{\mathcal{S}}_{\psi_1\sSince\psi_2}$,
and we accordingly adjust $\gamma$ to obtain $\gamma_{\psi_1\sSince\psi_2}$
as required.

\subsection*{Strict until: $\varphi=\psi_1\sUntil\psi_2$}

Assume we have constructed
\[
\mathcal A^{\mathcal S}_{\psi_i}=(Q_i,\iota_i,\Trans_i,F_i)
\]
with output maps $\gamma_{\psi_i}:Q_i\to\{0,1\}$ for $i\in\{1,2\}$.

Similarly to strict since, we construct a generalized B\"uchi automaton
\[
\mathcal{A}=(Q,\iota,\Trans,\mathcal F)
\]
where:
\begin{itemize}
\item $Q = Q_1\times Q_2\times\{\mathit{init}, 0, 1\}$
\item $\iota = (\iota_1,\iota_2,\mathit{init})$
\item $\Trans$ contains $((q_1,q_2,\theta),\sigma,(q_1',q_2',\theta'))$ if
$(q_1,\sigma,q_1')\in\Trans_1$, $(q_2,\sigma,q_2')\in\Trans_2$, and
the following hold:
\begin{itemize}
\item
$\theta' \in \{0,1\}$
\item
$
\theta \in \{0,1\} \;\implies\;
\bigl(\theta = 1 \;\;\Longleftrightarrow\;\;
\gamma_{\psi_2}(q_2')=1
\ \text{or}\
(\gamma_{\psi_1}(q_1')=1 \ \text{and}\ \theta'=1)\bigr)$
\end{itemize}
\end{itemize}

\paragraph{Intuition.}
The transition constraints implement the one-step unfolding of strict until.
A positive guess $\theta=1$ claims that $\psi_2$ will occur at some
strictly future position and that $\psi_1$ holds until then (starting from the next position). As long as $\psi_2$
does not hold, the automaton must therefore preserve $\theta=1$ and require
$\psi_1$ to hold at the next position.
A negative guess $\theta=0$ claims that no such future position exists. Such a
guess may persist as long as $\psi_2$ does not occur, but it is refuted as soon as
$\psi_2$ holds at the next position.

Formally, the transition constraints enforce the following semantic recurrence:
\[
\mathcal S,w,k\models \psi_1\sUntil\psi_2
\;\Longleftrightarrow\;
\bigl(\mathcal S,w,k{+}1\models \psi_2\bigr)
\ \ \text{or}\ \
\bigl(\mathcal S,w,k{+}1\models \psi_1 \ \text{and}\ \mathcal S,w,k{+}1\models
\psi_1\sUntil\psi_2\bigr)
\]

\paragraph{Acceptance condition.}
The acceptance condition has to ensure that positive guesses cannot remain unresolved
forever.
In particular, an accepting run cannot stay in states with $\theta=1$ infinitely
often without eventually visiting a position where $\psi_2$ holds.
Therefore, every accepting run resolves each positive guess by a strictly future
$\psi_2$-position.
This is achieved using
\[
\mathcal F =
\bigl\{\;
F_1\times Q_2\times\{0, 1\},\;
Q_1\times F_2\times\{0, 1\},\;
\{(q_1, q_2, \theta) \mid \theta=0 \text{ or } \gamma_{\psi_2}(q_2)=1\}
\;\bigr\}.
\]

\paragraph{Output map.}
Accordingly, we define the output map by:
\[
\gamma((q_1,q_2,\theta))=
\begin{cases}
0 & \text{if } \theta=\mathit{init}\\
\theta & \text{if } \theta\in\{0,1\}
\end{cases}
\]
Finally, we obtain the required
$\mathcal{A}^{\mathcal{S}}_{\psi_1\sUntil\psi_2}$
and the adjusted $\gamma_{\psi_1\sUntil\psi_2}$ transforming
$\mathcal{A}$ into a Büchi automaton.

\subsection*{Knowledge monitor:}

A knowledge monitor does not observe the actual events of the system, but
only the partial observations available to a given agent.
As a consequence, after each observation, the monitor cannot determine a
single current system state.
Instead, it maintains a \emph{belief set} consisting of all system states
that are consistent with the observations seen so far.
As the observation sequence grows, this belief set is updated accordingly.
The monitor evaluates a property by checking whether it holds in all
states of the current belief set.
In this way, the monitor can determine when a property is known to hold
and when it cannot yet be concluded from the available observations.

Assume we have constructed
\[
\mathcal{A}^{\mathcal{S}}_{\phi} = (Q_{\phi}, \iota_{\phi}, \Trans_{\phi}, F_{\phi})
\]
with output function $\gamma_\phi : Q_\phi \to \{0, 1\}$.

Without loss of generality, we assume that $\mathcal A^{\mathcal S}_{\phi}$ is
\emph{productive}, that is, for every state $q$ reachable from $\iota_\phi$,
there is an accepting $q$-run. Otherwise, $\mathcal A^{\mathcal S}_{\phi}$ can be
replaced by an equivalent productive automaton obtained by removing all
states from which there is no accepting run.

Given $a \in \Ag$,
we  construct a $(\mathcal{S}, \phi)$-knowledge monitor for $a$,
a DFA $\mathcal{K}_{\phi}^{{\cal S}, a}$ over $\Sigma_a$, as
\[
\mathcal{K}_{\phi}^{{\cal S}, a} = (Q, \iota, \delta, F)
\]
where:
\begin{itemize}
\item ${Q} = 2^{Q_\phi}$ (belief states)
\item ${\iota} = \{\iota_{\phi}\}$
\item ${\delta}$ is defined by
      \[
      \delta(B, \sigma) =
      \bigl\{\,q' \in Q_\phi \bigm| \exists q \in B\colon \exists \tau \in \Sigma\colon 
      (q, \tau, q') \in \Trans_{\phi} \text{ and } \sigma = \obs_a(\tau)\,\bigr\}
      \]
      (illustrated in Figure~\ref{fig:belief-construction})
\item ${F} = \{B \in 2^{Q_\phi} \mid \forall q \in B\colon \gamma_\phi(q) = 1\}$
\end{itemize}

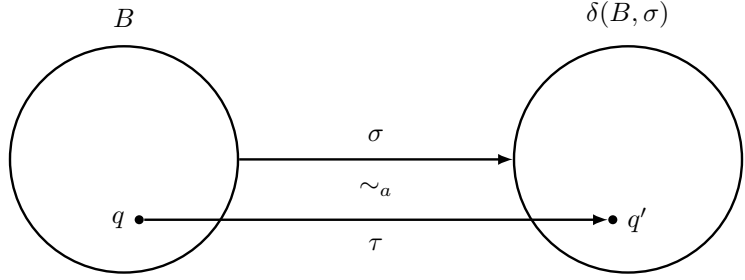
\begin{figure}[t]
\centering
\begin{tikzpicture}[
  >=latex,
  every node/.style={font=\small},
  belief/.style={
    draw,
    circle,
    minimum size=3cm,
    line width=0.9pt
  },
  statept/.style={
    circle,
    fill=black,
    inner sep=1.2pt
  }
]

\node[belief] (B) {};
\node[belief, right=3.6cm of B] (Bp) {};

\node[above=1mm of B] {$B$};
\node[above=1mm of Bp] {$\delta(B,\sigma)$};

\node[statept] (q)  at ($(B.center)+(0.2,-0.8)$) {};
\node[left=-0.1mm of q] {$q$};

\node[statept] (qp) at ($(Bp.center)+(-0.2,-0.8)$) {};
\node[right=-0.1mm of qp] {$q'$};

\draw[->, line width=0.9pt]
  (B.east) -- (Bp.west)
  node[midway, above=1mm] {$\sigma$};

\draw[->, line width=0.9pt]
  (q) -- (qp)
  node[midway, below=1mm] {$\tau$};

\node at ($(B.east)!0.5!(Bp.west)+(0,-0.4)$) {$\sim_a$};

\end{tikzpicture}
\caption{Belief update and representative state transition.}
\label{fig:belief-construction}
\end{figure}

\paragraph{Correctness of monitor.}

We have to show that, for all $u  \in \Obs_a(\mathcal S)$,
\begin{align}
\delta^\ast(\iota, u) \in F
\quad\Longleftrightarrow\quad
\mathcal S,u\models_a \phi .
\label{eq:knowledge}
\end{align}

\paragraph{Key invariant (belief semantics).}
For every $u \in \Sigma_a^\ast$,
\begin{align}
\begin{array}{rl}
\delta^\ast(\iota, u)
=
\Bigl\{\;
q\in Q_\phi \ \Bigm| & 
\!\!\exists v \in \Sigma^\ast\colon
u \sim_a v \text{ and}\\
& \!\!\exists \pi \text{ a finite run of }\mathcal A_\phi^{\mathcal S}\text{ over } v
\text{ with } \pi(|v|)=q
\;\Bigr\}.
\end{array}
\label{eq:belief}
\end{align}

\begin{proof}[Proof of (\ref{eq:belief})]
We proceed by induction on $k = |u|$.

\emph{\underline{Base $u = \varepsilon$.}}
($\subseteq$)
We have $\delta^\ast(\iota, u)=\{\iota_\phi\}$.
Taking $v = \varepsilon$ yields $u \sim_a v$.
Moreover, the trivial run $\pi=\iota_\phi$ is a finite run over $v$
with $\pi(0)=\iota_\phi$.

($\supseteq$)
Conversely, let $q$ belong to the right-hand side of (\ref{eq:belief}).
The unique $v$ such that $u \sim_a v$ is $v=\varepsilon$.
Over $\varepsilon$, the only finite run is the trivial run
$\pi=\iota_\phi$, hence $q=\pi(0)=\iota_\phi$.

\emph{\underline{Step $u \to u\sigma$.}}
Let $u \in \Sigma_a^\ast$ and $\sigma \in \Sigma_a$.
Let $k = |u|$.

($\subseteq$)
Let $q' \in \delta^\ast(\iota, u\sigma)$.
By definition of the belief update, there exist
$q\in \delta^\ast(\iota, u)$ and a letter $\tau\in\Sigma$ such that
\[
\sigma=\obs_a(\tau)
\quad\text{and}\quad
(q,\tau,q')\in\Trans_\phi .
\]
By the induction hypothesis for $u$,
there exist $v\in\Sigma^\ast$ with $u\sim_a v$ and a finite run
$\hat\pi$ over $v$ such that $\hat\pi(k)=q$.
Let $\pi$ be obtained from $\hat\pi$ by appending the transition
$q \xrightarrow{\tau} q'$.
Then, $\pi$ is a finite run over $v\tau$ and satisfies $\pi(k+1)=q'$.
Moreover, from $u\sim_a v$ and $\sigma = \obs_a(\tau)$, we get
$u\sigma\sim_a v\tau$.
Hence $q'$ belongs to the right-hand side.

($\supseteq$)
Conversely, assume $q'\in Q_\phi$ belongs to the right-hand side for $u\sigma$.
Then, there exist $v=\tau_1\ldots\tau_k\in\Sigma^\ast$ and $\tau \in \Sigma$ with
$u\sigma \sim_a v\tau$ and a finite run $\pi$ over $v\tau$ such that
$\pi(k+1)=q'$.
From $u\sigma \sim_a v\tau$, it follows that $u\sim_a v$, and the prefix
of $\pi$ up to position $k$ is a finite run over $v$ reaching
$\pi(k)$.
By the induction hypothesis, this implies $\pi(k)\in \delta^\ast(\iota, u)$.
Moreover, $\pi(k)\xrightarrow{\tau} q'$ is a transition of $\Trans_\phi$
and $\sigma=\obs_a(\tau)$.
Hence, by the definition of the belief update, we obtain $q'\in \delta^\ast(\iota, u\sigma)$.

This establishes (\ref{eq:belief}) for all $u \in \Sigma_a^\ast$.
\end{proof}


Let $u\in\Obs_a(\mathcal S)$ and let $k=|u|$. We show (\ref{eq:knowledge}).

Let $B_u = \delta^\ast(\iota, u)$.

($\Rightarrow$)
Assume $B_u \in F$.
By definition of $F$, this means
\[
\forall q\in B_u\colon \gamma_\phi(q)=1.
\]
Let $w\in L(\mathcal S)$ with $\obs_a(w_{\le k})=u$.
Fix any accepting run $\pi$ of $\mathcal A_\phi^{\mathcal S}$ over $w$ and set $q=\pi(k)$.
Then, the prefix of $\pi$ up to time $k$ is a finite run over $w_{\le k}$ reaching $q$.
Since $\obs_a(w_{\le k})=u$, by invariant~(\ref{eq:belief}) we have $q\in B_u$,
hence $\gamma_\phi(q)=1$.
By correctness of the inductive construction for $\phi$, this implies
$\mathcal S,w,k\models \phi$.
Since $w$ was arbitrary among executions compatible with $u$, we conclude
$\mathcal S,u\models_a \phi$.


($\Leftarrow$)
Assume $\mathcal S,u\models_a \phi$.
We show that $B_u \in F$, i.e.,
\[
\forall q \in B_u\colon\ \gamma_\phi(q)=1 .
\]
So fix an arbitrary $q \in B_u$.
By the invariant~(\ref{eq:belief}), there exist a word
$v \in \Sigma^\ast$ with $\obs_a(v)=u$ and a finite run $\pi$ of
$\mathcal A_\phi^{\mathcal S}$ over $v$ such that $\pi(k)=q$.
Since $\mathcal A_\phi^{\mathcal S}$ is productive, the finite run $\pi$
can be extended to an accepting infinite run $\pi'$ of
$\mathcal A_\phi^{\mathcal S}$.
Let $w = \trace(\pi')$.
Then $w \in L(\mathcal A_\phi^{\mathcal S}) = L(\mathcal S)$ and
$w_{\le k} = v$, hence $\obs_a(w_{\le k}) = u$.

From $\mathcal S,u\models_a \phi$ and $\obs_a(w_{\le k})=u$, it follows
by definition of $\models_a$ that $\mathcal S,w,k\models \phi$.
By correctness of the inductive construction for $\phi$, this implies
$\gamma_\phi(q)=1$.

Since $q \in B_u$ was arbitrary, we conclude that $B_u \in F$.

This proves (\ref{eq:knowledge}) and hence the correctness of the monitor.

\subsection*{Knowledge operator: $\varphi = \Knows_a\psi$}


%

Assume we have constructed the $(\mathcal{S}, \psi)$-knowledge monitor
$\mathcal{K}_{\psi}^{{\cal S}, a} = (Q_\mathcal{K}, \iota_\mathcal{K}, \delta_\mathcal{K}, F_\mathcal{K})$.
We construct the product of
$\mathcal{K}_{\psi}^{{\cal S}, a}$
with $\mathcal{S} = (Q_{\mathcal{S}}, \iota_{\mathcal{S}}, \Trans_{\mathcal{S}}, F_{\mathcal{S}})$:
\[
\mathcal{A}^{\mathcal{S}}_{\Knows_a\psi} = (Q, \iota, \Trans, F)
\]
where:
\begin{itemize}
\item $Q = Q_\mathcal{K} \times Q_{\mathcal{S}} = 2^{Q_\psi} \times Q_{\mathcal{S}}$
\item $\iota = (\iota_\mathcal{K}, \iota_{\mathcal{S}})$
\item $\Trans$ contains the triples $((B, q), \sigma, (B', q'))$ such that $\delta_\mathcal{K}(B, \obs_a(\sigma)) = B'$ 
      and $(q, \sigma, q') \in \Trans_{\mathcal{S}}$
\item $F = Q_\mathcal{K} \times F_{\mathcal{S}}$ (acceptance determined by $\mathcal{S}$ component)
\end{itemize}

The output function is defined as follows:
\[
\gamma_{\Knows_a\psi}((B,q)) =
\begin{cases}
1 & \text{if } B \in F_\mathcal{K}\\
0 & \text{otherwise}
\end{cases}
\]

Clearly, as acceptance is determined entirely by the second component, this ensures $L(\mathcal{A}^{\mathcal{S}}_{\Knows_a\psi}) = L(\mathcal{S})$.

\paragraph{Correctness of the output map.}
We show that for every word $w=\sigma_1\sigma_2\ldots\in L(\mathcal S)$, every
accepting run
\[
(B_0,q_0)\xrightarrow{\sigma_1}(B_1,q_1)\xrightarrow{\sigma_2}(B_2,q_2)\xrightarrow{\sigma_3}\cdots
\]
of $\mathcal A^{\mathcal S}_{\Knows_a\psi}$ over $w$, and every $k\in\Npos$, we have
\[
\gamma_{\Knows_a\psi}((B_k,q_k))=1
\quad\Longleftrightarrow\quad
\mathcal S,w,k\models \Knows_a\psi.
\]
By construction of $\Trans$, the first component is updated according to
\[
B_{k}=\delta_\mathcal{K}^\ast(\iota_\mathcal{K},\,\obs_a(w_{\le k})),
\]
since at step $i$ the monitor reads exactly the observation letter
$\obs_a(\sigma_i)$.

\smallskip\noindent
{($\Rightarrow$)}
Assume $\gamma_{\Knows_a\psi}((B_k,q_k))=1$.
Then $B_k\in F_\mathcal{K}$ by definition of $\gamma_{\Knows_a\psi}$.
Let $u=\obs_a(w_{\le k})$, which is contained in $\Obs_a(\mathcal S)$.
By the correctness of the knowledge-monitor construction, $u\in L(\mathcal K_\psi^{\mathcal S,a})$
implies
\[
\mathcal S,u\models_a \psi.
\]
By Proposition~\ref{prop:models-a-iff-K}, we obtain
$\mathcal S,w,k\models \Knows_a\psi$.


\smallskip\noindent
{($\Leftarrow$)}
Assume $\mathcal S,w,k\models \Knows_a\psi$.
Let $u=\obs_a(w_{\le k})$.
By Proposition~\ref{prop:models-a-iff-K},
we have 
$\mathcal S,u\models_a\psi$.
By correctness of the monitor construction, this implies
$u\in L(\mathcal K_\psi^{\mathcal S,a})$, hence
$\delta_\mathcal{K}^\ast(\iota_\mathcal{K},u)\in F_\mathcal{K}$.
Using $B_k=\delta_\mathcal{K}^\ast(\iota_\mathcal{K},u)$, we obtain $B_k\in F_\mathcal{K}$,
and therefore $\gamma_{\Knows_a\psi}((B_k,q_k))=1$.

\smallskip\noindent
This concludes the proof of correctness of the output map for all $k\in\Npos$.

\bigskip

Altogether, we have shown Theorems~\ref{thm:automaton-construction} and \ref{thm:monitor-knowledge}.

\subsection*{Solving the model-checking problem}

We now show how the model-checking problem can be solved using a
transducer.
Recall from Definition~\ref{def:model-checking} that, given a Büchi
automaton $\mathcal S = (Q,\iota,\Trans,F)$ and an \LTL formula $\varphi$,
the model-checking problem asks whether $\mathcal S \models \varphi$,
that is, whether $\mathcal S,w,1 \models \varphi$ holds for every
$w \in L(\mathcal S)$.

\begin{theorem}[Decidability of model checking]
\label{thm:decidability-mc}
The model-checking problem is decidable.
\end{theorem}

\begin{proof}
We reduce model checking to the emptiness problem for Büchi automata.
From $\mathcal S$ and $\phi$, we construct the $(\mathcal S,\phi)$-transducer
$\mathcal A^{\mathcal S}_\phi$ with output map $\gamma_\phi$.

Recall that a run of $\mathcal A^{\mathcal S}_\phi$ has the form
$q_0 \xrightarrow{\sigma_1} q_1 \xrightarrow{\sigma_2} q_2 \xrightarrow{\sigma_3} \cdots$.
The state $q_0$ is the initial state (before reading the word), and $q_1$
corresponds to the first letter $\sigma_1$.
By correctness of the transducer construction, for every accepting run $\pi$
over $w\in L(\mathcal S)$ we have
\[
\gamma_\phi(\pi(1)) = 1
\quad\Longleftrightarrow\quad
\mathcal S,w,1 \models \phi .
\]
Hence, we have the following equivalences:
\[
\begin{array}{rcl}
\mathcal S \models \phi
&\Longleftrightarrow&
\forall w\in L(\mathcal S)\colon \mathcal S,w,1\models\phi\\[1ex]
&\Longleftrightarrow&
\neg\exists w\in L(\mathcal S)\colon \mathcal S,w,1\not\models\phi\\[1ex]
&\Longleftrightarrow&
\neg\exists \text{ accepting run } \pi \text{ of }\mathcal A^{\mathcal S}_\phi
\text{ such that } \gamma_\phi(\pi(1))=0
\end{array}
\]

Thus, model checking reduces to checking whether
$\mathcal A^{\mathcal S}_\phi$ admits an accepting run whose state
$q$ at position 1 satisfies $\gamma_\phi(q)=0$.
This is an emptiness problem for Büchi automata, which can be solved using
standard graph-theoretic algorithms for finding reachable accepting
cycles.
\end{proof}

\begin{notebox}[Complexity]
Model checking for \pureLTL is PSPACE-complete.
In contrast, extending \pureLTL with epistemic operators as in \LTL
leads, in general, to a non-elementary complexity blow-up.
This increase is caused by the construction of the knowledge monitor,
which relies on a belief-state construction under perfect recall and
induces an iterated powerset blow-up.
\end{notebox}

\section{Bibliographic Notes}

Epistemic temporal logics under synchronous perfect recall have been
studied extensively.
Early decidability and complexity results were obtained by van der
Meyden and Shilov, who identified a non-elementary blow-up inherent to
reasoning about agents' information sets \cite{MeydenS99}.

In more recent work, Bozzelli, Maubert, and Murano gave a precise complexity analysis in terms
of the alternation depth of epistemic operators \cite{BozzelliMM19}.
Cohen and Lomuscio studied past-time epistemic temporal logics under
synchronous perfect recall and showed that, while the general problem is
non-elementary, substantially better bounds can be obtained for
restricted syntactic fragments \cite{CohenL10}.

Finkbeiner and Siber established
decidability of the logic considered in these lecture notes
via translations to first-order formalisms
for hyperproperties \cite{FinkbeinerS25}. They also considered additional operators,
such as counterfactuals, that go beyond the logic treated here.

For background on epistemic logic and the semantics of knowledge, we
refer to the textbook by Fagin, Halpern, Moses, and Vardi,
\emph{Reasoning about Knowledge} \cite{FHMV1995}.

\chapter{Diagnosability}

In this chapter, we illustrate how epistemic temporal logic can be used to
formalize and analyze classical runtime-verification properties.
In this chapter, we focus on diagnosability and related notions, which capture whether
agents can infer the occurrence or absence of faults from partial
observations.
Thanks to the epistemic perspective, these properties can be expressed
uniformly as logical formulas.
However, it should be noted that this uniform logical formulation does
not, in general, yield optimal complexity bounds for checking these
properties.

\section{Positive and Negative Diagnosability}

While classical definitions of (co-)diagnosability are phrased in terms of
explicit comparisons of observation projections (e.g., \cite{SampathSLST95,WangYL07}),
the epistemic formulation captures the same requirement at the semantic level,
expressing directly that an agent can infer the absence of an error after a
bounded delay.

In the following, the prefix P refers to \emph{positive} diagnosability, while
N refers to \emph{negative} diagnosability.

\begin{definition}[P- and N-diagnosability]
Let $\mathcal{S}$ be a Büchi automaton over $\Sigma$, $a \in \Ag$ be an agent, 
$e \in \AP \setminus \AP_a$ be an error proposition, and 
$\dbound \in \mathbb{N}$ be a maximal delay. We say that $\mathcal{S}$ is
\begin{itemize}
\item $\dbound$-\emph{P-diagnosable} (wrt.\ $a$ and $e$) if
\[
\mathcal{S} \models \tmp{G}\Bigl(e \to \tmp{X}^{\dbound} \Knows_a \tmp{P} e\Bigr),
\]
\item $\dbound$-\emph{N-diagnosable} (wrt.\ $a$ and $e$) if
\[
\mathcal{S} \models
\tmp{X}^{\dbound}\tmp{G}\Bigl(
  \tmp{H}\neg e \to
  \Knows_a\tmp{Y}^{\dbound}\tmp{H}\neg e
\Bigr).
\]
\end{itemize}
We say that $\mathcal{S}$ is \emph{P-diagnosable} (\emph{N-diagnosable})
if there is $\dbound \in \mathbb{N}$ such that it is
$\dbound$-P-diagnosable ($\dbound$-N-diagnosable, respectively).

Finally, we call $\mathcal{S}$ \emph{unbounded-P-diagnosable} (wrt.\ $a$ and $e$) if
\[
\mathcal{S} \models \tmp{G}\Bigl(e \to \tmp{F} \Knows_a \tmp{P} e\Bigr).
\]
\end{definition}

Intuitively, positive diagnosability means that whenever an error occurs,
agent~$a$ knows within $\dbound$ steps that an error has occurred in the
past.
In unbounded P-diagnosability, no bound is imposed on the delay before
such detection.

Dually, negative diagnosability means that after an initial delay of
$\dbound$ time steps, whenever the system has remained fault-free,
agent~$a$ can conclude that the system was fault-free up to $\dbound$
steps in the past.

\begin{example}[Positive- and negative-diagnosability]
We revisit Example~\ref{ex:intro-diagnosis}, where
$\mathit{AP}=\{p,r,e\}$, $a \in \Ag$ is the agent in question, and
$e$ is the unobservable error proposition.

Consider again four Büchi automata
$\mathcal S_1$--$\mathcal S_4$
such that:
\begin{align*}
L(\mathcal S_1) &= \{p\}^\omega + \{p\}^\ast \{p,e\}\{p\}\{r\}^\omega\\[1ex]
L(\mathcal S_2) &= \{p\}^\omega + \{p\}^\ast \{p,e\}\{p\}^\omega\\[1ex]
L(\mathcal S_3) &= \{p\}^\omega + \{p\}^\ast \{e\}\{p\}^\omega\\[1ex]
L(\mathcal S_4) &= \{p\}^\omega + \{p\}^\ast \{p,e\}\{p\}^\ast\{r\}^\omega
\end{align*}
Then, the following hold:
\begin{itemize}
\item $\mathcal S_1$ is both $2$-P-diagnosable and $2$-N-diagnosable.
\item $\mathcal S_2$ is neither P-diagnosable nor N-diagnosable.
\item $\mathcal S_3$ is both $0$-P-diagnosable and $0$-N-diagnosable.
\item $\mathcal S_4$ is neither P-diagnosable nor N-diagnosable; however, it is unbounded-P-diagnosable.
\end{itemize}
\end{example}

It is not a coincidence that positive- and negative-diagnosability coincide
on this example.
In fact, the two notions are equivalent, reflecting the duality between
detecting the occurrence and the absence of errors under partial
observation.
This symmetry is a general phenomenon in the one-agent setting.
The next proposition makes this precise in a uniform way.

\begin{proposition}
Fix $D\in\mathbb N$.
Let $\mathcal{S}$ be a system and fix an agent and an error proposition. Then,
\[
\mathcal{S} \text{ is $D$-P-diagnosable} \quad\Longleftrightarrow\quad \mathcal{S} \text{ is $D$-N-diagnosable}.
\]
\end{proposition}

\begin{proof}
Fix $D\in\mathbb N$. Let
\[
\text{(P$_D$)}\quad
\mathcal{S} \models \tmp{G}\Bigl(e \to \tmp{X}^{D} \Knows_a \tmp{P} e\Bigr)
\qquad
\text{(N$_D$)}\quad
\mathcal{S} \models
\tmp{X}^{D}\tmp{G}\Bigl(
  \tmp{H}\neg e \to
  \Knows_a\tmp{Y}^{D}\tmp{H}\neg e
\Bigr).
\]
We show (P$_D$)\,$\Rightarrow$\,(N$_D$) and (N$_D$)\,$\Rightarrow$\,(P$_D$).

\medskip\noindent
{(P$_D$)\,$\Rightarrow$\,(N$_D$).}
Assume (P$_D$). We prove (N$_D$), i.e.\ we show that for all $w\in L(\mathcal S)$ and all positions $k\ge 1$,
\[
\mathcal S,w,k{+}D\models \tmp{H}\neg e
\quad \Longrightarrow\quad
\mathcal S,w,k{+}D\models \Knows_a\tmp{Y}^{D}\tmp{H}\neg e.
\]
So fix $w\in L(\mathcal S)$ and $k\ge 1$ with $\mathcal S,w,k{+}D\models\tmp{H}\neg e$.
Let $w'\in L(\mathcal S)$ with $(w,k{+}D)\sim_a(w',k{+}D)$.
It suffices to prove $\mathcal S,w',k{+}D\models \tmp{Y}^{D}\tmp{H}\neg e$,
equivalently $\mathcal S,w',k\models\tmp{H}\neg e$.

Suppose towards a contradiction that $\mathcal S,w',k\not\models \tmp{H}\neg e$.
Then $\mathcal S,w',k\models\tmp{P}e$, so choose $j$ with $1\le j\le k$ such that
$\mathcal S,w',j\models e$.
By (P$_D$) applied to $w'$ at time $j$,
\[
\mathcal S,w',j{+}D \models \Knows_a\tmp{P}e.
\]
Since $(w,k{+}D)\sim_a(w',k{+}D)$ means
$\obs_a(w_{\le k{+}D})=\obs_a(w'_{\le k{+}D})$,
by truncation we obtain
$\obs_a(w_{\le j{+}D})=\obs_a(w'_{\le j{+}D})$, i.e.\ $(w,j{+}D)\sim_a(w',j{+}D)$.
By the semantics of knowledge, this yields
\[
\mathcal S,w,j{+}D \models \tmp{P}e.
\]
As $j{+}D\le k{+}D$, it follows that $\mathcal S,w,k{+}D \models \tmp{P}e$,
contradicting $\mathcal S,w,k{+}D\models \tmp{H}\neg e \equiv \neg\tmp{P}e$.
Hence $\mathcal S,w',k\models\tmp{H}\neg e$, as required.
Therefore $\mathcal S,w,k{+}D\models \Knows_a\tmp{Y}^{D}\tmp{H}\neg e$, proving (N$_D$).

\medskip\noindent
{(N$_D$)\,$\Rightarrow$\,(P$_D$).}
Assume (N$_D$). We prove (P$_D$), i.e.\ for all $w\in L(\mathcal S)$ and all $k\ge 1$,
\[
\mathcal S,w,k\models e
\quad \Longrightarrow\quad
\mathcal S,w,k{+}D\models \Knows_a\tmp{P}e.
\]
So fix $w\in L(\mathcal S)$ and $k\ge 1$ with $\mathcal S,w,k\models e$.
Let $w'\in L(\mathcal S)$ with $(w,k{+}D)\sim_a(w',k{+}D)$.
We show $\mathcal S,w',k{+}D\models\tmp{P}e$.

Suppose towards a contradiction that $\mathcal S,w',k{+}D\not\models\tmp{P}e$.
Then $\mathcal S,w',k{+}D\models \tmp{H}\neg e$.
Apply (N$_D$) to $w'$ at time $k$:
\[
\mathcal S,w',k{+}D\models \Knows_a\tmp{Y}^{D}\tmp{H}\neg e.
\]
Since $(w,k{+}D)\sim_a(w',k{+}D)$, the semantics of knowledge yields
\[
\mathcal S,w,k{+}D\models \tmp{Y}^{D}\tmp{H}\neg e,
\]
i.e.\ $\mathcal S,w,k\models \tmp{H}\neg e$, contradicting $\mathcal S,w,k\models e$.
Hence $\mathcal S,w',k{+}D\models\tmp{P}e$.

Since $w'$ was arbitrary, $\mathcal S,w,k{+}D\models \Knows_a\tmp{P}e$.
Thus (P$_D$) holds.
\end{proof}

\begin{exercise}[Four alternative definitions of diagnosability]
Let $\mathcal S$ be a system, $a \in \Ag$ be an agent, and
$e \in \AP \setminus \AP_a$ be an unobservable error proposition.
Fix a delay bound $D \in \mathbb{N}$.

Consider the following alternative candidate formulas for diagnosability:
\begin{align*}
\phi_1 &=
\tmp{G}\Bigl(
  e \to
  \tmp{X}^{\dbound}\Knows_a\tmp{Y}^{\dbound}\tmp{P} e
\Bigr)\\[0.5ex]
\phi_2 &=
\tmp{G}\Bigl(
  \tmp{H}\neg e \to
  \tmp{X}^{\dbound}\Knows_a\tmp{Y}^{\dbound}\tmp{H}\neg e
\Bigr)\\[0.5ex]
\phi_3 &=
\tmp{X}^{\dbound}\tmp{G}\Bigl(
  \tmp{H}\neg e \to
  \tmp{Y}^{\dbound}\Knows_a\tmp{H}\neg e
\Bigr)\\[0.5ex]
\phi_4 &=
\tmp{G}\Bigl(
  \tmp{H}\neg e \to
  \tmp{X}^{\dbound}\Knows_a\tmp{H}\neg e
\Bigr)
\end{align*}
Give an informal interpretation of each formula and
explain why neither formula captures
diagnosability as we had defined it before.
Which of them are ``reasonable'' definitions.
\emph{Hint:} You may use
$\{p, e\}\{p\}^\omega + 
\{p\}\{p,e\}\{p\}^\omega$ to separate certain formulas.
\end{exercise}

We conclude this section by considering the algorithmic complexity of
diagnosability problems.
Although the definitions are given in terms of epistemic temporal
formulas (and are therefore decidable by
Theorem~\ref{thm:decidability-mc}), the corresponding decision problems
can be solved efficiently on finite-state systems.
We illustrate this for $\dbound$-P-diagnosability.

Fix a (system) Büchi automaton $\mathcal{S}=(Q, \iota, \Delta, F)$, an agent $a \in \Ag$, 
an error proposition $e \in \AP \setminus \AP_a$, and a bound $\dbound \in \mathbb{N}$.
We say that a word $u = \sigma_1 \ldots \sigma_n \in \Sigma^\ast$ is
\emph{faulty} if there exists $i \in \{1,\ldots,n\}$ such that
$e \in \sigma_i$.
If $u$ is not faulty, it is called \emph{fault-free}.

The algorithm is based on the following characterization of $\dbound$-P-diagnosability:

\begin{proposition}[Characterization of bounded P-diagnosability]
\label{prop:D-P-char}
The following statements are equivalent:
\begin{enumerate}
\item[(1)] $\mathcal{S}$ is \emph{not} $\dbound$-P-diagnosable (wrt.\ $a$ and $e$), i.e.
$
\mathcal{S} \not\models \tmp{G}\bigl(e \to \tmp{X}^{\dbound}\Knows_a\tmp{P}e\bigr).
$
\item[(2)] There exist two words $w, w'\in L(\mathcal S)$ and a position
$k \ge 1$ such that we have
(i)~$\mathcal S, w,k\models e$,
(ii)~$w'_{\le k+\dbound}$ is fault-free, and
(iii)~$\obs_a(w_{\le k+\dbound})=\obs_a(w'_{\le k+\dbound})$.
\end{enumerate}
\end{proposition}

\begin{proof}
We show $(1)\Leftrightarrow(2)$.

\medskip\noindent
$(1)\Rightarrow(2)$:
Assume
\[
\mathcal{S} \not\models \tmp{G}\bigl(e \to \tmp{X}^{\dbound}\Knows_a\tmp{P}e\bigr).
\]
Then there exist $w\in L(\mathcal S)$ and a position $k\ge 1$ such that
\[
\mathcal S,w,k\models e
\quad\text{and}\quad
\mathcal S,w,k+\dbound \not\models \Knows_a\tmp{P}e.
\]
According to the semantics of $\Knows_a$, there exists $w'\in L(\mathcal S)$ such that
\[
\obs_a(w_{\le k+\dbound})=\obs_a(w'_{\le k+\dbound})
\qquad\text{and}\qquad
\mathcal S,w',k+\dbound \not\models \tmp{P}e.
\]
The latter implies that $w'_{\le k+\dbound}$ is fault-free.
Thus $w$, $w'$, and $k$ satisfy the requirements in~(2).

\medskip\noindent
$(2)\Rightarrow(1)$:
Assume that there exist $w,w'\in L(\mathcal S)$ and $k\ge 1$ such that
$\mathcal S,w,k\models e$, the prefix $w'_{\le k+\dbound}$ is fault-free, and
$\obs_a(w_{\le k+\dbound})=\obs_a(w'_{\le k+\dbound})$.
Since $w'_{\le k+\dbound}$ is fault-free, we have
\[
\mathcal S,w',k+\dbound \not\models \tmp{P}e.
\]
Moreover, $\obs_a(w_{\le k+\dbound})=\obs_a(w'_{\le k+\dbound})$ means,
by the semantics of knowledge, that
\[
\mathcal S,w,k+\dbound \not\models \Knows_a\tmp{P}e.
\]
Together with $\mathcal S,w,k\models e$, this shows that
\[
\mathcal S \not\models \tmp{G}\bigl(e \to \tmp{X}^{\dbound}\Knows_a\tmp{P}e\bigr),
\]
i.e.\ $\mathcal S$ is not $\dbound$-P-diagnosable.
\end{proof}

\begin{theorem}[Deciding bounded P-diagnosability in polynomial time]
\label{thm:decide-D-P-diagnosability}
The following problem can be decided in polynomial time:

Given a Büchi automaton $\mathcal{S}$ over $\Sigma$, $\dbound \in \mathbb{N}$ (encoded in unary), $a \in \Ag$, and an unobservable error proposition $e$,
is $\mathcal{S}$ $\dbound$-P-diagnosable (wrt.\ $a$ and $e$)?
\end{theorem}

\begin{exercise}
Prove Theorem~\ref{thm:decide-D-P-diagnosability}
using Proposition~\ref{prop:D-P-char}.
\end{exercise}

%
%
%

\begin{notebox}[On unbounded-P-diagnosability]
Unbounded-P-diagnosability is not a standard notion in the classical
diagnosis literature, but it arises naturally in the logical formulation.
Unlike P-diagnosability, it does not guarantee a uniform bound on the delay
until detection and therefore provides weaker runtime guarantees.

In finite-state systems modeled as Kripke structures (or, equivalently,
automata where all states are accepting), unbounded-P-diagnosability and
P-diagnosability coincide: eventual detection implies detection within a
uniform bound.
This fact underlies the classical diagnosability algorithms.

In contrast, in the more general setting of Büchi automata, acceptance
conditions may rule out infinite indistinguishability without imposing a
uniform bound.
As a result, unbounded-P-diagnosability and P-diagnosability no longer
coincide.
\end{notebox}

\begin{notebox}[General error properties]
In contrast, the epistemic-logic framework allows error events to be specified by
arbitrary temporal formulas, so that one can reason about more
general classes of faulty behaviors.
\end{notebox}

\section{Diagnosers and Knowledge Monitors}

In classical diagnosis, one often starts with an abstract definition
of a diagnoser as a function that maps finite observation sequences to
diagnosis decisions.

\begin{definition}[Diagnoser]
Let $\mathcal S$ be a Büchi automaton, $a\in\Ag$ be an agent, 
$e\in\AP\setminus\AP_a$ be an unobservable error proposition, and
$D\in\mathbb N$.
A \emph{$D$-P-diagnoser for $\mathcal S$} (wrt.\ $a$ and $e$) is a
function
\[
f \colon \Sigma_a^\ast \to \{0,1\}
\]
such that the following conditions hold:
\begin{enumerate}
\item (\emph{Detection})
For all faulty $u \in \Sigma^+$ and all $v \in \Sigma^\ast$ such that
$|v| = D$ and $uv \in \Pref(L(\mathcal S))$, we have $f(\obs_a(uv))=1$.
\item (\emph{No false alarms})
For all fault-free $u\in\Pref(L(\mathcal S))$,
we have $f(\obs_a(u))=0$.
\end{enumerate}
\end{definition}

The first condition means that on every execution prefix $u$ in which a fault has
already occurred, the diagnoser outputs $1$ after $D$ further steps.
Note that since the condition holds for all faulty prefixes, detection
persists at all later times.
The second condition forbids false positives.

\medskip

We now relate this classical notion to knowledge monitors.

\begin{proposition}[Knowledge monitors as diagnosers]
\label{prop:kmasdiagnosers}
Let $\mathcal S$ be a Büchi automaton over $\Sigma$, let $a\in\Ag$, let
$e\in\AP\setminus\AP_a$, and let $D\in\mathbb N$.
If $\mathcal S$ is $D$-P-diagnosable (wrt.\ $a$ and $e$),
then the function
$f=\sem{\mathcal K_{\tmp{P}e}^{\mathcal S,a}}$
is a $D$-P-diagnoser for~$\mathcal S$.

Recall that $\mathcal K_{\tmp{P}e}^{\mathcal S,a}$
is an $(\mathcal S,\tmp{P}e)$-knowledge monitor
and that $f$ satisfies, for all $u \in \Obs_a(\mathcal{S})$:
\[
f(u) =
\begin{cases}
1 & \text{if } \mathcal S,u\models_a \tmp{P}e\\
0 & \text{otherwise}
\end{cases}
\]

\end{proposition}

\begin{proof}
We show the two properties of a $D$-P-diagnoser.

\emph{Detection.}
Let $u$ be faulty and let $v\in\Sigma^\ast$ with $|v|=D$ such that
$uv\in\Pref(L(\mathcal S))$.
Choose $w\in L(\mathcal S)$ extending $uv$, and let $k\le |u|$ be such that
$\mathcal S,w,k\models e$.
By $D$-P-diagnosability, we obtain
\[
\mathcal S,w,k{+}D \models \Knows_a \tmp{P}e.
\]
Since $k\le |u|$ and $|v|=D$, we have $k{+}D \le |uv|$.
By non-forgetting (Exercise~\ref{ex:knowledge-sanity}(2d)),
knowledge of the past formula $\tmp{P}e$ persists at later positions. Thus
\[
\mathcal S,w,|uv| \models \Knows_a \tmp{P}e.
\]
By Proposition~\ref{prop:models-a-iff-K}, this is equivalent to
\[
\mathcal S,\obs_a(uv) \models_a \tmp{P}e.
\]
By correctness of the knowledge monitor,
$f(\obs_a(uv))=1$.

\emph{No false alarms.}
Let $u\in\Pref(L(\mathcal S))$ be fault-free.
Then $\mathcal S,\obs_a(u)\not\models_a\tmp{P}e$,
hence $f(\obs_a(u))=0$.

\smallskip
Thus $f$ satisfies both conditions of a $D$-P-diagnoser.
\end{proof}

\begin{proposition}[Diagnosability vs. existence of a diagnoser]
\label{prop:DPdiag-iff-diagnoser}
Let $\mathcal S$ be a Büchi automaton over $\Sigma$, let $a\in\Ag$, let
$e\in\AP\setminus\AP_a$, and let $D\in\mathbb N$.
The following are equivalent:
\begin{enumerate}
\item $\mathcal S$ is $D$-P-diagnosable (wrt.\ $a$ and $e$), i.e.
\[
\mathcal S \models \tmp{G}\bigl(e \to \tmp{X}^{D}\Knows_a \tmp{P}e\bigr).
\]
\item There exists a $D$-P-diagnoser for $\mathcal S$ (wrt.\ $a$ and $e$).
\end{enumerate}
\end{proposition}

\begin{proof}
($1\Rightarrow 2$)
This is Proposition~\ref{prop:kmasdiagnosers}.

\smallskip\noindent
($2\Rightarrow 1$)
Let $f\colon\Sigma_a^\ast\to\{0,1\}$ be a $D$-P-diagnoser.
Let $w\in L(\mathcal S)$ and let $k\ge 1$ with $\mathcal S,w,k\models e$.
Set $u=w_{\le k}$ and $u_D=w_{\le k+D}$.
Then $u$ is faulty and $u_D\in\Pref(L(\mathcal S))$.

By the detection condition, we have
\[
f(\obs_a(u_D))=1.
\]
We show that $\mathcal S,w,k{+}D\models \Knows_a\tmp{P}e$.
Assume towards a contradiction that this is not the case.
Then there exists $w'\in L(\mathcal S)$ such that
\[
\obs_a(w_{\le k+D})=\obs_a(w'_{\le k+D})
\qquad\text{and}\qquad
\mathcal S,w',k{+}D \not\models \tmp{P}e.
\]
The second conjunct means that $w'_{\le k+D}$ is fault-free.
By the no-false-alarms condition, we obtain
\[
f(\obs_a(w'_{\le k+D}))=0.
\]
But $\obs_a(w'_{\le k+D})=\obs_a(w_{\le k+D})$, contradicting
$f(\obs_a(u_D))=1$.

Hence $\mathcal S,w,k{+}D\models \Knows_a\tmp{P}e$.
Since $w$ and $k$ were arbitrary, it follows that
\[
\mathcal S \models \tmp{G}\bigl(e \to \tmp{X}^{D}\Knows_a \tmp{P}e\bigr),
\]
so $\mathcal S$ is $D$-P-diagnosable.
\end{proof}

\begin{exercise}
Consider running
$\mathcal K_{\tmp{P}e}^{\mathcal S,a}$
and
$\mathcal K_{\phi}^{\mathcal S,a}$, where $\phi = \tmp{Y}^{\dbound}\tmp{H}\neg e$,
 in parallel on the
same observation stream.
\begin{enumerate}
\item Describe the possible combinations of verdicts and interpret them
epistemically. When can a specific combination happen (can the two diagnosers accept simultaneously)?
\item Explain why this yields a three-valued diagnosis semantics.
\end{enumerate}
\end{exercise}

\begin{exercise}[Bounded predictability]
Let $\Ag = \{a\}$, $\AP=\{p,e\}$, $\AP_a=\{p\}$, and $\Sigma=2^{\AP}$.
Thus, the observer sees only proposition $p$, while the error proposition
$e\in\AP\setminus\AP_a$ is unobservable.

Fix an arbitrary delay bound $D\in\mathbb{N}$. Recall that
\[
\tmp{F}^{\le D}e = e\lor \tmp{X}e\lor\ldots\lor\tmp{X}^{D}e.
\]
Consider the following candidate notions of \emph{bounded predictability}:
\begin{align*}
\phi_{1}^{D}  &= \tmp{G}(\tmp{X}^{D}e \;\rightarrow\; \Knows_a \tmp{X}^{D}e)\\[1ex]
\phi_{2}^{D} &= \tmp{G}(\tmp{F}^{\le D}e \;\rightarrow\; \Knows_a \tmp{F}^{\le D}e)\\[1ex]
\phi_{3}^{D}   &= \tmp{G}(\tmp{X}^{D}e \;\rightarrow\; \Knows_a \tmp{F}e)
\end{align*}

\begin{enumerate}[label={(\alph*)}]
\item Briefly discuss which of the three notions is, in your view, the most reasonable formalization of ``predictability with delay~$D$'', and why.
\item For fixed $D$, determine which logical implications between $\phi_{1}^{D}$, $\phi_{2}^{D}$, and $\phi_{3}^{D}$ hold in general, and justify your answer.
\item Now fix $D=1$. Give separating examples showing that none of the converse implications holds, and that no further implications hold in general.
\item For each
$i\in\{1,2,3\}$, consider the decision
problem with input $\mathcal S$ (a Büchi automaton over $\Sigma$) and $D\in\mathbb{N}$ (with $D$ encoded in unary):
decide whether $\mathcal S \models \phi_{i}^{D}$.
Show that each of these three decision problems is solvable in polynomial time.
A proof sketch in each case is enough.
\item Assume $\mathcal S \models \phi_{i}^{D}$ for a given system $\mathcal S$ and $i\in\{1,2,3\}$.
In the style of the knowledge monitor used for diagnosis, describe a suitable
observation-based predictor for this property.
What will it read as input, what will it output after each observation
prefix, and what correctness guarantee does $\phi_{i}^{D}$ provide?
\end{enumerate}
\end{exercise}

\section{Decentralized Diagnosis}

We now turn to a decentralized setting, where several agents observe
the system independently and no single agent has access to all
observations.
In this context, diagnosability requires that the occurrence or absence of an error can be inferred by at least one agent, based solely on its own observations.
This leads to the notions of positive and negative \emph{codiagnosability}.

\begin{definition}[P- and N-codiagnosability]
Let $\mathcal{S}$ be a Büchi automaton over $\Sigma$, 
$e \in \AP \setminus \bigcup_{a \in \Ag} \AP_a$ be an error proposition, and 
$\dbound \in \mathbb{N}$. We say that $\mathcal{S}$ is
\begin{itemize}
\item $\dbound$-\emph{P-codiagnosable} (wrt.\ $e$) if
\[
\mathcal{S} \models \tmp{G}\Bigl(e \to \tmp{X}^{\dbound} \bigvee_{a \in Ag} \Knows_a \tmp{P} e\Bigr),
\]
\item $\dbound$-\emph{N-codiagnosable} (wrt.\ $e$) if
\[
\mathcal{S} \models
\tmp{X}^{D}\tmp{G}\Bigl(
  \tmp{H}\neg e \to
  \bigvee_{a \in \Ag}\Knows_a\tmp{Y}^{D}\tmp{H}\neg e
\Bigr).
\]
\end{itemize}
Accordingly, we say that $\mathcal{S}$ is \emph{P-codiagnosable} (\emph{N-codiagnosable})
if there is $\dbound \in \mathbb{N}$ such that it is
$\dbound$-P-codiagnosable ($\dbound$-N-codiagnosable, respectively).
\end{definition}

Unlike in the one-agent case, the decentralized versions of diagnosability
are incomparable:

\begin{example}[P- and N-codiagnosability are incomparable]
Let $\Ag = \{a_1, a_2\}$,
$\mathit{AP}=\{p_1, p_2, e\}$,
$\mathit{AP}_{a_1}=\{p_1\}$, and
$\mathit{AP}_{a_2}=\{p_2\}$.
In particular, $e$ is the unobservable error proposition.

Consider Büchi automata $\mathcal S_1$, $\mathcal S_2$, and $\mathcal S_3$ such that:
\begin{align*}
L(\mathcal S_1) &= \emptyset^\ast\{e\}(\{p_1\} + \{p_2\}) \emptyset^\omega\\[0.5ex]
L(\mathcal S_2) &= (\{p_1\} + \{p_2\} + \{e\}) \emptyset^\omega\\[0.5ex]
L(\mathcal S_3) &= (\{p_1,p_2,e\} + \{p_1,p_2\}) \emptyset^\omega
\end{align*}
Then,
\begin{itemize}
\item $\mathcal S_1$ is $1$-P-codiagnosable but not N-codiagnosable,
\item $\mathcal S_2$ is $0$-N-codiagnosable but not P-codiagnosable,
\item $\mathcal S_3$ is neither P-codiagnosable nor N-codiagnosable.
\end{itemize}
\end{example}

\paragraph{Knowledge monitors as decentralized diagnosers.}

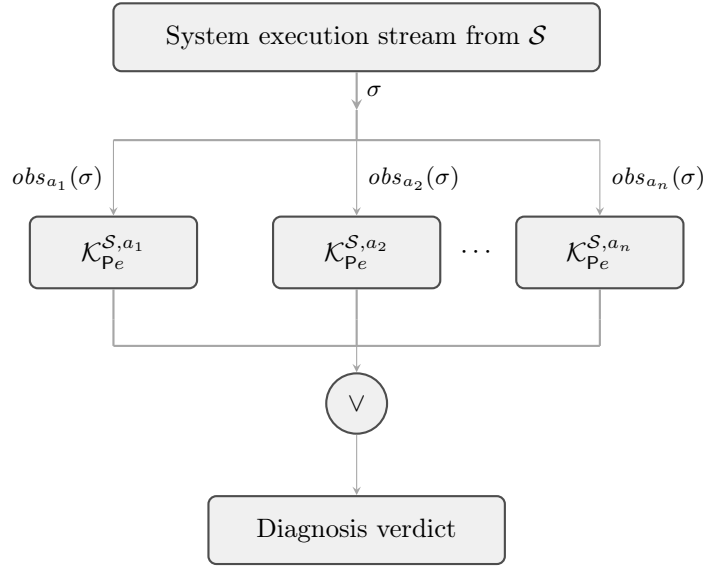
\begin{figure}[t]
\centering
\begin{tikzpicture}[
  node distance=1.4cm and 2.6cm,
  every node/.style={font=\small},
stream/.style={
  draw=black!70,
  rounded corners=3pt,
  minimum height=0.9cm,
  fill=black!6,
  line width=0.8pt
},
monitor/.style={
  draw=black!70,
  rectangle,
  rounded corners=3pt,
  minimum width=2.2cm,
  minimum height=0.95cm,
  fill=black!6,
  line width=0.8pt
},
fusion/.style={
  draw=black!70,
  circle,
  minimum width=0.8cm,
  minimum height=0.8cm,
  fill=black!6,
  line width=0.8pt
},
>=latex
]
\node (source) [stream, minimum width=6.4cm] {System execution stream from $\mathcal{S}$};

\coordinate (split) at ($(source.south)+(0,-0.5)$);
\node[right, font=\footnotesize] at ($(source.south)!0.5!(split)$) {$\sigma$};

\coordinate (hsplit) at ($(split)+(0,-0.4)$);

\coordinate (drop1) at ($(hsplit)+(-3.2,0)$);
\coordinate (drop2) at (hsplit);
\coordinate (drop3) at ($(hsplit)+(3.2,0)$);

\draw[-{stealth[scale=1.3]}, thick, draw=gray!70] (source.south) -- (split);
\draw[thick, draw=gray!70] (split) -- (hsplit);
\draw[thick, draw=gray!70] (drop1) -- (hsplit) -- (drop3);

\node (m1) [monitor, below=1.0cm of drop1] {$\mathcal{K}_{\tmp{P}e}^{\mathcal{S},a_1}$};
\node (m2) [monitor, below=1.0cm of drop2] {$\mathcal{K}_{\tmp{P}e}^{\mathcal{S},a_2}$};
\node (m3) [monitor, below=1.0cm of drop3] {$\mathcal{K}_{\tmp{P}e}^{\mathcal{S},a_n}$};

\node at ($(m2)!0.5!(m3)$) {$\cdots$};

\draw[-{stealth[scale=1.3]}, draw=gray!70] (drop1) -- node[midway, left, font=\footnotesize] {$\mathit{obs}_{a_1}(\sigma)$} (m1.north);
\draw[-{stealth[scale=1.3]}, draw=gray!70] (drop2) -- node[midway, right, font=\footnotesize] {$\mathit{obs}_{a_2}(\sigma)$} (m2.north);
\draw[-{stealth[scale=1.3]}, draw=gray!70] (drop3) -- node[midway, right, font=\footnotesize] {$\mathit{obs}_{a_n}(\sigma)$} (m3.north);

\coordinate (conv1) at ($(m1.south)+(0,-0.4)$);
\coordinate (conv2) at ($(m2.south)+(0,-0.4)$);
\coordinate (conv3) at ($(m3.south)+(0,-0.4)$);

\coordinate (hconv) at ($(conv2)+(0,-0.35)$);
\coordinate (hconv1) at ($(hconv)+(-3.2,0)$);
\coordinate (hconv3) at ($(hconv)+(3.2,0)$);

\node (or) [fusion, below=0.35cm of hconv] {$\lor$};

\draw[thick, draw=gray!70] (m1.south) -- (conv1);
\draw[thick, draw=gray!70] (m2.south) -- (conv2);
\draw[thick, draw=gray!70] (m3.south) -- (conv3);
\draw[thick, draw=gray!70] (conv1) -- (hconv1);
\draw[thick, draw=gray!70] (conv2) -- (hconv);
\draw[thick, draw=gray!70] (conv3) -- (hconv3);
\draw[thick, draw=gray!70] (hconv1) -- (hconv) -- (hconv3);
\draw[-{stealth[scale=1.3]}, draw=gray!70] (hconv) -- (or);

\node (out) [stream, below=0.8cm of or, minimum width=3.9cm] {Diagnosis verdict};
\draw[-{stealth[scale=1.3]}, draw=gray!70] (or) -- (out);
\end{tikzpicture}
\caption{Decentralized diagnosis via local knowledge monitors (positive case). Each global event $\sigma$
is locally filtered into $\mathit{obs}_{a_i}(\sigma)$ and processed by the corresponding monitor.}
\label{fig:decentralized-diagnoser}
\end{figure}

In the decentralized setting, each agent $a$ receives the local view $\obs_{a}(\sigma) \in \Sigma_a$ of every
global event $\sigma \in \Sigma$ and runs its own monitor.
Figure~\ref{fig:decentralized-diagnoser} illustrates this architecture for the
positive case: the global stream is projected to each agent, processed by a
local $(\mathcal S,\tmp{P}e)$-knowledge monitor, and the individual verdicts are aggregated by
disjunction.

Formally, for each agent $a \in \Ag$ we consider an $(\mathcal S,\tmp{P}e)$-knowledge monitor
for $a$, and we regard the decentralized diagnoser as the collection of these local
monitors together with the rule that a diagnosis is produced as soon as at least
one monitor accepts.
This operational view corresponds to the disjunctive knowledge condition
$\bigvee_{a \in \Ag}\Knows_a\tmp{P}e$ used in the definition of
P-codiagnosability.

\section{Bibliographic Notes}

The notion of diagnosability for discrete-event systems was introduced by
Sampath et al.~\cite{SampathSLST95}, who defined what is now commonly
referred to as F-diagnosability.
In our logical framework, this notion corresponds to P-diagnosability
and requires that the occurrence of an unobservable fault can be inferred
after a bounded delay.
The complementary perspective, sometimes called NF-diagnosability and
focusing on the inference of fault absence, was studied later by Wang,
Yoo, and Lafortune~\cite{WangYL07}, who showed that the two notions are
equivalent in the centralized setting.

Diagnosability in decentralized architectures was investigated by
Debouk, Lafortune, and Teneketzis~\cite{DeboukLT00}, leading to the notion
of codiagnosability, where at least one agent must be able to make a
correct diagnosis.
In contrast to the centralized case, the equivalence between positive
and negative diagnosability does not extend to decentralized settings,
where positive and negative codiagnosability are incomparable.

Polynomial-time decision procedures for diagnosability and
codiagnosability were established in the above works.

\chapter{Opacity}

As we have seen in the introduction, \emph{opacity} is dual to
diagnosability.
While diagnosability aims at eventually inferring the occurrence of a
fault, opacity requires that an observer can never infer that a
designated secret has occurred.
Equivalently, whenever a secret behavior takes place, there exists a
non-secret behavior that is observationally indistinguishable to the
observer.

\section{Definition of Opacity}

Let us define opacity within our logical framework:

\begin{definition}[Opacity]
\label{def:opacity}
Fix an agent $a\in\Ag$ and let $s \in \AP \setminus \AP_a$ be a secret.
A Büchi automaton $\mathcal{S}$ over $\Sigma$ representing a system
is \emph{opaque} (wrt.\ $a$ and $s$) if
\[
\mathcal{S} \models
\tmp{G}
  \neg \Knows_a \tmp{P}s.
\]
Thus, at every point in every execution, agent $a$ cannot know that the
secret has occurred.
\end{definition}

Opacity is useful when we want to hide \emph{only} the fact that a secret event happened, but not necessarily whether it did not happen. For example, in a system where some actions are private (like accessing a confidential file), we want observers to see general activity (like logins or public file access) but never be sure if the private action occurred. This way, we protect sensitive information while still allowing normal monitoring.

In contrast, in settings where sensitive features such as gender,
health status, or membership in a protected group must be hidden,
one often requires that neither the occurrence nor the absence
of the secret can be inferred; this stronger, two-sided variant
is discussed in Exercise~\ref{exercise:two-sided-opacity}.

\begin{example}[An opaque system]
\label{ex:opaque}
Let $\AP=\{p,r,s,\#\}$ and let agent $a$ observe only $\AP_a=\{p,r,\#\}$.
Consider the Büchi automaton $\mathcal S$ below:

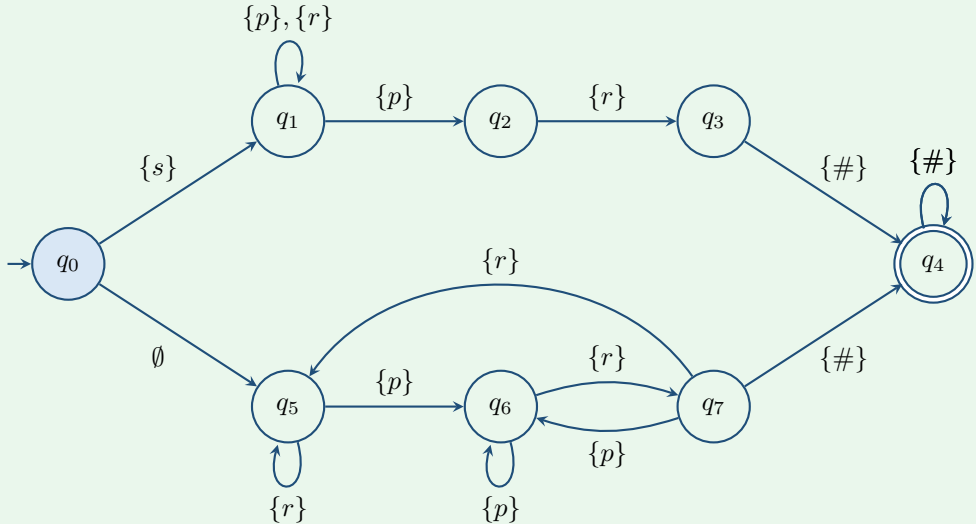
\begin{center}
\begin{tikzpicture}[
  ->, >=stealth, node distance=2.8cm,
  state/.style={circle, draw=ChapterAccent, thick, minimum size=0.95cm, inner sep=1pt},
  initial/.style={state, fill=ChapterLight},
  accepting/.style={state, double, double distance=1.4pt},
  every edge/.style={draw=ChapterAccent, thick},
  every node/.style={font=\small}
]

\node[initial] (q0) {$q_0$};

\node[state] (u1) [above right=1.2cm and 2.2cm of q0] {$q_1$};
\node[state] (u2) [right of=u1] {$q_2$};
\node[state] (u3) [right of=u2] {$q_3$};
\node[accepting] (u4) [below right=1.2cm and 2.2cm of u3] {$q_4$};

\node[state] (l1) [below right=1.2cm and 2.2cm of q0] {$q_5$};
\node[state] (l2) [right of=l1] {$q_6$};
\node[state] (l3) [right of=l2] {$q_7$};

\draw[thick, ChapterAccent, <-] (q0) -- ++(-0.8,0);

\path
(q0) edge node[above] {$\{s\}$\phantom{aaa}} (u1)
     edge node[below] {$\emptyset$\phantom{aaa}} (l1);

\path
(u1) edge[loop above] node {$\{p\},\{r\}$} (u1)
(u1) edge node[above] {$\{p\}$} (u2)
(u2) edge node[above] {$\{r\}$} (u3)
(u3) edge node[above] {\phantom{aaa}$\{\#\}$} (u4)
(u4) edge[loop above] node {$\{\#\}$} (u4);

\path
(l1) edge node[above] {$\{p\}$} (l2)
(l2) edge[bend left=18] node[above] {$\{r\}$} (l3)
(l3) edge[bend left=18] node[below] {$\{p\}$} (l2)
(l3) edge[bend right=55] node[above] {$\{r\}$} (l1)
(l3) edge node[below] {\phantom{aaa}$\{\#\}$} (u4)
(u4) edge[loop above] node {$\{\#\}$} (u4);

\path
(l1) edge[loop below] node {$\{r\}$} (l1);

\path
(l2) edge[loop below] node {$\{p\}$} (l2);
\end{tikzpicture}
\end{center}
The upper branch, after producing the secret,
generates exactly the words over 
$\{\{p\},\{r\}\}$ that end in the pattern $\{p\}\{r\}$
before the first occurrence of $\{\#\}$.
The lower branch generates the \emph{same} language,
but in a deterministic way.

In particular, for every execution along the upper branch there exists
an execution along the lower branch with the same observable trace.
Thus, the system is opaque.\footnote{%
In fact, the construction even satisfies the two-sided opacity
variant: neither the upper nor the lower branch can be inferred
from observations; see Exercise~\ref{exercise:two-sided-opacity}.}
\end{example}

\section{Complexity of Opacity}

While opacity is already known to be decidable by
Theorem~\ref{thm:decidability-mc}, we now establish that
the problem is, in fact, PSPACE-complete.

\begin{theorem}[Opacity is in PSPACE]
\label{thm:opacity-pspace}
The following problem can be decided in polynomial space:

Given a Büchi automaton $\mathcal{S}$ over $\Sigma$,
an agent $a \in \Ag$, and an unobservable secret $s \in \AP$,
decide whether $\mathcal{S}$ is opaque with respect to $a$ and $s$.
\end{theorem}

\begin{proof}
By Theorem~\ref{thm:automaton-construction},
there exists an $(\mathcal S,\Knows_a\tmp{P}s)$-transducer
consisting of $\mathcal A^{\mathcal S}_{\Knows_a\tmp{P}s}$
and $\gamma_{\Knows_a\tmp{P}s}$.
It satisfies the following correctness property:
$L(\mathcal A^{\mathcal S}_{\Knows_a\tmp{P}s}) = L(\mathcal{S})$
and,
for every $w\in L(\mathcal S)$, every accepting run $\pi$
of $\mathcal A^{\mathcal S}_{\Knows_a\tmp{P}s}$ over $w$,
and every position $k\ge 1$,
\[
\mathcal S,w,k \models \Knows_a\tmp{P}s
\quad\Longleftrightarrow\quad
\gamma_{\Knows_a\tmp{P}s}\bigl(\pi(k)\bigr)=1.
\]
We can assume that
its size (number of states)
is exponential in $|\mathcal S|$.
We first construct a linear-size $(\mathcal S,\tmp{P}s)$-transducer
$\mathcal A^{\mathcal S}_{\tmp{P}s}$ (together with $\gamma_{\tmp{P}s}$).
More precisely, $2 \cdot |\mathcal{S}|$ states are enough
where $|\mathcal{S}|$ is the number of states of $\mathcal{S}$.
We then obtain $\mathcal A^{\mathcal S}_{\Knows_a\tmp{P}s}$
using the knowledge-transducer construction, which yields at most
$|{\mathcal S}|\cdot 2^{|A^{\mathcal S}_{\tmp{P}s}|}$ states (up to a polynomial factor
depending on the chosen presentation of transducer states). In particular,
$|\mathcal A^{\mathcal S}_{\Knows_a\tmp{P}s}|$ is exponential in $|\mathcal S|$.

It remains to decide whether $\mathcal S$ is opaque, i.e.,
\[
\mathcal S \models \tmp{G}\,\neg \Knows_a\tmp{P}s.
\]
We have $\mathcal S \not\models \tmp{G}\,\neg \Knows_a\tmp{P}s$ iff there exist $w\in L(\mathcal S)$, an accepting
run $\pi$ of $\mathcal A^{\mathcal S}_{\Knows_a\tmp{P}s}$ over $w$, and a
position $k\ge 1$ such that
$\gamma_{\Knows_a\tmp{P}s}(\pi(k))=1$.
Equivalently, non-opacity holds iff
$\mathcal A^{\mathcal S}_{\Knows_a\tmp{P}s}$
has an accepting run that visits at least once a state
$q$ such that $\gamma_{\Knows_a\tmp{P}s}(q)=1$.

This can be checked in polynomial space without constructing
$\mathcal A^{\mathcal S}_{\Knows_a\tmp{P}s}$ explicitly.
Indeed, although $\mathcal A^{\mathcal S}_{\Knows_a\tmp{P}s}$ has
exponentially many states, its transition graph can be explored
\emph{on the fly}:
each state admits a polynomial-space representation,
and successor states can be computed in
polynomial space from the definition of the knowledge-transducer update.
Hence, reachability in this exponential graph can be decided in
$\mathrm{NPSPACE}$.
By $\mathrm{NPSPACE}=\mathrm{PSPACE}$ (Savitch's Theorem), opacity is decidable in polynomial
space.
\end{proof}

\begin{theorem}[Lower bound for opacity]
\label{thm:opacity-hard}
Deciding opacity is PSPACE-hard.
\end{theorem}

\begin{proof}
We reduce from the \emph{universality problem} for nondeterministic finite automata (NFAs), which is
PSPACE-complete\footnote{We do not introduce NFAs formally, as they are standard and their precise representation does not matter here.}: given an NFA $\mathcal{N}$ over a finite alphabet $\Gamma$, decide
whether $L(\mathcal{N}) = \Gamma^\ast$.

Let $\mathcal{N}$ be an NFA over $\Gamma$.
We set
$\AP = \Gamma \uplus \{\#\} \uplus \{s\}$ where $s$ is a fresh symbol
representing an \emph{unobservable secret}, and $\#$ is a fresh
\emph{observable} endmarker.
Consider a single agent $a$ and let $\AP_a = \AP \setminus \{s\}$.
We build a Büchi automaton $\mathcal{S}$ over the alphabet
$\Sigma = 2^\AP$ whose language $L(\mathcal{S})$ is the disjoint union
\[
L(\mathcal{S}) \;=\; L_{\mathsf{priv}} \;\uplus\; L_{\mathsf{pub}},
\]
where
\begin{itemize}
\item the \emph{private language} $L_{\mathsf{priv}}$ is
the set of words \[\{s\} \{\gamma_1\} \dots \{\gamma_n\} \{\#\}^\omega \in \Sigma^\omega\]
such that $n \in \mathbb{N}$ and $\gamma_i \in \Gamma$ for all $i\in\{1,\ldots,n\}$,
\item the \emph{public language} $L_{\mathsf{pub}}$ is
the set of words \[\emptyset \{\gamma_1\} \dots \{\gamma_n\} \{\#\}^\omega \in \Sigma^\omega\]
such that $n \in \mathbb{N}$, $\gamma_i \in \Gamma$ for all $i\in\{1,\ldots,n\}$,
and $\gamma_1 \ldots \gamma_n \in L(\mathcal{N})$.
\end{itemize}
It is easily seen that such a Büchi automaton $\mathcal{S}$ can be constructed in polynomial time.

We prove that
\[
L(\mathcal{N}) = \Gamma^\ast
\quad\Longleftrightarrow\quad
\mathcal{S} \models \tmp{G}\,\neg \Knows_a \tmp{P}s.
\]

\medskip\noindent
($\Rightarrow$)
Assume $L(\mathcal{N})=\Gamma^\ast$.
Let $w \in L(\mathcal{S})$ and $k \ge 1$ be arbitrary. We show
$\mathcal{S},w,k \models \neg \Knows_a \tmp{P}s$.

If $w \in L_{\mathsf{pub}}$, then $s$ never occurs in $w$, hence
$\mathcal{S},w,k \models \neg\tmp{P}s$, and therefore
$\mathcal{S},w,k \models \neg \Knows_a \tmp{P}s$.

Suppose now that $w \in L_{\mathsf{priv}}$. Then
\[
w = \{s\}\{\gamma_1\}\ldots\{\gamma_n\}\{\#\}^\omega
\]
for some $n\in\mathbb N$ and $\gamma_1,\ldots,\gamma_n \in \Gamma$.
Consider the word
\[
w' \;=\; \emptyset\{\gamma_1\}\ldots\{\gamma_n\}\{\#\}^\omega.
\]
Since $L(\mathcal{N})=\Gamma^\ast$, we have $\gamma_1\ldots\gamma_n\in L(\mathcal N)$,
hence $w' \in L_{\mathsf{pub}} \subseteq L(\mathcal{S})$.
Moreover,
$\obs_a(w_{\le k}) = \obs_a(w'_{\le k})$ and $\mathcal{S},w',k \not\models \tmp{P}s$.
Thus, by the semantics of knowledge, $\mathcal{S},w,k \models \neg \Knows_a \tmp{P}s$.

\medskip\noindent
($\Leftarrow$)
Assume $\mathcal{S} \models \tmp{G}\,\neg \Knows_a \tmp{P}s$.
We show that $L(\mathcal{N})=\Gamma^\ast$.

Let $\gamma_1\ldots\gamma_n \in \Gamma^\ast$ be arbitrary.
Consider the word
\[
w \;=\; \{s\}\{\gamma_1\}\ldots\{\gamma_n\}\{\#\}^\omega \in L_{\mathsf{priv}}
\subseteq L(\mathcal S).
\]
Let $k = n+2$ (the position of the first $\{\#\}$, i.e., after reading
$\{s\}$ and the whole block $\{\gamma_1\}\ldots\{\gamma_n\}$). Since $s$
occurs at position $1$, we have
\[
\mathcal{S},w,k \models \tmp{P}s.
\]

By $\mathcal{S} \models \tmp{G}\,\neg \Knows_a \tmp{P}s$, we obtain
\[
\mathcal{S},w,k \models \neg \Knows_a \tmp{P}s.
\]
Hence, by the semantics of knowledge, there exists $w' \in L(\mathcal S)$ such that
\[
\obs_a(w_{\le k}) = \obs_a(w'_{\le k})
\qquad\text{and}\qquad
\mathcal{S},w',k \models \neg\tmp{P}s.
\]
The second condition implies that $s$ does not occur in $w'_{\le k}$.
In particular, $w'$ cannot belong to $L_{\mathsf{priv}}$, so $w' \in L_{\mathsf{pub}}$.
Therefore
\[
w' \;=\; \emptyset\{\gamma'_1\}\ldots\{\gamma'_m\}\{\#\}^\omega
\quad\text{and}\quad
\gamma'_1\ldots\gamma'_m \in L(\mathcal N)
\]
for some $m\in\mathbb N$ and $\gamma'_i\in\Gamma$.

Moreover, since $s$ is unobservable to $a$, we have
\[
\obs_a(w_{\le k}) = \emptyset\{\gamma_1\}\ldots\{\gamma_n\}\{\#\}.
\]
From $\obs_a(w_{\le k})=\obs_a(w'_{\le k})$, we obtain
\[
m=n
\quad\text{and}\quad
\gamma_i=\gamma'_i \text{ for all } i\in\{1,\ldots,n\}.
\]
Hence $\gamma_1\ldots\gamma_n = \gamma'_1\ldots\gamma'_n \in L(\mathcal N)$.

Since $\gamma_1\ldots\gamma_n\in\Gamma^\ast$ was arbitrary, we conclude
\[
L(\mathcal N)=\Gamma^\ast.
\]

\medskip\noindent
This establishes a polynomial-time reduction from NFA universality to the
complement of opacity, hence deciding opacity is PSPACE-hard.
\end{proof}

\begin{exercise}[Two-sided opacity]
\label{exercise:two-sided-opacity}
Recall that we considered the following (one-sided) opacity requirement
for a secret $s\in\AP$ and agent $a\in\Ag$:
\[
\Phi_{\mathsf{opaque}}(s)
\;=\;
\tmp{G}\neg\Knows_a \tmp{P}s.
\]
Now consider the following two-sided opacity requirement:
\[
\Phi^{\mathsf{2}}_{\mathsf{opaque}}(s)
\;=\;
\tmp{G}\neg\bigl(\Knows_a \tmp{P}s \;\vee\; \Knows_a \neg\tmp{P}s\bigr).
\]
For instance, for the system $\mathcal{S}$ from Example~\ref{ex:opaque}, we have
$\mathcal{S} \models \Phi^{\mathsf{2}}_{\mathsf{opaque}}(s)$.
\begin{enumerate}
\item[(a)] Discuss when the additional requirement
$\neg\Knows_a\neg\tmp{P}s$ is desirable in practice.

\item[(b)] Show that $\Phi^{\mathsf{2}}_{\mathsf{opaque}}(s)$ implies the standard
opacity requirement
$\Phi_{\mathsf{opaque}}(s)$.

\item[(c)] Show that the converse implication does not hold by giving a Büchi automaton
$\mathcal S$ and agent $a$ such that
\[
\mathcal S \models \Phi_{\mathsf{opaque}}(s)
\quad\text{but}\quad
\mathcal S \not\models \Phi^{\mathsf{2}}_{\mathsf{opaque}}(s).
\]

\item[(d)] Show that deciding two-sided opacity is PSPACE-complete.
\item[(e)] Fix now $\Ag = \{a\}$, $\AP=\{p,s\}$, $\AP_a=\{p\}$, and $\Sigma=2^{\AP}$.
Consider systems $\mathcal S_1,\mathcal S_2,\mathcal S_3,\mathcal S_4$ whose languages
$L(\mathcal S_i) \subseteq \Sigma^\omega$ are given as follows:
\begin{align*}
L(\mathcal S_1)&=\{p\}^{\omega}+\{p\}\{p,s\}\{p\}^{\omega}\\[1ex]
L(\mathcal S_2)&=\{p\}^{\omega}+\{p,s\}\{p\}^{\omega}\\[1ex]
L(\mathcal S_3)&=\{p\}^{\omega}
+\{p\}^{\ast}\{p,s\}\{p\}^{\ast}\emptyset^{\omega}
+\{p\}^{\ast}\{p,s\}\{p\}^{\ast}\{p,s\}^{\omega}\\[1ex]
L(\mathcal S_4)&=\{p\}^{\omega}
+\{p\}^{\ast}\{p,s\}\{p\}^{\omega}
+\{p\}^{\ast}\emptyset\{p\}^{\omega}
+\{p\}^{\ast}\{p,s\}\{p\}^{\ast}\emptyset\{p\}^{\omega}
\end{align*}

For each of the systems $\mathcal S_1,\mathcal S_2,\mathcal S_3,\mathcal S_4$, determine whether it
satisfies $\Phi_{\mathsf{opaque}}(s)$ and whether it satisfies
$\Phi^{\mathsf{2}}_{\mathsf{opaque}}(s)$.
Justify your answer briefly from the definitions.
\end{enumerate}
\end{exercise}

\section{Bibliographic Notes}

Opacity was introduced as a confidentiality property for discrete-event systems, requiring that secret behavior cannot be inferred by an external observer. Bryans et al.~\cite{BryansKMR08} generalized opacity to transition systems and clarified its language-theoretic formulation. Cassez, Dubreil, and Marchand~\cite{CassezDM09} studied dynamic observers and established that deciding opacity is PSPACE-complete.

%
%


\chapter{Monitorability}

Monitoring addresses the question of whether the truth of a temporal property
can be determined by observing only a finite prefix of a system execution.
In practice, a monitor processes observations incrementally and tries to
reach a definitive verdict as early as possible.
In this section, we study when such a definitive verdict is
obtainable under partial observability and epistemic uncertainty.
This leads to the notion of \emph{monitorability}, which captures whether
sufficient information may eventually be gathered to decide a given property
based on observations.

We are particularly interested in identifying structural conditions on the
specification that guarantee the existence of such conclusive observations.

\section{Definition and Decidability}

We now formalize the idea that, from every possible observation,
it should always remain possible to eventually decide the property.

\begin{definition}[Monitorability]
\label{def:monitorability}
Let $\mathcal{S}$ be a Büchi automaton over $\Sigma$, $a \in \Ag$, 
and $\varphi$ an \LTL property.
We say that $(\mathcal{S}, \varphi)$ is \emph{monitorable}
(wrt.\ $a$) if for every observation prefix
$u \in \Obs_a(\mathcal S)$
there exists a finite extension
$v \in \Sigma_a^\ast$ such that
$uv \in \Obs_a(\mathcal S)$ and
\[
\mathcal S, uv \models_a \hat\varphi
\quad \text{or}\quad
\mathcal S, uv \models_a \neg\hat\varphi
\]
where
\[
\hat\varphi \;=\; \tmp{P}(\mathit{first} \wedge \varphi).
\]
\end{definition}
Here, the formula $\tmp{P}(\mathit{first}\wedge\varphi)$ is used to encode
whole-trace satisfaction of $\varphi$ in a position-based semantics:
$\mathcal S,w,k\models\tmp{P}(\mathit{first}\wedge\varphi)$ holds iff
the infinite word $w$ satisfies $\varphi$ at position~1.

Thus, if $(\mathcal S,\varphi)$ is monitorable, then from every observation
prefix it is always worth continuing the monitoring process: a definitive
verdict may eventually be reached.

\begin{example}[Monitorability]
\label{ex:monitorability}
Let $\AP=\{e,r,s\}$ and let agent $a$ observe only $\AP_a=\{r,s\}$, so $e$ is unobservable.
Consider the property
$\varphi = \tmp{F}e$. That is,
$\hat\varphi = \tmp{P}(\mathit{first}\wedge \tmp{F}e)$.

Let $\mathcal S$ be the Büchi automaton below:
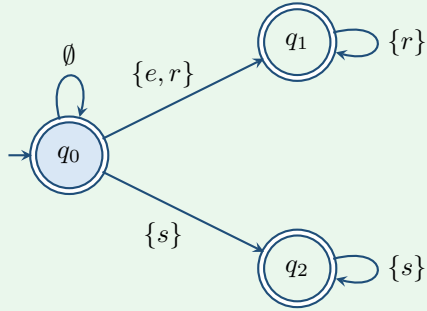
\begin{center}
\begin{tikzpicture}[
  ->, >=stealth, node distance=3cm,
  state/.style={circle, draw=ChapterAccent, thick, minimum size=0.95cm, inner sep=1pt},
  initial/.style={state, fill=ChapterLight},
  accepting/.style={state, double, double distance=1.4pt},
  every edge/.style={draw=ChapterAccent, thick},
  every node/.style={font=\small}
]
\node[accepting] (q1) {$q_1$};
\node[accepting] (q2) [below of=q1] {$q_2$};
\node[initial,accepting] (q0) [left of=q1, yshift=-1.5cm] {$q_0$};

\draw[thick, ChapterAccent, <-] (q0) -- ++(-0.8,0);

\path
(q0) edge[loop above] node {$\emptyset$} (q0)
     edge node[above] {$\{e,r\}$\phantom{aaa}} (q1)
     edge node[below] {$\{s\}$\phantom{aaa}} (q2)
(q1) edge[loop right] node {$\{r\}$} (q1)
(q2) edge[loop right] node {$\{s\}$} (q2);
\end{tikzpicture}
\end{center}

Intuitively, the system may stutter arbitrarily long in $q_0$ and will eventually
enter a recovery mode (after an error occurred) or a safe/success mode.
In the first case, $e$ occurs (together with $r$) and $r$ is observed forever.
In the second case, only $s$ is observed and $e$ never occurs.

From any observation prefix $u\in\Obs_a(\mathcal S)$,
one can extend the execution so that either $r$ or $s$ is observed next.
In the former case, the agent knows that $e$ has occurred;
in the latter case, the agent knows that $e$ never occurs.
Hence some finite continuation always yields a definitive verdict.
Thus, $(\mathcal S,\varphi)$ is monitorable.
\end{example}

In contrast to diagnosability and opacity, monitorability is not expressed
as a model-checking question for a single formula.
It is a structural property of $(\mathcal S,\varphi)$ that quantifies over
all observation prefixes.
However, the automata-theoretic machinery we developed still allows us
to decide monitorability.

\begin{theorem}[Decidability of monitorability]
Given a Büchi automaton $\mathcal{S}$ over $\Sigma$, an agent $a \in \Ag$, and
an \LTL property $\varphi$,
the problem of determining whether $(\mathcal{S}, \varphi)$ is monitorable (wrt.\ $a$)
is decidable.
\end{theorem}

\begin{proof}
Let
\[
\hat\varphi = \tmp{P}(\mathit{first} \wedge \varphi).
\]
By Definition~\ref{def:monitorability},
$(\mathcal S,\varphi)$ is monitorable iff
\begin{equation}\label{eq:monitor-cond}
\forall u\in\Obs_a(\mathcal S)\colon\exists v\in\Sigma_a^\ast\colon\;
uv\in\Obs_a(\mathcal S)
\ \land\
\bigl(
\mathcal S,uv\models_a \hat\varphi
~\text{ or }~
\mathcal S,uv\models_a \neg\hat\varphi
\bigr).
\end{equation}

\medskip
Construct the knowledge monitors
\[
\mathcal K^{\mathcal S,a}_{\hat\varphi}
=
(Q_1,\iota_1,\delta_1,F_1)
\qquad\text{and}\qquad
\mathcal K^{\mathcal S,a}_{\neg\hat\varphi}
=
(Q_0,\iota_0,\delta_0,F_0).
\]
By correctness of knowledge monitors, for every observation prefix
$u\in\Obs_a(\mathcal S)$,
\[
\delta_1^\ast(\iota_1,u)\in F_1
\;\Longleftrightarrow\;
\mathcal S,u\models_a \hat\varphi,
\]
and
\[
\delta_0^\ast(\iota_0,u)\in F_0
\;\Longleftrightarrow\;
\mathcal S,u\models_a \neg\hat\varphi.
\]

Hence \eqref{eq:monitor-cond} is equivalent to:
\begin{equation}\label{eq:reach-cond-final}
\forall u\in\Obs_a(\mathcal S)\colon\exists v\in\Sigma_a^\ast\colon\;
uv\in\Obs_a(\mathcal S)
\ \land\
\bigl(
\delta_1^\ast(\iota_1,uv)\in F_1
~\text{ or }~
\delta_0^\ast(\iota_0,uv)\in F_0
\bigr).
\end{equation}

\medskip
The observation language
\[
\Obs_a(\mathcal S)
=
\{\obs_a(u)\mid u\in\Pref(L(\mathcal S))\}
\subseteq \Sigma_a^\ast
\]
is regular and effectively constructible from $\mathcal S$
in terms of a DFA $\mathcal P$ over $\Sigma_a$
such that \[L(\mathcal{P}) = \Obs_a(\mathcal S).\]
Consider the synchronous product
\[
\mathcal G
=
\mathcal P \times
\mathcal K^{\mathcal S,a}_{\hat\varphi} \times
\mathcal K^{\mathcal S,a}_{\neg\hat\varphi}.
\]
States of $\mathcal G$ reachable by words in $\Obs_a(\mathcal S)$
correspond exactly to feasible nonempty observation prefixes together
with the current states of both knowledge monitors.
Condition~\eqref{eq:reach-cond-final} holds iff
from every such reachable state of $\mathcal G$
one can reach another state of $\mathcal G$
whose $\mathcal P$-component is accepting
and whose
$\mathcal K^{\mathcal S,a}_{\hat\varphi}$-component lies in $F_1$
or whose
$\mathcal K^{\mathcal S,a}_{\neg\hat\varphi}$-component lies in $F_0$.

This is a standard reachability property in a finite directed graph.
Thus, monitorability is decidable.
\end{proof}

\section{Monitors}

In the following, we make the notion of monitoring explicit by constructing a
monitor that produces three-valued verdicts.

\begin{definition}[Monitor]
Given a Büchi automaton $\mathcal{S}$ over $\Sigma$, an agent $a \in \Ag$, and
an \LTL property $\varphi$, a
$(\mathcal{S}, \varphi)$-\emph{monitor} for $a$ is a Moore machine
\[
\mathcal M^{\mathcal S,a}_\varphi =
(Q, \iota, \delta, \outmap)
\]
over $\Sigma_a$ and $\Gamma=\{\tverdict,\fverdict,\unknown\}$
such that, for every observation prefix
$u \in \Obs_a(\mathcal S)$, we have the following:
\[
\sem{\mathcal M^{\mathcal S,a}_\varphi}(u) = \outmap(\delta^\ast(\iota,u)) =
\begin{cases}
  \tverdict
    & \text{if } \mathcal S,u \models_a \hat\varphi\\[0.4ex]
  \fverdict
    & \text{if } \mathcal S,u \models_a \neg\hat\varphi\\[0.4ex]
  \unknown
    & \text{otherwise}
\end{cases}
\]
\end{definition}

Operationally, a monitor reads the observation sequence produced by the system
and incrementally updates its internal state.
The verdict reflects whether, based on the observations seen so far, the agent
has obtained sufficient knowledge to decide whether the execution satisfies
$\varphi$.

\begin{example}[Monitor]
Consider Example~\ref{ex:monitorability}.
The following Moore machine is an $(\mathcal{S}, \phi)$-monitor for $a$
(the output value for a state is given in the state itself):
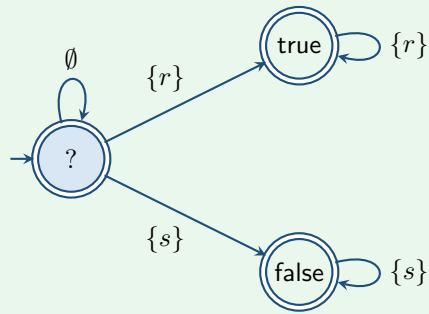
\begin{center}
\begin{tikzpicture}[
  ->, >=stealth, node distance=3cm,
  state/.style={circle, draw=ChapterAccent, thick, minimum size=0.95cm, inner sep=1pt},
  initial/.style={state, fill=ChapterLight},
  accepting/.style={state, double, double distance=1.4pt},
  every edge/.style={draw=ChapterAccent, thick},
  every node/.style={font=\small}
]
\node[accepting] (q1) {$\tverdict$};
\node[accepting] (q2) [below of=q1] {$\fverdict$};
\node[initial,accepting] (q0) [left of=q1, yshift=-1.5cm] {$?$};

\draw[thick, ChapterAccent, <-] (q0) -- ++(-0.8,0);

\path
(q0) edge[loop above] node {$\emptyset$} (q0)
     edge node[above] {$\{r\}$\phantom{aaa}} (q1)
     edge node[below] {$\{s\}$\phantom{aaa}} (q2)
(q1) edge[loop right] node {$\{r\}$} (q1)
(q2) edge[loop right] node {$\{s\}$} (q2);
\end{tikzpicture}
\end{center}
\end{example}

It is always possible to construct a monitor (however, recall
that only monitorability guarantees that a final verdict
may always eventually be reached):

\begin{proposition}[Monitor construction]
\label{prop:monitor}
Let $\mathcal S$ be a Büchi automaton over $\Sigma$ and $a$ be an agent.
Let $\varphi$ be an \LTL formula.
We can effectively construct an $(\mathcal{S}, \varphi)$-\emph{monitor} for $a$.
\end{proposition}

\begin{proof}
Let $\hat\varphi = \tmp{P}(\mathit{first} \wedge \varphi)$.
Construct the knowledge monitors
\[
\mathcal K^{\mathcal S,a}_{\hat\varphi}
=
(Q_1,\iota_1,\delta_1,F_1)
\quad \text{and}\quad
\mathcal K^{\mathcal S,a}_{\neg\hat\varphi}
=
(Q_0,\iota_0,\delta_0,F_0).
\]

Then, we obtain an $(\mathcal{S}, \varphi)$-\emph{monitor}
\[
\mathcal M^{\mathcal S,a}_\varphi =
(Q, \iota, \delta, \outmap)
\]
for $a$ as follows:
\begin{itemize}
\item $Q = Q_1 \times Q_0$
\item $\iota = (\iota_1, \iota_0)$
\item $\delta((q_1,q_0), \sigma) =
(\delta_1(q_1, \sigma), \delta_0(q_0, \sigma))$
\item
$
\outmap((q_1,q_0)) =
\begin{cases}
  \tverdict & \text{if } q_1 \in F_1 \\
  \fverdict & \text{if } q_0 \in F_0 \\
  \unknown & \text{otherwise}
\end{cases}
$
\end{itemize}

By correctness of knowledge monitors, at most one of the first two cases
can occur, since the agent cannot simultaneously know
$\hat\varphi$ and $\neg\hat\varphi$.
Thus the monitor yields a well-defined three-valued verdict
at every observation prefix.
By correctness of the knowledge monitors, we also immediately obtain
correctness of $\mathcal M^{\mathcal S,a}_\varphi$.
\end{proof}

\begin{notebox}[Good, bad, and ugly prefixes]
In the literature, an observation prefix $u$ is called
\emph{good} if $\mathcal S,u \models_a \hat\varphi$, and
\emph{bad} if $S,u \models_a \neg\hat\varphi$.
A prefix that is neither good nor bad is \emph{inconclusive}.
Such a prefix is called \emph{ugly} if no extension of it can ever become
good or bad, i.e., if every extension remains inconclusive.
In this case, a monitor will output $\unknown$ forever.
\end{notebox}

Once the monitor produces a definitive verdict, it is
permanent under all further observations:

\begin{proposition}[Permanency of monitor verdicts]
Let $\mathcal S$ be a system, $a\in\Ag$, and $\varphi$ a temporal property.
Let $\mathcal M = (Q,\iota,\delta,\outmap)$
be an $(\mathcal{S}, \varphi)$-\emph{monitor} for $a$.
Then, for all $u \in \Sigma_a^+$ and $v \in \Sigma_a^\ast$ such that
$uv \in \Obs_a(\mathcal S)$, the following holds:
\[
\sem{\mathcal M}(u) \in \{\tverdict,\fverdict\}
\quad
\Longrightarrow
\quad
\sem{\mathcal M}(uv)
=
\sem{\mathcal M}(u)
\]
\end{proposition}

\begin{proof}
We treat the case $\tverdict$; the case $\fverdict$ is symmetric.

Assume $\sem{\mathcal{M}}(u)=\tverdict$.
By the definition of monitors,
\[
\mathcal S,u \models_a \hat\varphi.
\]
Now fix an arbitrary extension $v\in\Sigma_a^\ast$
such that $uv \in \Obs_a(\mathcal S)$.
We will show
\[\mathcal S,uv \models_a \hat\varphi,\] which implies
$\sem{\mathcal{M}}(uv)=\tverdict$.

Consider any
$w \in L(\mathcal S)$ with
\[
\obs_a(w_{\le |uv|})= uv.
\]
This implies
$\obs_a(w_{\le |u|})=u$. Hence, by $\mathcal S,u \models_a \hat\varphi$, we have $\mathcal S,w,|u| \models \hat\varphi$.
Since $\hat\varphi=\tmp{P}(\mathit{first}\wedge\varphi)$, its truth is time-invariant
along a run: for all $k\ge 1$,
\[
\mathcal S,w,k \models \hat\varphi \quad\Longleftrightarrow\quad \mathcal S,w,1 \models \varphi.
\]
Thus, $\mathcal S,w,|uv| \models \hat\varphi$.

We have shown $\mathcal S,uv \models_a \hat\varphi$.
By the definition of monitors, this implies
$\sem{\mathcal{M}}(uv)=\tverdict$.
\end{proof}

\section{Classical Runtime Verification}

In classical runtime verification (what we call monitoring), one commonly
makes the following assumptions on the system $\mathcal{S}$,
the given property $\phi$,
the set of propositions $\AP$ and
the set of agents $\Ag$.
\begin{enumerate}
\item There is only one observer and observation is complete, i.e., $\Ag = \{a\}$ and $\AP_a = \AP$. 
\item The property $\phi$ is an \pureLTL formula (without knowledge operators).
\item The underlying system is the most permissive one (a complete black-box), i.e.,
\[L(\mathcal{S}) = \Sigma^\omega .\]
\end{enumerate}

We will adopt all these assumptions in the present section.
Consequently, given an \pureLTL formula $\phi$, we say that
$\phi$ is \emph{monitorable} if $(\mathcal{S}, \phi)$ is monitorable wrt.~$a$.
Moreover, we define its language as
\[
L(\phi) = \bigl\{\, w \in \Sigma^\omega \mid \mathcal{S}, w, 1 \models \phi \,\bigr\}.
\]
Recall that we hereby assume $L(\mathcal{S}) = \Sigma^\omega$.

\begin{notebox}[Complexity]
The problem of deciding monitorability is PSPACE-hard, and
an EXPSPACE upper bound was shown for the future fragment of \pureLTL~\cite{DiekertMW15}.
\end{notebox}

As we look at most permissive systems, we can examine specifications
more closely and examine which of them are monitorable.
Let us first look at closure properties:

\begin{proposition}[Closure of monitorability under Boolean operations]
\label{prop:closure-monitorability}
Let $\phi$ and $\psi$ be monitorable \pureLTL formulas.
Then, $\phi \wedge \psi$ and $\neg\phi$ are monitorable.
\end{proposition}

\begin{exercise}
Prove Proposition~\ref{prop:closure-monitorability}.
\end{exercise}

The next question to ask is: What are natural ``atomic'' language classes
that guarantee monitorability. Well-known monitorable properties are safety
and co-safety properties:

\begin{definition}[Safety and co-safety properties]
Let $\phi$ be an \pureLTL formula.
We call $\phi$ a
\begin{itemize}
\item \emph{safety property} if, for all $w \in \Sigma^\omega \setminus L(\phi)$,
there is $u \in \Pref(w)$ such that, for all $v \in \Sigma^\omega$, we have
$uv \not\in L(\phi)$,

\item \emph{co-safety property} if, for all $w \in L(\phi)$,
there is $u \in \Pref(w)$ such that, for all $v \in \Sigma^\omega$, we have
$uv \in L(\phi)$.
\end{itemize}
\end{definition}

\begin{proposition}[Monitorability of safety and co-safety properties]
\label{prop:safety-monitorability}
Safety and co-safety properties are monitorable.
\end{proposition}

\begin{exercise}
Prove Proposition~\ref{prop:safety-monitorability}.
\end{exercise}

\begin{solution}
Assume $L(\mathcal S)=\Sigma^\omega$ and let $\varphi$ be an \pureLTL formula.

\smallskip\noindent
\emph{Safety.}
Assume that $\varphi$ is a safety property and fix an arbitrary prefix $u\in\Sigma^\ast$.
If $u$ is already a bad prefix, i.e.\ if for all $w\in\Sigma^\omega$ we have
$uw\notin L(\varphi)$, then $u$ already yields a definitive verdict.

Otherwise, there exists some $w\in\Sigma^\omega$ such that $uw\notin L(\varphi)$.
By safety, there is a prefix $x\in\Pref(uw)$ such that for all $z\in\Sigma^\omega$,
\[
xz\notin L(\varphi).
\]
Since $u$ is a prefix of $uw$, we can write $x=uv$ for some $v\in\Sigma^\ast$.
Then $uv$ is a bad prefix. Hence from every prefix $u$ there exists a finite
continuation leading to a definitive verdict, so $\varphi$ is monitorable.

\smallskip\noindent
\emph{Co-safety.}
Assume that $\varphi$ is a co-safety property and fix $u\in\Sigma^\ast$.
If $u$ is already a good prefix, i.e.\ if for all $w\in\Sigma^\omega$ we have
$uw\in L(\varphi)$, then $u$ already yields a definitive verdict.

Otherwise, there exists some $w\in\Sigma^\omega$ such that $uw\in L(\varphi)$.
By co-safety, there is a prefix $x\in\Pref(uw)$ such that for all $z\in\Sigma^\omega$,
\[
xz\in L(\varphi).
\]
Write again $x=uv$ for some $v\in\Sigma^\ast$. Then $uv$ is a good prefix.
Thus from every prefix $u$ there exists a finite continuation leading to a
definitive verdict, so $\varphi$ is monitorable.
\end{solution}

Safety and co-safety properties play a fundamental role in runtime
verification, even though they do not capture all temporal specifications.
They characterize the two extremal cases in which violations or
satisfactions, respectively, can be determined from a finite observation.
Within the class of \pureLTL-definable properties, the class of monitorable
formulas strictly contains safety and co-safety,
but is itself a strict subset of \pureLTL (cf.\ Exercise~\ref{exercise:safety}):
\[
\text{Safety} \;\cup\; \text{Co-safety}
\;\subsetneqq\;
\text{Monitorable}
\;\subsetneqq\;
\text{\pureLTL}
\]

\begin{exercise}
\label{exercise:safety}
Which of the following \pureLTL formulas are safety properties, which are co-safety properties,
which are monitorable properties?
\[
\begin{array}{llll}
(1)\ \tmp{G}p
&
(2)\ \tmp{F}p
&
(3)\ \tmp{X}p
&
(4)\ \tmp{G}\tmp{F}p
\\[1mm]
(5)\ \tmp{F}\tmp{G}p
&
(6)\ \tmp{X}p \,\vee\, \tmp{G}\tmp{F}p
&
(7)\ p \tmp{U} r
&
(8)\ \tmp{F}p \,\vee\, \tmp{G}r
\end{array}
\]
\end{exercise}

\begin{solution}
The full solution can be found in Table~\ref{formula-classification}.
Note that there is a \emph{monitorable} \pureLTL formula that is neither a safety nor a co-safety property.
\end{solution}

\begin{table}
\[
\begin{tabular}{lccc}
\toprule
Formula & Safety & Co-safety & Monitorable \\
\midrule
$\tmp{G}p$ & \yes & $\no$ & \yes \\[0.5ex]
$\tmp{F}p$ & $\no$ & \yes & \yes \\[0.5ex]
$\tmp{X}p$ & \yes & \yes & \yes \\[0.5ex]
$\tmp{G}\tmp{F}p$ & $\no$ & $\no$ & $\no$ \\[0.5ex]
$\tmp{F}\tmp{G}p$ & $\no$ & $\no$ & $\no$ \\[0.5ex]
$\tmp{X}p \,\vee\, \tmp{G}\tmp{F}p$ & $\no$ & $\no$ & $\no$ \\[0.5ex]
$p \tmp{U} r$ & $\no$ & \yes & \yes \\[0.5ex]
$\tmp{F}p \,\vee\, \tmp{G}r$ & $\no$ & $\no$ & \yes\\
\bottomrule
\end{tabular}
\]
\caption{Classification of formulas\label{formula-classification}}
\end{table}

\begin{exercise}[Monitorability with Past Operators]
Fix $\Ag = \{a\}$, $\AP=\AP_a=\{p,q\}$, and $\Sigma=2^{\AP}$.
We work in the classical monitoring setting of a given most-permissive system $\mathcal{S}$, i.e.,
$L(\mathcal S)=\Sigma^\omega$.

For each of the following formulas $\varphi$, determine whether $(\mathcal S,\varphi)$ is monitorable.
Justify your answer briefly.
\[
{(a)}\ \tmp{G}\tmp{F}\tmp{Y}p\qquad\qquad
{(b)}\ \tmp{Y}\tmp{G}\tmp{F}p\qquad\qquad
{(c)}\ \tmp{G}\tmp{F}\tmp{P}p\qquad\qquad
{(d)}\ \tmp{G}(q \rightarrow \tmp{P}p)
\]
\end{exercise}

\section{Bibliographic Notes}

The framework of runtime verification and monitor-based checking was
systematically developed by Leucker and Schallhart~\cite{LeuckerS09}
and further elaborated for temporal logics and three-valued semantics
by Bauer, Leucker, and Schallhart~\cite{BauerLS11}.
In this setting, monitoring is anticipatory in that verdicts are based
on all possible continuations of the current finite prefix.

The notion of monitorability was introduced by Pnueli and Zaks in the
context of PSL~\cite{PnueliZ06}.
For $\omega$-regular languages given by Büchi automata,
Diekert, Muscholl, and Walukiewicz showed that deciding
monitorability is PSPACE-complete~\cite{DiekertMW15}.
For the future fragment of \pureLTL, they established PSPACE-hardness
and an EXPSPACE upper bound.

Epistemic and information-flow specifications can also be formulated as
hyperproperties, that is, properties relating multiple executions.
Runtime monitoring techniques for HyperLTL and related formalisms were
developed by Bonakdarpour and Finkbeiner~\cite{BonakdarpourF16} and
further studied by Finkbeiner, Hahn, Stenger, and Tentrup~\cite{FinkbeinerHST17}.



\chapter{Knowledge in Timed Systems}

\newcommand{\Time}{\mathbb{R}_{\ge 0}}
\newcommand{\TW}{\mathsf{TW}}
\newcommand{\Clocks}{X}
\newcommand{\Guards}{\mathit{Guards}}
\newcommand{\Val}{\mathit{Val}}
\newcommand{\Intervals}{\mathbb{I}}
\newcommand{\ival}{\overline{0}}
\newcommand{\timetrace}{\mathit{trace}}

This final chapter briefly discusses the timed case, i.e.,
timed diagnosability, timed opacity, and timed monitorability.
The aim is not to redevelop the full untimed theory in complete detail, but to
record the key definitions and the main phenomena that appear once timestamps
enter the picture and become part of the model. Two messages are central:
\begin{itemize}
\item timed temporal reasoning without epistemic operators remains decidable,
\item adding knowledge already leads very quickly to undecidability.
\end{itemize}

\section{Timed Words and Timed Automata}

As before, we fix a finite set of atomic propositions $\AP$, a finite set of agents
$\Ag$, and for each agent $a \in \Ag$ a set of observable propositions
$\AP_a \subseteq \AP$.
We again write $\Sigma = 2^{\AP}$ and $\Sigma_a = 2^{\AP_a}$.

\begin{definition}[Timed words]
A \emph{timed word} over $\Sigma$ is an infinite sequence
\[
w = (\sigma_1,t_1)(\sigma_2,t_2)(\sigma_3,t_3)\ldots
\]
with $\sigma_i \in \Sigma$ and $t_i \in \Time$ such that
\[
0 \le t_1 \le t_2 \le t_3 \le \ldots
\]
It is \emph{non-zeno} if time diverges:
\[
\forall t \in \Time\colon \exists k \in \Npos\colon t \le t_k.
\]

A \emph{finite timed word} is defined analogously as a finite sequence
$
u = (\sigma_1,t_1)\ldots(\sigma_n,t_n)
$
with nondecreasing timestamps.
\end{definition}

Intuitively, each timestamp $t_i$ is the time at which the event labeled by
$\sigma_i$ occurs.

We write $\TW^\omega(\Sigma)$ for the set of non-zeno infinite timed words
over $\Sigma$, and $\TW^\ast(\Sigma)$ and $\TW^+(\Sigma)$ for the sets of
finite and nonempty finite timed words over $\Sigma$.

\begin{definition}[Timed observation]
For an agent $a \in \Ag$, the observation of a timed event $(\sigma,t) \in \Sigma \times \Time$ is
\[
\obs_a((\sigma,t)) = (\sigma \cap \AP_a,t).
\]
This is extended pointwise to finite and infinite timed words.
\end{definition}

Thus, in this chapter we assume that time is observable: the observer sees the
same timestamps as the system, but only the propositions in the set $\AP_a$.
This is natural in settings where events are globally timestamped or agents
have access to a shared clock.
Other models may treat time as partially observable or unobservable, which
would lead to a different indistinguishability relation.

\begin{definition}[Timed indistinguishability]
For finite timed words $u,u' \in \TW^\ast(\Sigma)$, we write
\[
u \sim_a u'
\quad\text{if}\quad
\obs_a(u)=\obs_a(u').
\]
Here, as before, a pointed timed word $(w,k)$ consists of a timed word $w$ and
one of its positions $k$, and $w_{\le k}$ denotes the timed prefix up to
position $k$.
For pointed timed words $(w,k)$ and $(w',k)$, we write
\[
(w,k)\sim_a (w',k)
\quad\text{if}\quad
w_{\le k}\sim_a w'_{\le k}.
\]
\end{definition}

\begin{notebox}[Synchronous perfect recall with time]
In the timed setting we again use synchronous perfect recall.
The observer remembers the whole observation history and also sees all global
timestamps.
Hence indistinguishability compares complete timed observation prefixes.
\end{notebox}

\paragraph{Timed B\"uchi automata.}
Timed B\"uchi automata serve as the system model.

Let $\Clocks$ be a finite set of clocks.
We write $\Guards(\Clocks)$ for the set of Boolean combinations of atomic
constraints
\[
x \in I
\]
with $x \in \Clocks$ and $I \in \Intervals$, where $\Intervals$ is the set of
intervals of the form
\[
(m,n),\ [m,n),\ (m,n],\ [m,n],\ (m,\infty),\ [m,\infty),
\]
with $m,n \in \mathbb{N}$ and $m \le n$ (note that this includes the empty
interval).
A member $g \in \Guards(\Clocks)$ is called a \emph{guard} over $\Clocks$.
A valuation is a mapping $\nu \colon \Clocks \to \Time$.
We write $\nu \models g$ in the obvious way, denote by $\ival$ the all-zero
valuation, and write $\nu[R:=0]$ for the valuation obtained from $\nu$
by resetting the clocks in $R \subseteq \Clocks$ to $0$.

\begin{definition}[Timed Büchi automaton]
A \emph{timed B\"uchi automaton} over $\Sigma$ is a tuple
\[
\mathcal A=(Q,\iota,\Clocks,\Trans,F)
\]
where:
\begin{itemize}
\item $Q$ is a finite set of states and $\iota \in Q$ is the initial state,
\item $\Clocks$ is a finite set of clocks,
\item $F \subseteq Q$ is the set of accepting states,
\item $\Trans \subseteq Q \times \Sigma \times \Guards(\Clocks) \times 2^{\Clocks} \times Q$.
\end{itemize}
\end{definition}

The automaton is \emph{deterministic} if for all distinct transitions
\[
\tau_1=(q,\sigma,g_1,R_1,q_1),\qquad
\tau_2=(q,\sigma,g_2,R_2,q_2)
\]
in $\Trans$, we have
\[
g_1 \wedge g_2 \equiv \mathsf{false}.
\]

\begin{definition}[Configurations, runs, and traces]
A \emph{configuration} is a pair $(q,\nu) \in Q \times \Val(\Clocks)$.
Let $d \in \Time$ denote a delay, that is, the amount of time elapsed before
the next event occurs.
We write
\[
(q,\nu)\xrightarrow[(g,R)]{(\sigma,d)}(q',\nu')
\]
if $(q,\sigma,g,R,q') \in \Trans$, $\nu+d \models g$, and
$\nu'=(\nu+d)[R:=0]$.

A \emph{run} of $\mathcal A$ is an infinite sequence
\[
\pi = (q_0,\nu_0)\xrightarrow[(g_1,R_1)]{(\sigma_1,d_1)}(q_1,\nu_1)
\xrightarrow[(g_2,R_2)]{(\sigma_2,d_2)}(q_2,\nu_2)\ldots
\]
starting in $(\iota,\ival)$, i.e., such that $(q_0,\nu_0) = (\iota,\ival)$.

Its induced timed word is
\[
\timetrace(\pi)
=
(\sigma_1,d_1)(\sigma_2,d_1+d_2)(\sigma_3,d_1+d_2+d_3)\ldots.
\]
The run is \emph{accepting} if it visits $F$ infinitely often.
\end{definition}

The timed language of $\mathcal A$ is
\[
L(\mathcal A)
=
\{\, w \in \TW^\omega(\Sigma) \mid \exists~ \text{accepting run }
\pi \text{ of } \mathcal A \text{ such that } \timetrace(\pi)=w \,\}.
\]
Thus, by definition, the language of a timed B\"uchi automaton contains only
non-zeno timed words.

A central algorithmic result underlying much of the theory of timed automata
is the decidability of the emptiness problem, originally established by
Alur and Dill~\cite{AlurD94} using a finite region abstraction.

\begin{theorem}[Decidability of emptiness for timed Büchi automata]
Given a timed B\"uchi automaton $\mathcal{A}$, it is decidable whether
$L(\mathcal{A}) = \emptyset$.
\end{theorem}

\section{Deterministic Timed Automata on Finite Words}

For monitoring and diagnosis we also need deterministic timed automata over
finite timed words.

\begin{definition}[Deterministic timed automaton]
A \emph{deterministic timed automaton} (DTA) over $\Sigma$ is a tuple
\[
\mathcal A=(Q,\iota,\Clocks,\delta,F)
\]
where:
\begin{itemize}
\item $Q$ is a finite set of states and $\iota \in Q$ is the initial state,
\item $\Clocks$ is a finite set of clocks,
\item $F \subseteq Q$ is the set of final states,
\item $\delta$ is a symbolic transition function
\[
\delta \colon Q \times \Sigma \to
2^{\Guards(\Clocks)\times 2^{\Clocks}\times Q}
\]
such that, for every $q \in Q$ and $\sigma \in \Sigma$, $\delta(q,\sigma)$ is
finite and, if
\[
\delta(q,\sigma)=\{(g_1,R_1,q_1),\ldots,(g_n,R_n,q_n)\},
\]
then \[g_i \wedge g_j \equiv \mathsf{false}\] for all $i \neq j$, and
\[
g_1 \vee \ldots \vee g_n \equiv \mathsf{true}.
\]
\end{itemize}
\end{definition}

Hence for every state, letter, and valuation there is a unique enabled
transition.

This induces a concrete transition function
\[
\widehat{\delta}\colon Q \times \Val(\Clocks) \times \Sigma
\to 2^{\Clocks} \times Q.
\]
in the obvious way: $\widehat{\delta}(q,\nu,\sigma)=(R,q')$ for the unique
triple $(g,R,q') \in \delta(q,\sigma)$ such that $\nu \models g$.
We extend it to finite timed words as usual:
\[
\delta^\ast \colon Q \times \Val(\Clocks) \times \TW^\ast(\Sigma)
\to Q \times \Val(\Clocks),
\]
with $\delta^\ast(q,\nu,\varepsilon)=(q,\nu)$, and for
$u \cdot (\sigma,t) \in \TW^\ast(\Sigma)$,
\[
\delta^\ast(q,\nu,u\cdot(\sigma,t))
=
(q',(\nu_u+d)[R:=0]),
\]
where $\last(u)$ denotes the timestamp of the last event of $u$,
$\delta^\ast(q,\nu,u)=(q_u,\nu_u)$, $d=t-\last(u)$,
$\last(\varepsilon)=0$, and
$(R,q')=\widehat{\delta}(q_u,\nu_u+d,\sigma)$.

The language of the DTA is
\[
L(\mathcal A)
=
\{\,u\in\TW^\ast(\Sigma)\mid \delta^\ast(\iota,\ival,u)=(q,\nu)
\text{ for some }q\in F\,\}.
\]

\section{Timed Temporal Logic}

Like in the untimed case, we use strict until and strict since as primitive operators.
Strict operators are particularly natural for timed words.
They are also close in spirit to the event-clock view of timed temporal logic,
where constraints are evaluated with respect to the last past or first future
occurrence of an event~\cite{RaskinS99,BauerLS11}.
Moving to a later position does not necessarily mean that time has advanced,
since two consecutive events may have the same timestamp.
Thus $\sFuture_{[0,0]}\psi$ says that $\psi$ holds at a distinct later event
occurring at the same time as the current event.
The non-strict operators below add the current-position case explicitly when
$0$ belongs to the interval.

\begin{definition}[Epistemic timed temporal logic]
Formulas from epistemic timed temporal logic, denoted by \TLTL, are generated by:
\[
\phi ~::=~ p \;\mid\; \neg\phi \;\mid\; \phi \wedge \phi \;\mid\; \phi \,\sUntil_I\, \phi \;\mid\; \phi \,\sSince_I\, \phi \;\mid\; \Knows_a\phi
\]
where $p \in \AP$, $a \in \Ag$, and $I \in \Intervals$.

The fragment without epistemic operators is denoted by \TLTLFP.
It still contains both future and past timed modalities.
\end{definition}

Timed formulas in this chapter are interpreted at event positions only.
Thus, the logic reasons about the timestamps of events, not about all
intermediate time instants.

\begin{definition}[Semantics]
Let $\mathcal T$ be a timed B\"uchi automaton over $\Sigma$ and let
$w=(\sigma_1,t_1)(\sigma_2,t_2)\ldots \in L(\mathcal T)$.
For $k\ge 1$, satisfaction is defined inductively as follows:
\begin{itemize}
\item $\mathcal T,w,k \models p$ if $p \in \sigma_k$
\item the Boolean cases are standard
\item $\mathcal T,w,k \models \phi \,\sUntil_I\, \psi$ if there exists $j>k$ such that
\[
\mathcal T,w,j \models \psi,\qquad
\mathcal T,w,i \models \phi \text{ for all }k<i<j,
\qquad
t_j-t_k \in I,
\]
and $j$ is the first future position satisfying $\psi$, i.e.,
\[
\mathcal T,w,i \not\models \psi \quad\text{for all } k<i<j,
\]
\item $\mathcal T,w,k \models \phi \,\sSince_I\, \psi$ if there exists $j<k$ such that
\[
\mathcal T,w,j \models \psi,\qquad
\mathcal T,w,i \models \phi \text{ for all }j<i<k,
\qquad
t_k-t_j \in I
\]
and $j$ is the last past position satisfying $\psi$, i.e.,
\[
\mathcal T,w,i \not\models \psi \quad\text{for all } j<i<k,
\]
\item $\mathcal T,w,k \models \Knows_a\phi$ if for all $w' \in L(\mathcal T)$ with
$(w,k)\sim_a (w',k)$, we have $\mathcal T,w',k \models \phi$
\end{itemize}
\end{definition}

Note that $\phi \,\sUntil_I\, \psi$ requires that the first future position
satisfying $\psi$ occurs at a delay in $I$, and $\phi \,\sSince_I\, \psi$
requires that the last past position satisfying $\psi$ lies at a delay in $I$.
This differs from the more common existential interval-witness reading, where
any suitable future or past witness in the interval would suffice.
Related strict first-time semantics have also been studied explicitly in timed
logics~\cite{Alsmann025}. Event-clock logic is likewise based on the first
future and last past occurrence of an event~\cite{RaskinS99,BauerLS11}.
We adopt this convention because it leads to simpler automata-theoretic
constructions in the sequel.

As before, we write $\mathcal T \models \phi$ if
$\mathcal T,w,1 \models \phi$ for every $w \in L(\mathcal T)$.

\paragraph{Derived operators.}
We define:
\[
\phi \,\tmp{U}_I\, \psi \;:=\;
\begin{cases}
\psi \lor (\phi \wedge \phi \,\sUntil_I\, \psi) & \text{if } 0 \in I\\
\phi \wedge \phi \,\sUntil_I\, \psi & \text{if } 0 \notin I
\end{cases}
\]
\[
\phi \,\tmp{S}_I\, \psi \;:=\;
\begin{cases}
\psi \lor (\phi \wedge \phi \,\sSince_I\, \psi) & \text{if } 0 \in I\\
\phi \wedge \phi \,\sSince_I\, \psi & \text{if } 0 \notin I
\end{cases}
\]
and therefore
\begin{align*}
\sFuture_I\phi &\;:=\; \top \,\sUntil_I\, \phi
&
\sPast_I\phi &\;:=\; \top \,\sSince_I\, \phi
\\[1ex]
\tmp{F}_I\phi &\;:=\; \top \,\tmp{U}_I\, \phi
&
\tmp{P}_I\phi &\;:=\; \top \,\tmp{S}_I\, \phi
\\[1ex]
\tmp{G}_I\phi &\;:=\; \neg\tmp{F}_I\neg\phi
&
\tmp{H}_I\phi &\;:=\; \neg\tmp{P}_I\neg\phi
\end{align*}
For the unbounded versions, we use the interval $[0,\infty)$ and write:
\begin{align*}
\phi \,\sUntil\, \psi &\;:=\; \phi \,\sUntil_{[0,\infty)}\, \psi
&
\phi \,\sSince\, \psi &\;:=\; \phi \,\sSince_{[0,\infty)}\, \psi
\\[1ex]
\phi \,\tmp{U}\, \psi &\;:=\; \phi \,\tmp{U}_{[0,\infty)}\, \psi
&
\phi \,\tmp{S}\, \psi &\;:=\; \phi \,\tmp{S}_{[0,\infty)}\, \psi
\\[1ex]
\tmp{F}\phi &\;:=\; \tmp{F}_{[0,\infty)}\phi
&
\tmp{P}\phi &\;:=\; \tmp{P}_{[0,\infty)}\phi
\\[1ex]
\tmp{G}\phi &\;:=\; \tmp{G}_{[0,\infty)}\phi
&
\tmp{H}\phi &\;:=\; \tmp{H}_{[0,\infty)}\phi.
\end{align*}

\begin{definition}[Past Fragment]
The \emph{past fragment} of \TLTLFP, denoted $\mathrm{TLTL}_{\mathsf{past}}$, consists of
formulas of the form $\tmp{G}\psi$, where $\psi$ does not use the future
operator~$\sUntil_I$.
\end{definition}

\section{Timed Model Checking}

The first important point is that timed model checking remains decidable as
long as we do not add epistemic operators.
We recall this through the same kind of construction as in the untimed case:
given a system and a formula, we build an automaton that runs on the same timed
words and records, at each event position, whether the formula is true there.
For past formulas this construction can be made deterministic.
For formulas with future operators one still obtains a timed transducer, but the
construction is less direct.

\smallskip

Let us first define the timed analogue of a transducer:

\begin{definition}[Timed transducer]
Let $\mathcal T$ be a timed B\"uchi automaton over $\Sigma$ and let $\phi$ be
a timed formula.
A timed \emph{$(\mathcal T,\phi)$-transducer} consists of a timed B\"uchi
automaton
\[
\mathcal A_\phi=(Q_\phi,\iota_\phi,\Clocks_\phi,\Trans_\phi,F_\phi)
\]
over $\Sigma$ together with an output map
$\gamma_\phi \colon Q_\phi \to \{0,1\}$ such that:
\[
L(\mathcal A_\phi)=L(\mathcal T),
\]
and for every accepting run
\[
\pi =
(q_0,\nu_0)\xrightarrow[(g_1,R_1)]{(\sigma_1,d_1)}(q_1,\nu_1)
\xrightarrow[(g_2,R_2)]{(\sigma_2,d_2)}(q_2,\nu_2)\ldots
\]
of $\mathcal A_\phi$ over a timed word $w \in L(\mathcal T)$ and every
$k \ge 1$,
\[
\gamma_\phi(q_k)=1
\quad\Longleftrightarrow\quad
\mathcal T,w,k \models \phi.
\]
\end{definition}

\begin{theorem}[Timed transducers for TLTL]
\label{thm:timed-transducers-tltl}
For every timed B\"uchi automaton $\mathcal T$ and every \TLTLFP formula
$\phi$, one can effectively construct a timed
$(\mathcal T,\phi)$-transducer.
For formulas from the past fragment
$\mathrm{TLTL}_{\mathsf{past}}$, the transducer can be chosen deterministic.
\end{theorem}

\begin{proof}[Proof sketch]
Atomic propositions and Boolean connectives are handled exactly as in the
untimed case, by copying the system automaton and decorating states with the
corresponding truth value.

The first genuinely timed case is strict since.
For every pair of formulas $\phi_1,\phi_2$ and interval $I$, we have
\[
\phi_1 \,\sSince_I\, \phi_2
\;\equiv\;
(\phi_1 \,\sSince\, \phi_2)\ \wedge\ \sPast_I\phi_2.
\]
The untimed strict-since part is handled as before, so it suffices to explain
the construction for $\sPast_I\psi$.

Assume that by induction we already have a timed $(\mathcal T,\psi)$-transducer
\[
\mathcal A_\psi=(Q_\psi,\iota_\psi,\Clocks_\psi,\Delta_\psi,F_\psi)
\]
with output map $\gamma_\psi \colon Q_\psi \to \{0,1\}$.
Let $z \notin \Clocks_\psi$ be a fresh clock.
We build a new transducer for $\sPast_I\psi$ whose states are triples
\[
(q,\mathit{past},\mathit{sat}) \in Q_\psi \times \{0,1\} \times \{0,1\}.
\]
Their intended meaning is:
\begin{itemize}
\item $q$ is the current state of the transducer for $\psi$,
\item $\mathit{past}=1$ means that a strict past occurrence of $\psi$ already exists,
\item $\mathit{sat}=1$ means that $\sPast_I\psi$ holds at the current position.
\end{itemize}

The fresh clock $z$ stores the elapsed time since the most recent position at
which $\psi$ held.
Whenever the construction enters a position satisfying $\psi$, the clock $z$ is
reset so that this position becomes the new candidate past witness for all
future positions.

For every transition
\[
(q,\sigma,g,R,q') \in \Delta_\psi
\]
and every $\mathit{past},\mathit{sat}\in\{0,1\}$, we add new transitions from
$(q,\mathit{past},\mathit{sat})$ as follows.
First define
\[
R'=
\begin{cases}
R \cup \{z\} & \text{if } \gamma_\psi(q')=1,\\
R & \text{otherwise.}
\end{cases}
\]
Thus $z$ is reset exactly when the new position satisfies $\psi$.
The next past flag is
\[
\mathit{past}'=
\begin{cases}
1 & \text{if } \mathit{past}=1 \text{ or } \gamma_\psi(q)=1,\\
0 & \text{otherwise.}
\end{cases}
\]
We then add the transition tuple
\[
\bigl((q,\mathit{past},\mathit{sat}),\sigma,g\wedge z\in I,R',
(q',\mathit{past}',\mathit{past}')\bigr)
\]
and the transition tuple
\[
\bigl((q,\mathit{past},\mathit{sat}),\sigma,g\wedge \neg(z\in I),R',
(q',\mathit{past}',0)\bigr).
\]
The guard on $z$ is tested before the possible reset of $z$ in $R'$.
Hence the satisfaction bit at the target position is computed from the old
remembered witness, while a current occurrence of $\psi$ is stored only for
future positions.

Thus, for each transition of the transducer for $\psi$, the new construction
adds a split between the case $z \in I$ and the case
$z \notin I$.
Because the most recent past witness is unique, no additional nondeterministic
guessing is needed.
This yields a timed transducer for $\sPast_I\psi$.
Moreover, if the transducer for $\psi$ is deterministic, then the new one is
deterministic as well.
Thus the construction gives deterministic transducers throughout the past
fragment.

For the future operator one reduces
\[
\phi_1 \,\sUntil_I\, \phi_2
\]
to an untimed strict-until component together with a timed eventuality
constraint using the equivalence
\[
\phi_1 \,\sUntil_I\, \phi_2
\;\equiv\;
(\phi_1 \,\sUntil\, \phi_2)\ \wedge\ \sFuture_I \phi_2.
\]
Moreover, for finite bounds one may split intervals as
\[
\sFuture_{[m,n]}\phi
\;\equiv\;
\sFuture_{[0,n]}\phi \ \wedge\ \sFuture_{[m,\infty)}\phi.
\]
The strict-future case is fundamentally harder: several pending timing
obligations may have to be guessed and verified in parallel.
Hence one can still obtain a timed transducer, but in general no deterministic one.
\end{proof}

\paragraph{Solving the timed model-checking problem.}
We now use the transducer construction to decide model checking, exactly as in
the untimed case.

\begin{theorem}[Timed model checking is decidable]
The following problem is decidable:
given a timed B\"uchi automaton $\mathcal T$ over $\Sigma$ and a \TLTLFP formula
$\phi$, do we have
\[
\mathcal T \models \phi\,?
\]
\end{theorem}

\begin{proof}
By Theorem~\ref{thm:timed-transducers-tltl}, we can construct a timed
$(\mathcal T,\phi)$-transducer
\[
\mathcal A_\phi=(Q_\phi,\iota_\phi,\Clocks_\phi,\Trans_\phi,F_\phi)
\]
with output map $\gamma_\phi$.
By correctness of the transducer, for every accepting run
\[
(q_0,\nu_0)\xrightarrow[(g_1,R_1)]{(\sigma_1,d_1)}(q_1,\nu_1)
\xrightarrow[(g_2,R_2)]{(\sigma_2,d_2)}(q_2,\nu_2)\ldots
\]
over a word $w \in L(\mathcal T)$, we have
\[
\gamma_\phi(q_1)=1
\quad\Longleftrightarrow\quad
\mathcal T,w,1 \models \phi.
\]
Thus, $\mathcal T \not\models \phi$ iff $\mathcal A_\phi$ has an accepting run
whose first reached state $q_1$ satisfies $\gamma_\phi(q_1)=0$.
This is an emptiness problem for timed B\"uchi automata, with the additional
finite-prefix constraint on the first transition.
Such a constraint can be encoded by a standard finite-state modification of the
automaton.
Since emptiness for timed B\"uchi automata is decidable, timed model checking
for \TLTLFP is decidable.
\end{proof}

\begin{notebox}[Past vs. future]
In the timed setting, past operators are technically simpler than future
operators.
The past can be summarized by finitely many clock values at the current
position, whereas future operators may generate several simultaneous
obligations that have to be checked later.
Figure~\ref{fig:timed-finite-memory-window} illustrates the kind of finite
memory that is needed in such constructions: at a given time, several earlier
values may still be relevant for the current interval constraints.
Among the values that already lie in the interval window,
only the rightmost one is relevant for satisfying an existential past constraint;
older ones are shadowed by it.
Values to the right of the interval are still too recent, but they may become
relevant later as time advances.
A naive construction would therefore keep all such values, whose number is not
bounded a priori.
The nontrivial point is to replace this unbounded memory by a finite summary.
\end{notebox}

\begin{figure}[t]
\centering
\begin{tikzpicture}[x=0.95cm,y=0.9cm,every node/.style={font=\small}]
\draw[->,thick] (0,0) -- (9.1,0) node[right] {$t$};
\foreach \x in {0,...,8} {
  \draw (\x,0.08) -- (\x,-0.08);
}
\foreach \x in {0.7,1.4,2.7,4.2} {
  \fill[ChapterAccent] (\x,0) circle (2pt);
}
\draw[very thick] (8,0.42) -- (8,-0.42);
\node[above=4pt] at (8,0.42) {current time};
\draw[fill=ChapterLight,draw=ChapterAccent] (4.4,-0.22) rectangle (6.7,0.22);
\node[above=7pt] at (5.55,0.22) {$I$};
\foreach \x in {4.8,5.35,5.9} {
  \fill[ChapterAccent] (\x,0) circle (2.2pt);
}
\fill[red!70!black] (6.45,0) circle (2.6pt);
\foreach \x in {6.95,7.25,7.55} {
  \fill[red!70!black] (\x,0) circle (2.4pt);
}
\draw[->,red!70!black,thick] (6.45,-1.28) -- (6.45,-0.12);
\node[below,red!70!black,align=center,text width=2.3cm] at (6.72,-1.5)
{rightmost\\value in $I$};
\draw[red!70!black] (6.82,-0.38) -- (7.68,-0.38);
\draw[red!70!black] (6.82,-0.32) -- (6.82,-0.44);
\draw[red!70!black] (7.68,-0.32) -- (7.68,-0.44);
\node[below=8pt,red!70!black] at (7.55,-0.2) {too recent};
\node[below=8pt,ChapterAccent,align=center,text width=3.0cm] at (4.5,0)
{older values in $I$\\are shadowed};
\end{tikzpicture}
\caption{For a last-past constraint, only the rightmost value inside the
current interval window can witness the formula now.  Older values inside the
window are shadowed.  Values to the right of the interval are too recent now,
but they must still be represented because they may enter the interval later.}
\label{fig:timed-finite-memory-window}
\end{figure}
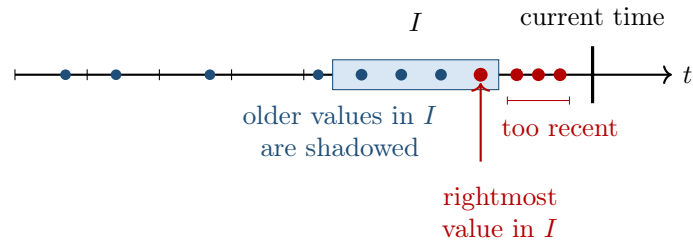

The situation changes drastically once knowledge is added.

\begin{theorem}[Timed epistemic model checking is undecidable]
\label{thm:timed-epistemic-undec}
Model checking becomes undecidable for the epistemic extension of timed
temporal logic, i.e.\ for \TLTL.
In fact, undecidability already holds for the fixed formula
\[
\tmp{G}\neg\Knows_a\tmp{P}s
\]
where $s$ is an unobservable proposition.
\end{theorem}

\begin{proof}
We reduce from the universality problem for timed automata (over finite words),
which is undecidable~\cite{AlurD94}.
Given a finite-word timed automaton $\mathcal A$ over an alphabet $\Gamma$, we ask whether
\[
L(\mathcal A)=\TW^\ast(\Gamma).
\]
Here $\mathcal A$ may be nondeterministic; we do not define this finite-word
model explicitly in these notes, since it is the standard finite-word analogue
of the timed B\"uchi automata used above.

Fix
\[
\AP=\Gamma \cup \{\#,s\},\qquad \Ag=\{a\},\qquad \AP_a=\Gamma \cup \{\#\}.
\]
Thus, the marker $\#$ is observable while the proposition $s$ is not.
We build a timed B\"uchi automaton $\mathcal T_{\mathcal A}$ with the
following behavior.

For every finite timed word
\[
u=(\gamma_1,t_1)\ldots(\gamma_m,t_m)\in\TW^\ast(\Gamma)
\]
let $T_u=\last(u)$, with $\last(\varepsilon)=0$, and let
\[
v_u=(\emptyset,T_u+1)(\emptyset,T_u+2)(\emptyset,T_u+3)\ldots
\]
be a fixed non-zeno observable continuation shifted to start after time $T_u$
and containing no marker $\#$.
We introduce a \emph{private} timed word
\[
w_{\mathsf{priv}}(u)
=
(\{s\},0)\,(\{\gamma_1\},t_1)\ldots(\{\gamma_m\},t_m)\,(\{\#\},T_u)\,v_u
\]
and a \emph{public} timed word
\[
w_{\mathsf{pub}}(u)
=
(\emptyset,0)\,(\{\gamma_1\},t_1)\ldots(\{\gamma_m\},t_m)\,(\{\#\},T_u)\,v_u.
\]
The system $\mathcal T_{\mathcal A}$ is arranged so that
\begin{itemize}
\item every private word $w_{\mathsf{priv}}(u)$ belongs to $L(\mathcal T_{\mathcal A})$,
\item the public word $w_{\mathsf{pub}}(u)$ belongs to $L(\mathcal T_{\mathcal A})$
iff $u \in L(\mathcal A)$,
\item there are no other words in $L(\mathcal T_{\mathcal A})$.
\end{itemize}

Because $s$ is unobservable, the private and public words corresponding to the
same timed word $u$ are indistinguishable to agent $a$ at every position.

\smallskip\noindent
Assume first that
\[
\mathcal T_{\mathcal A} \models \tmp{G}\neg\Knows_a\tmp{P}s.
\]
Take any $u \in \TW^\ast(\Gamma)$ and consider $w_{\mathsf{priv}}(u)$ at the
position of the marker $\#$.
At that point, $s$ has occurred earlier, so
\[
\mathcal T_{\mathcal A},w_{\mathsf{priv}}(u),k \models \tmp{P}s.
\]
From $\neg\Knows_a\tmp{P}s$, there must exist an indistinguishable word
$w' \in L(\mathcal T_{\mathcal A})$ such that
$\mathcal T_{\mathcal A},w',k \models \neg\tmp{P}s$.
By construction, $w'$ cannot be private, and the observation up to the marker
fixes the finite timed word $u$.
Hence $w'$ can only be the corresponding public word $w_{\mathsf{pub}}(u)$.
Hence $u \in L(\mathcal A)$.
Since $u$ was arbitrary, $\mathcal A$ is universal.

\smallskip\noindent
Conversely, assume $\mathcal A$ is universal.
Let $w \in L(\mathcal T_{\mathcal A})$ and $k \ge 1$.
If $w$ is public, then $s$ has not occurred, hence
$\mathcal T_{\mathcal A},w,k \models \neg\tmp{P}s$.
Since $w$ is indistinguishable from itself, we get
$\mathcal T_{\mathcal A},w,k \models \neg\Knows_a\tmp{P}s$.
If $w=w_{\mathsf{priv}}(u)$ is private, universality yields
$u \in L(\mathcal A)$, so the matching public word
$w_{\mathsf{pub}}(u)$ also belongs to $L(\mathcal T_{\mathcal A})$.
Since $w_{\mathsf{pub}}(u)$ is indistinguishable from $w$ for agent $a$ and
does not satisfy $\tmp{P}s$, the observer cannot know that $s$ occurred.
Thus
\[
\mathcal T_{\mathcal A},w,k \models \neg\Knows_a\tmp{P}s.
\]
Hence
\[
\mathcal T_{\mathcal A} \models \tmp{G}\neg\Knows_a\tmp{P}s.
\]
This concludes the proof.
\end{proof}

\section{Timed Diagnosability and Timed Opacity}

We next revisit diagnosability and opacity in the timed setting.
The difference from the untimed case is that delay bounds are now measured in
time units rather than in discrete steps.

\begin{definition}[Timed diagnosability]
Let $\mathcal T$ be a timed B\"uchi automaton over $\Sigma$, let $a \in \Ag$,
let $e \in \AP \setminus \AP_a$, and let $D \in \mathbb N$.
We say that $\mathcal T$ is \emph{timed-$D$-P-diagnosable} (wrt.\ $a$ and $e$) if
\[
\mathcal T \models \tmp{G}\bigl(e \to \tmp{F}_{[0,D]}\Knows_a\tmp{P}e\bigr).
\]
Thus, whenever an error occurs, the observer must know within at most $D$ time
units that an error has occurred in the past.
\end{definition}

\begin{theorem}[Characterization of timed bounded P-diagnosability]
The following are equivalent:
\begin{enumerate}
\item $\mathcal T$ is not timed-$D$-P-diagnosable,
\item there exist timed words
\[
w=(\sigma_1,t_1)(\sigma_2,t_2)\ldots\quad
\text{and}\quad
w'=(\sigma'_1,t'_1)(\sigma'_2,t'_2)\ldots
\]
in $L(\mathcal T)$ and a position $k \ge 1$ such that
\begin{itemize}
\item $\mathcal T,w,k \models e$,
\item for every position $k' \ge k$ with $t_{k'}-t_k \le D$
\begin{itemize}
\item the prefix $w'_{\le k'}$ is fault-free and
\item $(w,k') \sim_a (w',k')$.
\end{itemize}
\end{itemize}
\end{enumerate}
\end{theorem}

\begin{proof}[Idea]
This is the timed analogue of the untimed twin-plant characterization.
If knowledge of the error is missing within the allowed time window, then there
must be an indistinguishable fault-free companion run for the whole duration of
that window.
Because timed words are non-zeno, there are only finitely many positions
$k' \ge k$ with $t_{k'}-t_k \le D$; hence there is a last one.
An indistinguishable fault-free companion for that last position also works for
all earlier positions in the window, by prefix closure of fault-freeness and of
indistinguishability.
Conversely, such a companion run directly witnesses failure of the formula
$\tmp{G}(e \to \tmp{F}_{[0,D]}\Knows_a\tmp{P}e)$.
\end{proof}

\begin{theorem}[Timed Bounded P-diagnosability is decidable]
Given $\mathcal T$, $a$, $e$, and $D$, one can decide whether
$\mathcal T$ is timed-$D$-P-diagnosable.
\end{theorem}

\begin{proof}[Proof sketch]
One uses a twin-plant construction.
Two runs $w$ and $w'$ are simulated in parallel under the constraint that they
remain indistinguishable to the observer.
The first run must contain an error.
A fresh clock $z$ is reset at the first such error and measures the time
elapsed since that occurrence.
The second run is forced to remain fault-free while $z \le D$.
See Figure~\ref{fig:timed-diagnosability-window} for the intuition.

Non-diagnosability is therefore equivalent to the existence of an accepting run
of this product automaton.
Since emptiness for timed B\"uchi automata is decidable, timed diagnosability
is decidable as well.
\end{proof}

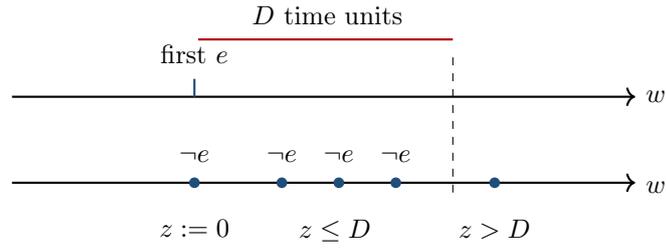
\begin{figure}[t]
\centering
\begin{tikzpicture}[x=1.0cm,y=0.95cm,every node/.style={font=\small}]
\draw[->,thick] (0,1.2) -- (8.2,1.2) node[right] {$w$};
\draw[->,thick] (0,0) -- (8.2,0) node[right] {$w'$};
\draw[ChapterAccent,thick] (2.4,1.2) -- (2.4,1.45);
\node[above=2pt] at (2.4,1.45) {first $e$};
\fill[ChapterAccent] (2.4,0) circle (2pt);
\fill[ChapterAccent] (3.55,0) circle (2pt);
\fill[ChapterAccent] (4.3,0) circle (2pt);
\fill[ChapterAccent] (5.05,0) circle (2pt);
\fill[ChapterAccent] (6.35,0) circle (2pt);
\node[above=4pt] at (2.4,0) {$\neg e$};
\node[above=4pt] at (3.55,0) {$\neg e$};
\node[above=4pt] at (4.3,0) {$\neg e$};
\node[above=4pt] at (5.05,0) {$\neg e$};
\node[below=10pt] at (2.4,0) {$z:=0$};
\node[below=10pt] at (4.25,0) {$z\le D$};
\node[below=10pt] at (6.35,0) {$z>D$};
\draw[red!70!black,thick] (2.45,2.0) -- (5.8,2.0);
\node[above=2pt] at (4.15,2.0) {$D$ time units};
\draw[dashed] (5.8,1.75) -- (5.8,-0.2);
\end{tikzpicture}
\caption{Twin-plant view of non-diagnosability: after the first error on the
upper run, the lower run stays fault-free yet observationally indistinguishable
for the whole time window of length $D$.}
\label{fig:timed-diagnosability-window}
\end{figure}

\paragraph{Timed DTA-diagnosers.}
We also consider the operational notion of a diagnoser given as a
DTA over the observation alphabet.
This is the automaton that reads the observer's timed information and is
supposed to switch to an alarm state once enough evidence for the fault has
been gathered.

\begin{definition}[Timed DTA-diagnoser]
Let $\mathcal T$ be a timed system, let $a \in \Ag$, let
$e \in \AP \setminus \AP_a$, and let $D \in \mathbb N$.
A DTA $\mathcal D$ over $\Sigma_a$ is a \emph{timed-$D$-diagnoser} if for
every timed word $w=(\sigma_1,t_1)(\sigma_2,t_2)\ldots \in L(\mathcal T)$ and every $k \ge 1$:
\begin{itemize}
\item if $\mathcal T,w,k \models \neg\tmp{P}e$, then
\[
\obs_a(w_{\le k}) \notin L(\mathcal D),
\]
\item if $\mathcal T,w,k \models e$, then there exists $k' \ge k$ such that
\[
t_{k'}-t_k \le D
\]
and for all $k'' \ge k'$,
\[
\obs_a(w_{\le k''}) \in L(\mathcal D).
\]
\end{itemize}
\end{definition}

This is the natural timed analogue of an untimed diagnoser:
no false positives are allowed, and after an error the diagnoser must switch to
acceptance within the allowed delay and then stay there.

\begin{notebox}[Diagnosability and event-driven diagnosers]
Both the logical notion of timed diagnosability, given by the formula
\[\tmp{G}\bigl(e \to \tmp{F}_{[0,D]}\Knows_a\tmp{P}e\bigr),\]
and the DTA-based diagnoser
model considered here are event-based.
However, unlike in the untimed setting, it is a separate question whether
timed-$D$-P-diagnosability yields a concrete timed-$D$-diagnoser in this
operational model.
Independently, one may also ask whether an event-driven diagnoser model is
suitable at all in the timed setting.
The example below is meant to illustrate this latter issue.
\end{notebox}

\begin{example}[A limitation of event-driven timed diagnosers]
Assume the observer sees all propositions except the error proposition $e$.
Suppose the system has only two possible timed behaviors:
\begin{itemize}
\item either an event labeled $\{e\}$ occurs at some time $t<1$,
\item or no error occurs and the first event observable to the diagnoser occurs only after time $100$.
\end{itemize}
This illustrates why timed diagnosis is more subtle than in the untimed setting.
On the faulty branch, the diagnoser receives an input event before time $1$ and can in
principle react to it immediately.
On the non-faulty branch, however, no input event is received before time $100$.
Hence, if one wants the diagnoser to conclude already at time $1$ that no fault has
occurred so far, this conclusion would have to be based solely on the absence of events
during the interval $[0,1]$.
But the event-driven DTA model used here changes state only when a new observable
timed event is read.
So this example should be understood as motivation: in the timed setting, a concrete
real-time diagnoser may need timeout transitions or some other mechanism that can react
to the passage of time itself, and not only to newly read events.
\end{example}

We now turn to opacity.
As in the untimed setting, this is a knowledge-based property, but here
the relevant indistinguishability relation is the timed one introduced above.

\begin{definition}[Timed opacity]
Let $\mathcal T$ be a timed B\"uchi automaton over $\Sigma$, let $a \in \Ag$,
and let $s \in \AP \setminus \AP_a$.
We say that $\mathcal T$ is \emph{timed opaque} if
\[
\mathcal T \models \tmp{G}\neg\Knows_a\tmp{P}s.
\]
\end{definition}

\begin{theorem}
Timed opacity is undecidable.
\end{theorem}

\begin{proof}
This is immediate from the proof of Theorem~\ref{thm:timed-epistemic-undec},
since the fixed formula
$\tmp{G}\neg\Knows_a\tmp{P}s$ used in the reduction is already the opacity
formula.
\end{proof}

\section{Timed Monitorability}

For monitorability we return to the classical runtime-verification setting:
\[
\Ag=\{a\},\qquad \AP_a=\AP,\qquad L(\mathcal T)=\TW^\omega(\Sigma).
\]
Thus there is a single fully informed observer, and the system is the
most-permissive timed system.
So we only consider the classical one-observer case here.
We also restrict attention to timed temporal formulas without epistemic
operators.
Accordingly, we write $w,k \models \phi$ for satisfaction in this fixed
most-permissive timed model.
Later we briefly contrast this with monitorability questions for timed
languages given by automata.
Even in this restricted timed setting, the monitorability question is already
subtle enough.
This exhibits the main difficulties caused by real-time constraints.
For this reason, we do not discuss more general variants here, as we did in the
untimed setting.

For a \TLTLFP formula $\phi$, we write
\[
L(\phi)=\{\,w\in\TW^\omega(\Sigma)\mid w,1\models \phi\,\}
\]
in this fixed most-permissive timed model.

\begin{definition}[Timed monitorability]
A formula $\phi \in \mathrm{TLTL}$ is \emph{timed-monitorable} if for every finite timed
word $u \in \TW^+(\Sigma)$ there exists a finite timed word
$v \in \TW^+(\Sigma)$ such that $u$ is a prefix of $v$ and one of the following holds:
\begin{itemize}
\item for every $w \in \TW^\omega(\Sigma)$, if $v$ is a prefix of $w$, then $w \in L(\phi)$;
\item for every $w \in \TW^\omega(\Sigma)$, if $v$ is a prefix of $w$, then $w \in L(\neg\phi)$.
\end{itemize}
\end{definition}

In other words, every finite observation should admit a finite continuation that
forces a definitive positive or negative verdict.

\begin{theorem}[Past timed formulas are monitorable]
Let $\phi \in \mathrm{TLTL}_{\mathsf{past}}$, that is,
\[
\phi = \tmp{G}\psi
\]
for some \TLTLFP formula $\psi$ that does not use future operators.
Then $\phi$ is timed-monitorable.
\end{theorem}

\begin{proof}
By the past-fragment construction from the proof of
Theorem~\ref{thm:timed-transducers-tltl}, one effectively obtains a
deterministic timed automaton
\[
\mathcal A_\psi=(Q,\iota,\Clocks,\delta,F)
\]
over $\Sigma$ such that for every infinite timed word
$w \in \TW^\omega(\Sigma)$ and every $k \ge 1$,
\[
w_{\le k} \in L(\mathcal A_\psi)
\quad\Longleftrightarrow\quad
w,k \models \psi.
\]

Thus, $F$ exactly represents the states where the local past property $\psi$
holds at the current position.
To monitor $\phi=\tmp{G}\psi$, we must distinguish three situations:
\begin{itemize}
\item $\bot$: the property has already been violated,
\item $?$: the property holds so far, but could still be violated later,
\item $\top$: the property holds so far and can never be violated later.
\end{itemize}

The key auxiliary ingredient is the following reachability guard.
For every state $q \in Q$ and every reset set $R \subseteq \Clocks$, one can
effectively compute a guard
\[
g_q^R \in \Guards(\Clocks)
\]
such that for every valuation $\nu \in \Val(\Clocks)$:
\[
\nu \models g_q^R
\quad\Longleftrightarrow\quad
(q,\nu[R:=0]) \text{ admits a finite timed continuation into some state }
q' \notin F.
\]
Intuitively, $g_q^R$ says that after applying the reset set $R$, it is still
possible to reach a future position where $\psi$ is false.
Such a guard is effectively computable by the region construction, since the
set of valuations with this reachability property is a finite union of regions
and is therefore definable by a guard.

\smallskip
We now construct a timed monitor
\[
\mathcal M_\phi = (Q \times \{\bot, ?, \top\},(\iota, ?),\Clocks,\delta_\mathcal M)
\]
whose second component is the output value.
It is best viewed as a timed Moore machine: there is no accepting set, only the
current verdict.

Fix $q \in Q$ and $\sigma \in \Sigma$, and write
\[
\delta(q,\sigma)=\{(g_i,R_i,q_i)\mid i=1,\dots,m\}.
\]
We define $\delta_\mathcal M$ as follows.

\paragraph{Mode $\bot$.}
Once the property is violated, it remains violated:
\[
\delta_\mathcal M((q,\bot),\sigma)
\ni
(g_i,R_i,(q_i,\bot))
\qquad\text{for all }i.
\]

\paragraph{Mode $\top$.}
Once the property is guaranteed forever, it remains guaranteed:
\[
\delta_\mathcal M((q,\top),\sigma)
\ni
(g_i,R_i,(q_i,\top))
\qquad\text{for all }i.
\]

\paragraph{Mode $?$.
}
Here the current run has not yet been decided.
\begin{itemize}
\item If $q_i \notin F$, then $\psi$ is false at the current position, hence
$\tmp{G}\psi$ is already violated:
\[
\delta_\mathcal M((q,?),\sigma)
\ni
(g_i,R_i,(q_i,\bot)).
\]

\item If $q_i \in F$, then $\psi$ holds at the current position.
There are now two subcases:
\begin{align*}
\delta_\mathcal M((q,?),\sigma)
&\ni
(g_i \wedge g_{q_i}^{R_i},R_i,(q_i,?)),
\\[1ex]
\delta_\mathcal M((q,?),\sigma)
&\ni
(g_i \wedge \neg g_{q_i}^{R_i},R_i,(q_i,\top)).
\end{align*}
The first transition means: $\psi$ currently holds, but a future violation is
still reachable, so the verdict stays undecided.
The second means: $\psi$ currently holds and no future violation is reachable,
so the verdict becomes permanently true.
\end{itemize}
Transitions from mode $?$ are illustrated in Figure~\ref{fig:timed-monitor-modes}.

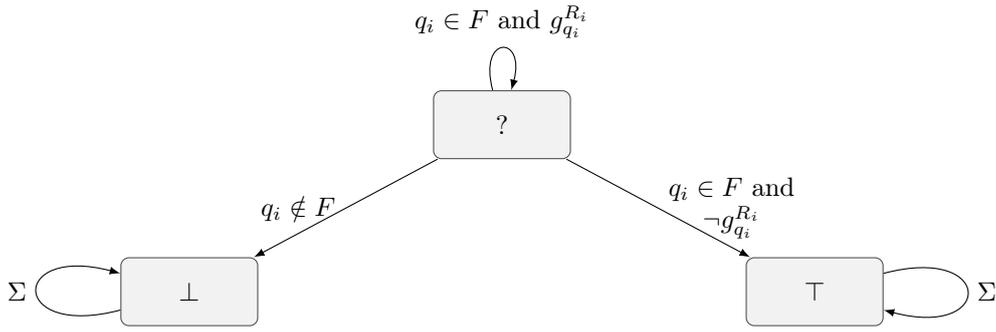
\begin{figure}[t]
\centering
\begin{tikzpicture}[
  ->, >=latex, node distance=3cm,
  every node/.style={font=\small},
  verdict/.style={draw=black!70, rounded corners=3pt, minimum width=1.8cm, minimum height=0.9cm, fill=black!5}
]
\node[verdict] (unk) {$?$};
\node[verdict] (bot) [below left=1.3cm and 2.3cm of unk] {$\bot$};
\node[verdict] (top) [below right=1.3cm and 2.3cm of unk] {$\top$};
\draw[->] (unk) -- node[left, align=center] {$q_i \notin F$} (bot);
\draw[->] (unk) -- node[right, align=center] {$q_i \in F$ and\\ $\neg g_{q_i}^{R_i}$} (top);
\draw[loop above] (unk) to node {$q_i \in F$ and $g_{q_i}^{R_i}$} (unk);
\draw[loop left] (bot) to node {$\Sigma$} (bot);
\draw[loop right] (top) to node {$\Sigma$} (top);
\end{tikzpicture}
\caption{Three-valued timed monitor for $\tmp{G}\psi$.
Mode $?$ means ``true so far, but not yet guaranteed forever''.
Mode $\top$ means ``permanently safe'', and mode $\bot$ means ``already violated''.}
\label{fig:timed-monitor-modes}
\end{figure}

The guards remain complete and disjoint, so $\mathcal M_\phi$ is again a
deterministic timed automaton.
Its correctness is immediate from the meaning of the three verdicts and of the
guards $g_q^R$:
\begin{itemize}
\item entering mode $\bot$ means that some position already violates $\psi$,
\item staying in mode $?$ means that all positions so far satisfy $\psi$, but a
future bad position is still possible,
\item entering mode $\top$ means that every continuation will satisfy $\psi$ at
all future positions.
\end{itemize}

Now let $u \in \TW^+(\Sigma)$ be any finite timed word, and let
$(q,\theta)$ be the state reached by $\mathcal M_\phi$ on $u$.
If $\theta=\bot$, then every infinite timed word extending $u$ lies in
$L(\neg\phi)$.
If $\theta=\top$, then every infinite timed word extending $u$ lies in
$L(\phi)$.
Finally, if $\theta={?}$, then the last transition entering mode ${?}$ must
have been of the form
\[
(g_i \wedge g_{q_i}^{R_i}, R_i, (q_i, ?))
\]
with $q_i \in F$.
By the definition of $g_{q_i}^{R_i}$, after applying the reset set $R_i$
there exists a finite timed continuation leading to some state outside $F$.
Following such a continuation from $u$ brings $\mathcal M_\phi$ to mode
$\bot$, and from then on every continuation satisfies $\neg\phi$.
Thus in all cases, $u$ admits a finite timed extension that forces either
$\phi$ or $\neg\phi$.
Hence the formula $\phi$ is timed-monitorable.
\end{proof}

\begin{notebox}[Status of full TLTL]
Recall that we covered only the past fragment.
The exact decidability status of monitorability for full \TLTLFP is left
open.
\end{notebox}

There is a separate automaton-theoretic question: given a timed automaton
interpreted as a property, is its language monitorable?
Here we view a timed automaton simply as a language acceptor over a finite
alphabet $\Gamma$, without introducing an additional proposition structure.
This problem is undecidable.

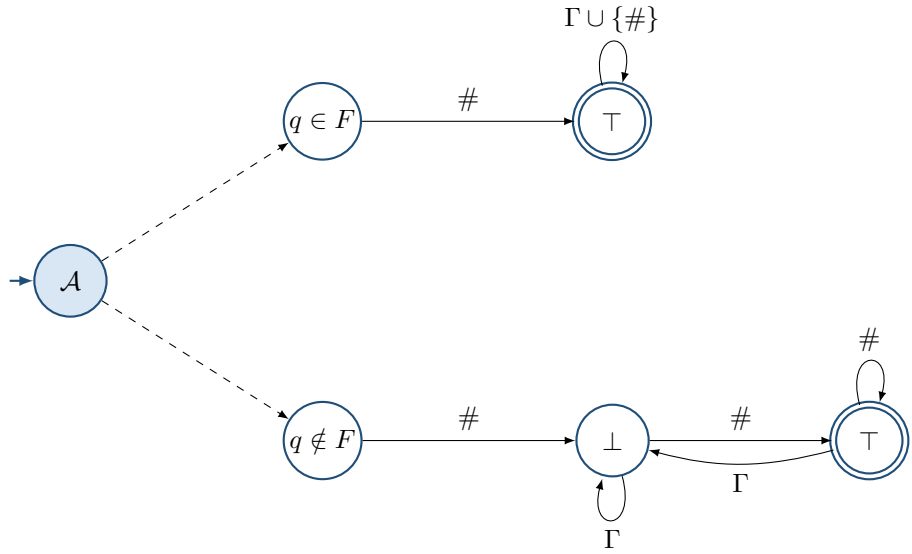
\begin{figure}[t]
\centering
\begin{tikzpicture}[
  ->, >=latex, node distance=2.6cm,
  every node/.style={font=\small},
  state/.style={circle, draw=ChapterAccent, thick, minimum size=0.95cm, inner sep=1pt},
  initial/.style={state, fill=ChapterLight},
  sink/.style={state, double, double distance=1.4pt}
]
\node[initial] (sim) {$\mathcal A$};
\node[state] (acc) [above right=1.4cm and 2.6cm of sim] {$q\in F$};
\node[state] (rej) [below right=1.4cm and 2.6cm of sim] {$q\notin F$};
\node[sink] (top) [right=2.8cm of acc] {$\top$};
\node[state] (bot) [right=2.8cm of rej] {$\bot$};
\node[sink] (botacc) [right=2.4cm of bot] {$\top$};

\draw[thick, ChapterAccent, <-] (sim) -- ++(-0.8,0);
\path
(sim) edge[dashed] (acc)
      edge[dashed] (rej)
(acc) edge node[above] {$\#$} (top)
(rej) edge node[above] {$\#$} (bot)
(top) edge[loop above] node {$\Gamma \cup \{\#\}$} (top)
(bot) edge[loop below] node {$\Gamma$} (bot)
      edge node[above] {$\#$} (botacc)
(botacc) edge[loop above] node {$\#$} (botacc)
         edge[bend left=15] node[below] {$\Gamma$} (bot);
\end{tikzpicture}
\caption{Shape of the reduction for monitorability of timed-automaton
languages. After the marker $\#$, the construction exposes whether the
simulated run of $\mathcal A$ was in an accepting state. In the lower branch,
the additional rightmost state ensures that no finite continuation yields a
definitive verdict.}
\label{fig:timed-monitorability-reduction}
\end{figure}
\begin{theorem}[Monitorability of timed automaton properties]
Given a timed B\"uchi automaton $\mathcal A$, it is undecidable whether its language is
monitorable.
\end{theorem}

\begin{proof}[Proof sketch]
The reduction is again from universality of timed automata over finite timed words.
Starting from a \emph{complete} nondeterministic timed automaton over an alphabet $\Gamma$, add a fresh
symbol $\#$.
The construction, which is illustrated in
Figure~\ref{fig:timed-monitorability-reduction}, first simulates $\mathcal A$ on $\Gamma$.
When a marker $\#$ is read in an accepting simulated state, it moves to a
permanently accepting sink.
When a marker is read in a nonaccepting simulated state, it moves to the lower
component: there a
$\Gamma$-transition stays rejecting, while a $\#$-transition moves to an
accepting state, from which every further $\#$ keeps the word accepted and a
letter from $\Gamma$ leads back to the rejecting state.

If $\mathcal A$ is universal, the resulting language is trivially monitorable.
If $\mathcal A$ is not universal, choose a timed word $u \notin L(\mathcal A)$.
After reading the prefix $u\#$, the automaton is in the lower rejecting state.
From there, no finite continuation can force a conclusive verdict:
continuations over $\Gamma$ stay rejecting, whereas a continuation with a
marker $\#$ moves to the lower accepting state.
Moreover, even after reaching that accepting state, a later letter from
$\Gamma$ returns to the rejecting state.
Hence after every finite extension of $u\#$, both acceptance and rejection are
still possible.
Hence monitorability fails.
\end{proof}

\section{Bibliographic Notes}

Timed automata were introduced by Alur and Dill and are the standard model
for finite-state real-time systems \cite{AlurD94}.
For a detailed and accessible presentation of timed automata, including the
undecidability proof for the universality problem used above, we refer the
reader to Bouyer's lecture notes~\cite{Bouyer2026}.
The event-clock perspective underlying the first/last-witness intuition goes
back to Raskin and Schobbens~\cite{RaskinS99}. A related strict first-time
semantics was studied recently by Alsmann and Lange~\cite{Alsmann025}.
Related automata-theoretic constructions for future timed modalities were given
by Akshay, Gastin, Govind, and Srivathsan~\cite{AGGS2004}.
Timed diagnosis for timed automata has been studied extensively; for a survey
of the main results and variants, see Cassez and Tripakis~\cite{CassezT13}.
Timed opacity was studied by Cassez~\cite{Cassez09}, who already showed
undecidability for a very restricted class of timed automata.
Monitor constructions for timed systems were treated thoroughly by Bauer,
Leucker, and Schallhart~\cite{BauerLS11}.
For monitoring algorithms for past-only real-time logics, and for a comparison
between point-based and interval-based semantics, see Basin, Klaedtke, and
Zalinescu~\cite{BasinKZ18}.
Monitorability for timed-automaton languages was studied by Grosen, Kauffman,
Larsen, and Zimmermann~\cite{GrosenKL025}; in particular, the undecidability
proof sketched above follows that construction.
Overall, the timed setting is more delicate than the untimed one:
many central timed verification problems without knowledge remain decidable,
while the combination with epistemic reasoning quickly reaches the boundary of
undecidability.

\backmatter
\clearpage
\bibliographystyle{alpha}
\bibliography{lit}

\end{document}